\documentclass[a4paper, 11pt]{article}

\usepackage[backref=page]{hyperref}  %
\usepackage{geometry}  %
\usepackage{amssymb}  %
\usepackage{amsthm}
\usepackage{mathtools}  %
\usepackage{isomath}  %
\usepackage{xfrac}  %
\usepackage{xspace}  %
\usepackage{float}  %
\usepackage{graphicx}  %
\usepackage{subcaption}  %
\usepackage{natbib}  %
\usepackage[capitalise]{cleveref}  %
\usepackage{appendix}  %
\usepackage{enumitem}  %
\usepackage{multicol}  %
\usepackage[noend]{algpseudocode}  %
\usepackage{algorithm}  %
\usepackage{acro}  %
\usepackage[noblocks,affil-sl]{authblk}  %
\usepackage{pdflscape}  %
\usepackage{bm} %
\usepackage{tabularx}  %
\usepackage{etoolbox}  %
\usepackage[section,below]{placeins}  %

\renewcommand*{\backrefalt}[4]{%
 \ifcase #1 %
 \or
 [Page #2]%
 \else
 [Pages #2]%
 \fi
}

\hypersetup{colorlinks=true, allcolors=blue}

\setcitestyle{authoryear}

\algrenewcommand\alglinenumber[1]{\scriptsize #1}
\algrenewcommand{\algorithmicrequire}{\textbf{Inputs:}}
\algrenewcommand{\algorithmicensure}{\textbf{Outputs:}}
\algrenewcommand\algorithmicthen{}
\algrenewcommand\algorithmicdo{}
\algnewcommand{\Raise}{\textbf{throw}\xspace}
\algblockdefx[Try]{Try}{EndTry}{\textbf{try}}{}
\algcblockdefx[Try]{Try}{Catch}{EndTry}[1][]{\textbf{catch} #1}{}
\algtext*{EndTry}
\algnewcommand{\IfThenElse}[3]{%
  #2\ \algorithmicif\ #1\ \algorithmicelse\ #3}
\algrenewcommand{\algorithmiccomment}[1]{\hfill \texttt{\#}~\textit{#1}}
\renewcommand{\gets}{=}

\newtheorem{definition}{Definition}[section]
\newtheorem{proposition}{Proposition}[section]
\newtheorem{remark}{Remark}[section]
\newtheorem{corollary}{Corollary}[section]
\newtheorem{lemma}{Lemma}[section]

\newtheorem{assumption}{Assumption}[section]
\crefname{assumption}{Assumption}{Assumptions}

\floatstyle{ruled}
\newfloat{model}{thp}{lom}
\floatname{model}{Model}
\crefname{model}{Model}{Models}

\DeclareSymbolFont{bbold}{U}{bbold}{m}{n}
\DeclareSymbolFontAlphabet{\mathbbold}{bbold}

\newcommand{\reals}{\mathbb{R}}

\newcommand{\nonnegreals}{\mathbb{R}_{\geq 0}}
\newcommand{\integers}{\mathbb{Z}}

\newcommand{\set}[1]{\mathcal{#1}}
\newcommand{\vct}[1]{\boldsymbol{#1}}
\newcommand{\vctfunc}[1]{\boldsymbol{#1}}
\newcommand{\opfunc}[1]{#1}
\newcommand{\mtx}[1]{#1}
\newcommand{\idmtx}{\mathbb{I}}
\newcommand{\intg}[1]{\mathtt{#1}}
\newcommand{\rvar}[1]{\mathsf{#1}}
\newcommand{\rvct}[1]{\vct{\rvar{#1}}}

\newcommand{\prob}{\mathbb{P}}
\newcommand{\nrm}{\mathcal{N}}
\newcommand{\gvn}{\,|\,}
\newcommand{\range}[2][1]{#1{:}\mkern0.5mu#2}
\newcommand{\tr}{^{\mkern-1.5mu\mathsf{T}}}
\newcommand{\dr}{\mathrm{d}}
\newcommand{\diff}{\mathop{}\!\dr}
\newcommand{\jacob}{\mathop{}\!\partial}
\newcommand{\grad}{\mathop{}\!\nabla}
\newcommand{\hess}{\mathop{}\!\nabla^2}
\newcommand{\td}[2]{\frac{\dr #1}{\dr #2}}
\newcommand{\inlinetd}[2]{\sfrac{\dr #1}{\dr #2}}

\newcommand{\ind}[1]{\mathbbold{1}_{#1}}
\renewcommand{\det}[1]{\left| #1 \right|}
\newcommand{\lebm}[1]{{\lambda}_{#1}}
\newcommand{\haum}[2]{{\eta}\mkern1mu^{#1}_{#2}}
\newcommand{\riem}[2]{{\sigma}\mkern1mu^{#1}_{#2}}

\newcommand{\borelprob}{\mathfrak{P}}
\newcommand{\linops}{\mathfrak{L}}
\newcommand{\bigo}{\mathcal{O}}
\newcommand{\timestep}{t}
\newcommand{\vcthbar}{\mathchar'26\mkern-9mu \boldsymbol{h}}
\newcommand{\tangent}{\mathsf{T}}

\DeclareMathOperator{\diag}{diag}

\DeclareMathOperator{\chol}{chol}

\DeclarePairedDelimiter\ceil{\lceil}{\rceil}
\DeclarePairedDelimiter\floor{\lfloor}{\rfloor}

\newcommand{\shorturl}[1]{\href{https://#1}{\nolinkurl{#1}}}

\newcounter{magicrownumbers}
\newcommand\rownumber{\refstepcounter{magicrownumbers}{\scriptsize\arabic{magicrownumbers}}}
\AtBeginEnvironment{tabularx}{\setcounter{magicrownumbers}{0}}
\crefname{magicrownumbers}{line}{lines}
\Crefname{magicrownumbers}{Line}{Lines}

\newif\ifseparatesupplement
\separatesupplementfalse

\newcommand{\appcref}[1]{%
\ifseparatesupplement%
  \cref*{#1} in the Supplementary Material%
\else%
  \cref{#1}%
\fi
}

\makeatletter
\def\ifclass#1#2#3{\@ifundefined{opt@#1.cls}{#3}{#2}}
\makeatother

\acsetup{long-format=\it, short-format=\sc}
\DeclareAcronym{sde}{short=sde, long=stochastic differential equation}
\DeclareAcronym{ode}{short=ode, long=ordinary differential equation}
\DeclareAcronym{dae}{short=dae, long=differential algebraic equation}
\DeclareAcronym{mcmc}{short=mcmc, long=Markov chain Monte Carlo}
\DeclareAcronym{hmc}{short=hmc, long=Hamiltonian Monte Carlo}
\DeclareAcronym{nuts}{short=nuts, long=No U Turn Sampler}
\DeclareAcronym{abc}{short=abc, long=approximate Bayesian computation}
\DeclareAcronym{ess}{short=ess, long=effective sample size}
\DeclareAcronym{ad}{short=ad, long=algorithmic differentation}
\DeclareAcronym{svd}{short=svd, long=singular value decomposition}
\DeclareAcronym{rwm}{short=rwm, long=random-walk Metropolis}
\DeclareAcronym{sir}{short=sir, long=susceptible-infected-recovered} 
\newcommand{\fullTitle}{%
  Manifold Markov chain Monte Carlo methods for %
  Bayesian inference in diffusion models%
}

\newcommand{\mgFullName}{Matthew M. Graham\xspace}

\newcommand{\mgEmail}{m.graham@ucl.ac.uk\xspace}

\newcommand{\atFullName}{Alexandre H. Thiery\xspace}

\newcommand{\atEmail}{a.h.thiery@nus.edu.sg\xspace}

\newcommand{\abFullName}{Alexandros Beskos\xspace}

\newcommand{\abEmail}{a.beskos@ucl.ac.uk\xspace}

\newcommand{\nusAddress}{National University of Singapore, Singapore}
\newcommand{\uclAddress}{University College London, London, UK}
 
\title{\fullTitle}
\author[1]{\mgFullName (\texttt{\mgEmail})}
\affil[1]{\uclAddress}
\author[2]{\atFullName (\texttt{\atEmail})}
\affil[2]{\nusAddress}
\author[1]{\abFullName (\texttt{\abEmail})}

\date{\vspace{-6ex}}

\begin{document}

\maketitle

\begin{abstract}
Bayesian inference for nonlinear diffusions, observed at discrete times, is a challenging task that has prompted the development of a number of algorithms, mainly within the computational statistics community. We propose a new direction, and accompanying methodology --- borrowing ideas from statistical physics and computational chemistry ---  for inferring the posterior distribution of latent diffusion paths and model parameters, given observations of the process. Joint configurations of the underlying process noise and of parameters, mapping onto diffusion paths consistent with observations, form an implicitly defined manifold. Then, by making use of a constrained Hamiltonian Monte Carlo algorithm on the embedded manifold, we are able to perform computationally efficient inference for a class of discretely observed diffusion models. Critically,  in contrast with other approaches proposed in the literature, our methodology is \emph{highly automated}, requiring minimal user intervention and applying \emph{alike} in a range of settings, including: elliptic or hypo-elliptic systems; observations with or without noise; linear or non-linear observation operators. Exploiting Markovianity, we propose a variant of the method with complexity that scales linearly in the resolution of path discretisation and the number of observation times. Python code reproducing the results is available at \shorturl{doi.org/10.5281/zenodo.5796148}. %
\\[3pt]
\end{abstract}

\section{Introduction}
\label{sec:introduction}

A large number of stochastic dynamical systems are modelled via the use of \mbox{diffusion} processes, see e.g.~\citet{kloeden1992numerical, oksendal2013stochastic} and the references therein. An enormous amount of research has been dedicated to both the theoretical foundations of such processes and --- as with this work --- their statistical calibration. Our work lies in the context of processes observed discretely in time, under a low frequency regime, so that approximations of typically analytically intractable transition densities are assumed to be inaccurate. In this setting, \emph{data augmentation} approaches within a Bayesian framework have delivered the prevailing methodologies, see e.g.~\citet{sorensen2009parametric,papaspiliopoulos2013data}, as they  provide various model-specific algorithms for treating a number of different specifications of the structure of the diffusion process and of the observation regime. The performance of the developed algorithms can be improved via a combination of model transforms, often motivated by the \emph{Roberts--Stramer critique} \citep{roberts2001inference} --- that the posterior distribution of the diffusivity parameters given a time discretisation of the process degenerates at finer resolutions --- and more efficient \ac{mcmc} kernels.

The work herein provides a natural approach for Bayesian inference over diffusion processes. Observations are treated as \emph{constraints} placed on latent paths and parameters. This gives rise to the viewpoint that the posterior can be expressed as the prior distribution restricted to a \emph{manifold}. We apply existing \ac{mcmc} methods for sampling from distributions supported on submanifolds based on the simulation of constrained Hamiltonian dynamics (see, e.g., \citet{hartmann2005constrained,rousset2010free,brubaker2012family, lelievre2019hybrid}) to efficiently explore this manifold-supported posterior distribution. This class of methods relies on symplectic integrators for constrained Hamiltonian systems \citep{andersen1983rattle,leimkuhler1994symplectic,reich1996symplectic,leimkuhler2016efficient}. Critically, we leverage the Markovian structure of the diffusion process and Gaussianity of the driving noise to design a scalable inferential procedure. The main contributions of the proposed methodology can be summarised as follows:
\begin{enumerate}[label=(\roman*), noitemsep]
\item  We provide a new viewpoint and accompanying algorithmic methodology for calibrating \ac{sde} models. The posterior is expressed as a distribution supported on a manifold embedded in a non-centred parametrization of the latent path and parameter space. We then make use of a constrained \ac{hmc} scheme to explore this manifold, jointly updating both the parameters and latent path.
\item Unlike other algorithms that are often limited to specific model families, our approach is highly automated and remains unchanged irrespective of the choice of diffusion and observation models, including: elliptic or hypo-elliptic systems; data observed with or without noise; linear or non-linear observation operators.
\item We propose a novel constrained integrator that exploits the Gaussianity of the prior distribution on the pathspace. This leads to an improved scaling in sampling efficiency as the resolution of the latent path discretisation is refined.
\item We propose a scheme to exploit the Markovian structure of \ac{sde} models to ensure that the computational cost of the integrator for the Hamiltonian dynamics scales linearly both with the resolution of path discretisation and the number of observation times. To the best of our knowledge, the developed approach for leveraging the Markovian structure of the model is new.
\item Our method extends the family of \acp{sde} for which statistical calibration is now attainable. Consider the class of $\reals^\intg{X}$-valued \acp{sde} directly observed (without noise) via a non-linear function $\vctfunc{h}:\reals^\intg{X} \to \reals^{\intg{Y}}$ at a finite set of times, for $\intg{X} \geq \intg{Y}\ge 1$. In such a scenario, standard data augmentation schemes fail (for non-trivial choices of $\vctfunc{h}$) as the posterior of the latent variables given the observations does not have a density with respect to the Lebesgue measure. In contrast, our method remains applicable and unchanged.
\end{enumerate}
\begin{remark}[Criteria]
\label{rem:criteria}
To clarify the position of the framework put forward in this work within the wide field of statistical calibration for \ac{sde}s, we list a number of criteria met by our algorithm:
\begin{enumerate}[label=(\roman*), noitemsep]
\item \label{item:criterion-bayesian}
  It carries out full Bayesian inference for the model at hand.
\item \label{item:criterion-roberts-stramer}
  It respects the Roberts--Stramer critique: the mixing times remain stable as the resolution of the path-discretisation is refined.
\item \label{item:criterion-robustness-to-obs-noise}
  It is applicable in  scenarios where data is observed with or without noise; it is stable in the setting of diminishing noise.
\item \label{item:criterion-partial-or-complete-obs}
  It is applicable in the case of both full and partial observations. For partial observations, it accommodates both linear and non-linear observation operators.
\item \label{item:criterion-automated}
  It attains the above via a unified and, in principle, automated methodology.
\end{enumerate}
We have chosen the applications in Section \ref{sec:numerical-experiments} to highlight these properties. To our knowledge, the proposed method is unique in satisfying all of criteria \ref{item:criterion-bayesian}--\ref{item:criterion-automated}.
\end{remark}
The rest of the paper is organised as follows.
Section \ref{sec:diffusion-model} presents a generic class of \ac{sde} models relevant to our work.
Section \ref{sec:posterior-distribution} recasts the inferential problem as one of exploring a posterior distribution supported on a manifold.
Section~\ref{sec:manifold-mcmc} describes the constrained \ac{hmc} method for sampling such distributions on implicitly defined manifolds.
Section \ref{sec:computational-cost} shows how the Markovian structure of \acp{sde} model can be exploited to design a scalable implementation of the methodology.
Section \ref{sec:related-work} discusses related works.
Section~\ref{sec:numerical-experiments} illustrates the approach on several numerical examples, with comments on algorithmic performance and comparisons to alternative \ac{mcmc} methods.
Section \ref{sec:conclusion} concludes with a brief summary and directions for future research.

\paragraph{Notation.}
Sans-serif symbols are used to distinguish random variables from their realisations (respectively, $\rvar{x}$ and $x$).
The set of integers from $\intg{A}\in\integers$ to $\intg{B} \in\integers$ inclusive, $\intg{B} \geq \intg{A}$, is $\range[\intg{A}]{\intg{B}}$. Floor and ceiling operations are denoted $\floor{x}$ and  $\ceil{x}$ respectively. A symbol subscripted by a set indicates an indexed tuple, e.g. $x_{\range[\intg{A}]{\intg{B}}} = \left( x_\intg{s} \right)_{\intg{s}\in\range[\intg{A}]{\intg{B}}}$.
The set of linear maps from a vector space $\set{X}$ to a vector space $\set{Y}$ is $\linops(\set{X},\set{Y})$.
For $\vct{f} : \reals^{\intg{M}} \to \reals^{\intg{N}}$, the Jacobian of $\vct{f}$ is $\jacob\vct{f} : \reals^{\intg{M}} \to \reals^{\intg{N}\times\intg{M}}$ and for $f : \reals^{\intg{M}} \to \reals$, its gradient and Hessian are $\grad f : \reals^{\intg{M}} \to \reals^{\intg{M}}$ and $\grad^2 f:  \reals^{\intg{M}} \to \reals^{\intg{M}\times\intg{M}}$. For a multiple argument function $\vctfunc{g}$ the Jacobian with respect to the $\intg{i}$th argument is denoted $\jacob_{\intg{i}}\vctfunc{g}$ and $\jacob\vctfunc{g} = (\jacob_{1}\vctfunc{g}, \jacob_2\vctfunc{g},\dots)$. The concatenation of vectors $\vct{x}$ and $\vct{y}$ is denoted $[\vct{x}; \vct{y}]$ and the concatenation of a tuple of vectors $\vct{x}_{\range{\intg{N}}}$ is $[\vct{x}_{\range{\intg{N}}}] = [\vct{x}_1; \dots; \vct{x}_{\intg{N}}]$ with the operation acting recursively e.g. $[\vct{x}_{\range{\intg{N}}}; \vct{y}] = [[\vct{x}_{\range{\intg{N}}}];\vct{y}]$.
The determinant of a square matrix $\mtx{M}$ is $\det{\mtx{M}}$. The $\intg{N}\times\intg{N}$ identity matrix is $\idmtx_{\intg{N}}$.
The block diagonal matrix with $\mtx{M}_{\range{\intg{N}}}$ left-to-right along its diagonal is $\diag\mtx{M}_{\range{\intg{N}}}$.
The $\intg{N}$-dimensional Lebesgue measure is $\lebm{\intg{N}}$.
The set of Borel probability measures on a space $\set{X}$ is $\borelprob(\set{X})$.
\section{Diffusion model}
\label{sec:diffusion-model}

We consider the task of inferring the parameters of It{\^o}-type \acp{sde} of the form
\begin{equation}\label{eq:sde-model}
  \diff\rvct{x}(\tau) =
  \vctfunc{a}(\rvct{x}(\tau), \rvct{z}) \diff \tau +
  \opfunc{B}(\rvct{x}(\tau), \rvct{z}) \diff \rvct{b}(\tau)
\end{equation}
defined on a time interval $\set{T} \subseteq \nonnegreals$, where $\rvct{z}$ is a $\set{Z} \subseteq \reals^\intg{Z}$-valued vector of model parameters, $\rvct{x}$ a $\set{X} \equiv \reals^\intg{X}$-valued random process, $\rvct{b}$ a $\set{B} \equiv \reals^\intg{B}$-valued standard Wiener process, $\vctfunc{a} : \set{X} \times \set{Z} \to \set{X}$ the drift function and $\opfunc{B} : \set{X} \times \set{Z} \to \linops(\set{B},\set{X})$ the diffusion coefficient function.
This time-homogeneous \ac{sde} system can be characterised by a family of Markov kernels $\kappa_{\tau} : \set{X} \times \set{Z} \to \borelprob(\set{X})$ with $\kappa_{\tau' - \tau}(\vct{x},\vct{z})(\diff \vct{x})$ the probability of $\rvct{x}(\tau') \in \diff \vct{x}$ given $(\rvct{x}(\tau) = \vct{x}, \rvct{z} = \vct{z})$, for $(\tau, \tau', \vct{x}, \vct{z}) \in \set{T}\times\set{T}\times\set{X}\times\set{Z}$. The parameter $\rvct{z}$ is assigned a prior distribution $\mu \in \borelprob(\set{Z})$ and, given $\rvct{z}$, the initial state $\rvct{x}_0$ is given a prior  $\nu : \set{Z} \to \borelprob(\set{X})$.

We assume the system is observed at $\intg{T}$ times with a fixed inter-observation interval $\upDelta > 0$ and $\set{T} = [0, \intg{T}\upDelta]$. The $\set{Y} \subseteq \reals^\intg{Y}$-valued observed vectors $\rvct{y}_{\range{\intg{T}}}$ are then defined for each $\intg{t} \in \range{\intg{T}}$ as $\rvct{y}_{\intg{t}} = \vctfunc{h}(\rvct{x}(\intg{t}\upDelta), \rvct{z}, \rvct{w}_\intg{t})$ with $\vctfunc{h} : \set{X} \times \set{Z} \times \set{W} \to \set{Y}$ the \emph{observation function}, and $\rvct{w}_{\intg{t}} \sim \eta$ the \emph{observation noise vector} at time index $\intg{t}$ with distribution $\eta \in \borelprob(\set{W})$ and $\set{W} \subseteq \reals^\intg{W}$.
\begin{remark}\label{rem:observation-models}
Two common special cases of our observation model are
\begin{enumerate}[label=(\roman*), noitemsep, topsep=3pt]
  \item noiseless observations: $\vctfunc{h}(\vct{x}, \vct{z}, \vct{w}) := \vcthbar(\vct{x})$ with $\intg{Y} \leq \intg{X}$ and $\intg{W} = 0$,
  \item additive (Gaussian) noise: $\vctfunc{h}(\vct{x}, \vct{z}, \vct{w}) := \vcthbar(\vct{x}) + \mtx{L}(\vct{z}) \vct{w}$ (and $\eta = \nrm(\mathbf{0}, \idmtx_{\intg{Y}})$).
\end{enumerate}
In the former case the observation noise vectors $\rvct{w}_{\range{\intg{T}}}$ can be omitted from the model. Our methodology readily extends to irregular observation times and time-varying model specification -- for the \ac{sde} and the observation parts -- however, for brevity of exposition, we only describe the equispaced and time-independent case here.
\end{remark}
In general, it is neither possible to exactly sample from the Markov kernels $\kappa_\tau$ nor evaluate their densities with respect to the Lebesgue measure on $\set{X}$. We thus adopt a data-augmentation approach \citep{elerian2001likelihood,roberts2001inference} and consider a discrete-time model formed by numerically integrating the original \ac{sde}; although this will introduce discretisation error, the error can be controlled by using a fine time-resolution. We split each inter-observation interval into $\intg{S}$ smaller time steps $\delta = \tfrac{\upDelta}{\intg{S}}$. Given a time discretisation, a variety of numerical schemes for integrating \ac{sde} systems are available with varying levels of complexity and convergence properties \citep{kloeden1992numerical}. The schemes of interest in this article can be expressed as a forward operator $\vctfunc{f}_\delta : \set{Z} \times \set{X} \times \reals^{\intg{V}} \to \set{X}$ defined such that, given parameters $\vct{z} \in \set{Z}$, a current state $\vct{x} \in \set{X}$ and a random vector $\rvct{v}\sim \nrm(\vct{0},\idmtx_{\intg{V}})$, $\vctfunc{f}_\delta(\vct{z},\vct{x}, \rvct{v})$ is approximately distributed according to $\kappa_\delta(\vct{x},\vct{z})$ for small time steps $\delta > 0$.
The simplest and most commonly used scheme is the Euler--Maruyama method, where $\intg{V} = \intg{B}$ and
$
\vctfunc{f}_\delta(\vct{z}, \vct{x}, \vct{v}) =
  \vct{x} + \delta\vctfunc{a}(\vct{x},\vct{z}) +
  \delta^{\frac{1}{2}}\opfunc{B}(\vct{x},\vct{z})\vct{v}.
$
Importantly, the methodology developed in this article straightforwardly accommodates higher order methods, such as the Milstein scheme \citep{milstein1975approximate}.

\begin{model}[t]
  \caption{Time-discretised diffusion generative model.}
  \label{model:generative-model}
\setlength\multicolsep{0pt}
\begin{multicols}{2}
  \begin{algorithmic}
    \Function{$\vctfunc{g}_{\rvct{x}_{:},\rvct{y}_{:}}\mkern-1mu$}
        {$\vct{z},\vct{x}_0,\vct{v}_{\range[1]{\intg{St}}},\vct{w}_{\range[1]{\intg{t}}}$}
      \For{$\intg{s} \in \range[1]{\intg{St}}$}
        \State $\vct{x}_{\intg{s}} \gets \vctfunc{f}_\delta(\vct{z},\vct{x}_{\intg{s}-1},\vct{v}_s)$
        \If{$\intg{s} \bmod \intg{S} \equiv 0$}
          \State $\vct{y}_{\intg{s}/\intg{S}} \gets \vctfunc{h}(\vct{x}_{\intg{s}}, \vct{z}, \vct{w}_{\intg{s}/\intg{S}})$
        \EndIf
      \EndFor
      \vspace{-3pt}
      \State \Return $\vct{x}_{\range[1]{\intg{St}}}$, $\vct{y}_{\range[1]{\intg{t}}}$
    \EndFunction
    \columnbreak
    \State $\rvct{z} \sim \mu$
    \State $\rvct{x}_0 \sim \nu(\rvct{z})$
    \State $\rvct{v}_{\intg{s}} \sim \nrm(\vct{0},\idmtx_{\intg{V}}) ~\forall \intg{s} \in \range{\intg{S}\intg{T}}$
    \State $\rvct{w}_{\intg{t}} \sim \eta ~\forall \intg{t} \in \range{\intg{T}}$
    \State $\rvct{x}_{\range{\intg{S}\intg{T}}}$, $\rvct{y}_{\range{\intg{T}}} \gets \vctfunc{g}_{\rvct{x}_{:},\rvct{y}_{:}}\mkern-1mu(\rvct{z},\rvct{x}_0,\rvct{v}_{\range{\intg{S}\intg{T}}}, \rvct{w}_{\range{\intg{T}}})$
  \end{algorithmic}
\end{multicols}
\vspace{3pt} \end{model}

For a particular choice of numerical scheme, given the parameters $\rvct{z} \sim \mu$ and initial position $\rvct{x}_0 \sim \nu(\rvct{z})$, the states at all subsequent time steps $\rvct{x}_{\range{\intg{S}\intg{T}}}$ are iteratively generated via the forward operator $\vctfunc{f}_\delta$ with $\rvct{x}_{\intg{s}}$ denoting the discrete time approximation to the continuous time state $\rvct{x}(\intg{s}\delta)$. The observations $\rvct{y}_{\range{\intg{T}}}$ are computed from the discrete time state sequence $\rvct{x}_{\range{\intg{S}\intg{T}}}$ via the observation function $\vctfunc{h}$ and observation noise vectors $\rvct{w}_{\range{\intg{T}}}$. The overall generative model is summarised in \cref{model:generative-model}.
\section{Inferential objective on a manifold}
\label{sec:posterior-distribution}

We are interested in computing expectations with respect to the joint posterior of $\rvct{z}$, $\rvct{x}_0$, $\rvct{x}_{\range[1]{\intg{S}{\intg{T}}}}$, given observations $\rvct{y}_{\range{\intg{T}}} = \vct{y}_{\range{\intg{T}}}$. However, the states at nearby time steps will be highly dependent under the prior on $\rvct{x}_{\range{\intg{S}\intg{T}}}$ for small $\delta$.
Such strong dependencies are characteristic of \emph{centred} parametrisations of hierarchical models, and have a deleterious effect on the performance of many approximate inference algorithms \citep{papaspiliopoulos2003non,papaspiliopoulos2007general,betancourt2015hamiltonian}.

\subsection{Non-centred parametrisation}

\label{subsec:non-centred-parametrisation}

One can instead choose to parametrise the inference problem in terms of the latent vectors $\rvct{v}_{\range{\intg{S}\intg{T}}}$ used to numerically integrate the \ac{sde}. Given values for $\rvct{z}$, $\rvct{x}_0$ and $\rvct{v}_{\range{\intg{S}\intg{T}}}$, the state sequence $\rvct{x}_{\range[1]{\intg{S}\intg{T}}}$ can be deterministically computed.
Such a reparametrisation has the property that, under the prior, all components of the latent vectors $\rvct{v}_{\range{\intg{S}\intg{T}}}$ are independent standard normal variables. We further assume the following.
\begin{assumption}
  \label{ass:parameter-and-initial-state-generator-functions}
    There exist functions $\vctfunc{g}_{\rvct{z}} : \reals^\intg{U} \to \set{Z}$ and  $\vctfunc{g}_{\rvct{x}_0} : \set{Z} \times \reals^{\intg{V}_0} \to \set{X}$ and corresponding distributions $\tilde{\mu} \in \borelprob(\reals^{\intg{U}})$, $\tilde{\nu} \in \borelprob(\reals^{\intg{V}_0})$ with strictly positive smooth density functions with respect to the Lebesgue measures $\lebm{\intg{U}}$ and $\lebm{\intg{V}_0}$, respectively, such that $\vctfunc{g}_{\rvct{z}}(\rvct{u}) \sim \mu$ and $\vctfunc{g}_{\rvct{x}_0}\mkern-2mu(\vct{z},\rvct{v}_0) \sim \nu(\vct{z})~\forall \vct{z} \in \set{Z}$ if $\rvct{u} \sim \tilde{\mu}$ and $\rvct{v}_0 \sim \tilde{\nu}$.
\end{assumption}
Under such parametrisation in terms of $\rvct{q}:=[\rvct{u}; \rvct{v}_0; \rvct{v}_{\range[1]{\intg{S}{\intg{T}}}}; \rvct{w}_{\range{\intg{T}}}]$ all of $(\rvct{u}, \rvct{v}_0, \rvct{v}_{\range[1]{\intg{S}{\intg{T}}}}, \rvct{w}_{\range{\intg{T}}})$ are then a-priori independent and the resulting prior distribution $\rho \in \borelprob(\reals^\intg{Q})$ with $\intg{Q} = \intg{U} +\intg{V}_0 + \intg{S}\intg{T}\intg{V} + \intg{T}\intg{W}$, has a density with respect to the Lebesgue measure $\lebm{\intg{Q}}$,
\begin{equation}\label{eq:non-centred-prior-distribution-density}
  \td{\rho}{\lebm{\intg{Q}}}([\vct{u}; \vct{v}_0; \vct{v}_{\range[1]{\intg{S}{\intg{T}}}}; \vct{w}_{\range{\intg{T}}}])
  \propto
  \td{\tilde\mu}{\lebm{\intg{U}}}(\vct{u})
  \td{\tilde\nu}{\lebm{\intg{V}_0}}(\vct{v}_0)
  \prod_{\intg{s}=1}^{\intg{S}{\intg{T}}}
  \exp\left(
    -\tfrac{1}{2}
      \vct{v}_{\intg{s}}\tr\vct{v}_{\intg{s}}
  \right)
  \prod_{\intg{t}=1}^{\intg{T}}
  \td{\eta}{\lebm{\intg{W}}}(\vct{w}_\intg{t}).
\end{equation}
\cref{alg:non-centred-generative-model} gives the generative model under this \emph{non-centred} parametrisation and  defines a function $\vctfunc{g}_{\rvct{y}_{:}}$ which generates observations given values for the latent variables.
\begin{model}[t]
  \caption{Non-centred parametrisation of generative model.}
  \label{alg:non-centred-generative-model}
\setlength\multicolsep{0pt}
\begin{multicols}{2}
  \begin{algorithmic}
    \Function{$\vctfunc{g}_{\rvct{y}_{:}}\mkern-2mu$}
      {$\vct{u}, \vct{v}_0, \vct{v}_{\range[1]{\intg{S}\intg{t}}}, \vct{w}_{\range{\intg{t}}}$}
    \vspace{2pt}
    \State $\vct{z} \gets \vctfunc{g}_{\rvct{z}}(\vct{u})$
    \State $\vct{x}_0 \gets \vctfunc{g}_{\rvct{x}_0}\mkern-2mu(\vct{z},\vct{v}_0)$
    \State $\vct{x}_{\range[1]{\intg{S}\intg{t}}}$, $\vct{y}_{\range{\intg{t}}} \gets \vctfunc{g}_{\rvct{x}_{:},\rvct{y}_{:}}(\vct{z}, \vct{x}_0, \vct{v}_{\range[1]{\intg{S}\intg{t}}}, \vct{w}_{\range{\intg{t}}})$
    \State \Return $\vct{y}_{\range{\intg{t}}}$
    \EndFunction
    \columnbreak
    \State $\rvct{u} \sim \tilde{\mu}$
    \State $\rvct{v}_0 \sim \tilde{\nu}$
    \State $\rvct{v}_{\intg{s}} \sim \nrm(\vct{0},\idmtx_{\intg{V}}) ~\forall \intg{s} \in \range{\intg{S}\intg{T}}$
    \State $\rvct{w}_{\intg{t}} \sim \eta ~\forall \intg{t} \in \range{\intg{T}}$
    \State $\rvct{y}_{\range{\intg{T}}} \gets \vctfunc{g}_{\rvct{y}_{:}}\mkern-2mu(\rvct{u},\rvct{v}_0,\rvct{v}_{\range[1]{\intg{S}\intg{T}}}, \rvct{w}_{\range{\intg{T}}})$
  \end{algorithmic}
\end{multicols}
\vspace{3pt} \end{model}
The observations can be thought of as imposing a series of constraints on the possible values of the latent variables $\rvct{q}$; under additional assumptions on the regularity of the mapping $\vctfunc{g}_{\rvct{y}_{:}}$ from latent variables to observations, the set of $\rvct{q}$ values satisfying the constraints will form a differentiable manifold embedded in $\reals^{\intg{Q}}$.
The posterior distribution on $\rvct{q}$ given $\rvct{y}_{\range{\intg{T}}} = \vct{y}_{\range{\intg{T}}}$ will not have a density with respect to the Lebesgue measure $\lebm{\intg{Q}}$ as the manifold it has support on is a $\lebm{\intg{Q}}$-null set. In the following section we show however that by using a different reference measure we can compute a tractable density function for the posterior.

\subsection{Target posterior on manifold}
\label{subsec:posterior-distribution-non-centred}

We define a \emph{constraint function} $\vctfunc{c} : \reals^\intg{Q} \to \reals^{\intg{C}}$ with $\intg{C}=\intg{TY} < \intg{Q}$ as
\begin{equation}\label{eq:constraint-function}
  \vctfunc{c}([\vct{u}; \vct{v}_0; \vct{v}_{\range[1]{\intg{S}\intg{T}}}; \vct{w}_{\range{\intg{T}}}]) :=
    [\vctfunc{g}_{\rvct{y}_{:}}\mkern-2mu(\vct{u}, \vct{v}_0, \vct{v}_{\range[1]{\intg{S}\intg{T}}}, \vct{w}_{\range{\intg{T}}})]
     - [\vct{y}_{\range{\intg{T}}}],
\end{equation}
with the set of values on the manifold $\set{M} := \lbrace \vct{q} \in \reals^{\intg{Q}}: \vctfunc{c}(\vct{q}) = \vct{0}\rbrace$ corresponding to all inputs of $\smash{\vctfunc{g}_{\rvct{y}_{:}}}$ consistent with the observations. We make the following assumption.
\begin{assumption}\label{ass:differentiability-and-surjectivity-constraint-function}
  The constraint function $\vctfunc{c}$ is continuously differentiable and has Jacobian $\jacob\vctfunc{c}$ which is full row-rank $\rho$-almost surely.
\end{assumption}
The differentiability requirement will be met if $\vctfunc{f}_\delta$, $\vctfunc{g}_{\rvct{z}}$, $\vctfunc{g}_{\rvct{x}_0}$ and $\vctfunc{h}$ are all themselves continuously differentiable with respect to each of their arguments. The rank condition on the Jacobian requires that the observed variables do not give redundant information about the latent variables, i.e.~no observed variable can be expressed as a deterministic function of a subset of the other observed variables. In the case of observations subject to Gaussian additive noise this will always be satisfied if the noise covariance is full-rank. In the noiseless observation case the condition will be met if no component of the state at an observation time $\rvct{x}_{\intg{S}\intg{t}}$ is fully determined by the state at the previous observation time $\rvct{x}_{\intg{S}(\intg{t}-1)}$ and parameters $\rvct{z}$, and the function $\vcthbar: \set{X} \to \set{Y}$ has Jacobian with full row-rank everywhere.

Under these assumptions $\set{M}$ will be a $\intg{D} = \intg{Q} - \intg{C}$ dimensional differentiable  manifold embedded into the $\intg{Q}$ dimensional ambient space. The posterior distribution $\smash{\pi \in \borelprob(\reals^{\intg{Q}})}$ on $\rvct{q}$ given $\smash{\vctfunc{c}(\rvct{q}) = \vct{0}}$ (and so $\smash{\rvct{y}_{\range{\intg{T}}} = \vct{y}_{\range{\intg{T}}}}$) is supported only on $\set{M}$. Note that $\set{M}$ has zero Lebesgue measure, so $\pi$ does not have a density with respect to $\lebm{\intg{Q}}$. To define an appropriate reference measure we further assume the following.
\begin{assumption}\label{ass:latent-space-metric}
  The ambient latent space $\reals^{\intg{Q}}$ is equipped with a metric tensor with a fixed positive definite matrix representation $\mtx{M}$.
\end{assumption}
A possible reference measure is then the $\intg{D}$-dimensional Hausdorff measure $\haum{\mtx{M}}{\intg{D}}$ on the ambient space, which has the required property that $\pi$ is absolutely continuous with respect to $\haum{\mtx{M}}{\intg{D}}$. For measurable subsets $\set{A} \subseteq \set{M}$ we have that $\riem{\mtx{M}}{\set{M}}(\set{A}) = \haum{\mtx{M}}{\intg{D}}(\set{A})$ where $\riem{\mtx{M}}{\set{M}}$ is the Riemannian measure on the manifold $\set{M}$ with metric induced from the ambient metric (see \cref{lem:equivalence-of-riemannian-and-hausdorff-measures} in \appcref{app:proof-manifold-density}). As later results will be more naturally stated in terms of the Riemannian measure, we will use $\riem{\mtx{M}}{\set{M}}$ as the reference measure here.
\begin{proposition}
  \label{prop:manifold-density}
  Under \cref{ass:parameter-and-initial-state-generator-functions,ass:differentiability-and-surjectivity-constraint-function,ass:latent-space-metric} the posterior $\pi$ has a density
  \begin{equation}\label{eq:manifold-density}
    \td{\pi}{\riem{\mtx{M}}{\set{M}}}(\vct{q}) \propto
    \td{\rho}{\lebm{\intg{Q}}}(\vct{q})
    \det{
    \jacob{\vctfunc{c}}(\vct{q}) \mtx{M}^{-1} \jacob{\vctfunc{c}}(\vct{q})\tr
    }^{-\frac{1}{2}}.
  \end{equation}
\end{proposition}
A proof is given in \appcref{app:proof-manifold-density}. See also \citet{rousset2010free}, \citet{diaconis2013sampling} and \citet{graham2017asymptotically}.
The negative log posterior density thus reads
\begin{align*}
  \ell(\vct{q}):=
  -\log \td{\rho}{\lebm{\intg{Q}}}(\vct{q}) +
  \frac{1}{2}
  \log\det{
	  \mtx{G}_{\mtx{M}}(\vct{q})
  },
\end{align*}
where the $\intg{C}\times\intg{C}$ matrix $\opfunc{G}_{\mtx{M}}(\vct{q}) := \jacob{\vctfunc{c}}(\vct{q}) \mtx{M}^{-1} \jacob{\vctfunc{c}}(\vct{q})\tr$ is termed the \emph{Gram matrix}. %
\section{MCMC on implicitly defined manifolds}
\label{sec:manifold-mcmc}

In this section, we review \ac{mcmc} methods for sampling from a distribution supported on an implicitly defined manifold. We stress at this point that we do not design a fundamentally new such \ac{mcmc} method but instead rely on modifying and combining existing methodologies. In particular, we adopt a symplectic integrator for constrained Hamiltonian systems \citep{andersen1983rattle,leimkuhler1994symplectic,leimkuhler2016efficient} to simulate Hamiltonian dynamics trajectories on the manifold, and use this as a proposal generating mechanism within a \acf{hmc} scheme \citep{duane1987hybrid,neal2011mcmc,betancourt2017conceptual}. The use of constrained Hamiltonian dynamics within an \ac{hmc} context has been previously proposed multiple times - see for example \citet{hartmann2005constrained}, \citet[Ch.~3]{rousset2010free}, \citet{brubaker2012family} and \citet{lelievre2019hybrid}. We refer the interested reader to \citet{arnol2013mathematical,holm2009geometric} for general background on constrained mechanics and to \citet[Ch. 3]{barp2020bracket} for a comprehensive review of \ac{hmc} on manifolds.

\subsection{Constrained Hamiltonian dynamics}

To define the constrained Hamiltonian system, we first augment the latent vector~$\rvct{q}$, henceforth the \emph{position}, with a \emph{momentum} $\rvct{p}$. Formally the momentum is a \emph{co-vector} i.e.~a linear form in $\linops(\reals^{\intg{Q}},\reals)$ and the metric on the position space induces a \emph{co-metric} on the momentum space with matrix representation $\mtx{M}^{-1}$. As a common abuse of notation, we will not distinguish between vectors and co-vectors and simply consider $\rvct{p}$ as a vector in $\reals^\intg{Q}$ equipped with a metric with matrix representation $\mtx{M}^{-1}$. The \emph{Hamiltonian} function $h : \reals^\intg{Q} \times \reals^\intg{Q} \to \reals$ is then defined as%
\begin{equation}\label{eq:hamiltonian}
  h(\vct{q},\vct{p}) :=
  \ell(\vct{q})+
  \tfrac{1}{2}\vct{p}\tr\mtx{M}^{-1}\vct{p}.
\end{equation}
Thus far we have an unconstrained Hamiltonian system. To restrict $\rvct{q}$ to $\set{M}$, we introduce a Lagrange multiplier function
$\vctfunc{\lambda} : \reals^{\intg{Q}} \times \reals^{\intg{Q}} \to \reals^{\intg{C}}$ implicitly defined so that constraint $\vctfunc{c}(\vct{q}) = \vct{0}$ is enforced, at all times, in the following dynamics. The constrained Hamiltonian dynamics associated with the Hamiltonian in \eqref{eq:hamiltonian} are then described by the system of \acp{dae}
\begin{equation}\label{eq:constrained-hamiltonian-dynamic}
  \td{\vct{q}}{t} =
  \mtx{M}^{-1}\vct{p},
  \qquad
  \td{\vct{p}}{t} =
  -\grad \ell(\vct{q})
  -\jacob\vctfunc{c}(\vct{q})\tr\vctfunc{\lambda}(\vct{q},\vct{p}),
  \qquad
  \vctfunc{c}(\vct{q}) = \vct{0}.
\end{equation}
The condition that the \emph{primary constraints}, $\vctfunc{c}(\vct{q}) = \vct{0}$, are preserved in time implies a set of \emph{secondary constraints} of the form
\(
  \label{eq:momenta-constraint}
  \jacob\vctfunc{c}(\vct{q})\mtx{M}^{-1}\vct{p} = \vct{0}.
\)
\begin{definition}
  \label{def:cotangent-space}
  The set of momenta satisfying the secondary constraints at a position $\vct{q}$ coincides with the \emph{co-tangent space} of the manifold $\set{M}$ at $\vct{q}$, denoted
  \begin{equation*}
    \tangent^*_{\vct{q}}\set{M} :=
    \lbrace
      \vct{p} \in \reals^{\intg{Q}} :
      \jacob\vctfunc{c}(\vct{q})\mtx{M}^{-1}\vct{p} = \vct{0}
    \rbrace.
  \end{equation*}
\end{definition}
\begin{definition}
  \label{def:cotangent-bundle}
  The set of positions and momenta in the manifold and corresponding co-tangent spaces respectively are termed the \emph{co-tangent bundle}, denoted
  \begin{equation*}
    \tangent^*\set{M} :=
    \lbrace
      \vct{q} \in \set{M}, \vct{p} \in \tangent^*_{\vct{q}}\set{M}
    \rbrace
    =
    \lbrace
      \vct{q} \in \reals^\intg{Q}, \vct{p} \in \reals^{\intg{Q}}:
      \vctfunc{c}(\vct{q}) = \vct{0},
      \jacob\vctfunc{c}(\vct{q})\mtx{M}^{-1}\vct{p} = \vct{0}
    \rbrace.
  \end{equation*}
\end{definition}
\begin{remark}
  $\tangent^*\set{M}$ is a \emph{symplectic manifold} with a \emph{symplectic form} given by the restriction of the symplectic form on $\reals^{\intg{Q}}\times\reals^{\intg{Q}}$ to $\tangent^{*}\set{M}$, which under \cref{ass:differentiability-and-surjectivity-constraint-function,ass:latent-space-metric} is almost surely non-degenerate.
\end{remark}
\begin{definition}
  The symplectic form on $\tangent^*\set{M}$ induces a volume form and corresponding \emph{Liouville measure} denoted $\sigma_{\tangent^*\set{M}}$, which can be decomposed as
  \begin{equation}\label{eq:liouville-measure}
    \sigma_{\tangent^{*}\set{M}}(\diff\vct{q},\diff\vct{p}) =
    \riem{\mtx{M}}{\set{M}}(\diff\vct{q}) \,
    \riem{\mtx{M}^{-1}}{\tangent^{\ast}_{\vct{q}}\set{M}}(\diff\vct{p}),
  \end{equation}
  which is independent of the choice of $\mtx{M}$ \citep[Proposition 3.40]{rousset2010free}.
\end{definition}
\noindent The \emph{flow map} associated with the solution to the \acp{dae} in \cref{eq:constrained-hamiltonian-dynamic} is
\(
  \opfunc{\Phi}^{h_{\vctfunc{c}}}_t:
  \tangent^*\set{M} \to \tangent^*\set{M},
\)
such that for $(\vct{q}(0), \vct{p}(0)) \in \tangent^*\set{M}$ and $t \geq 0$ we have $(\vct{q}(t),\vct{p}(t)) = \opfunc{\Phi}^{h_{\vctfunc{c}}}_t(\vct{q}(0), \vct{p}(0))$. Fundamental properties of $\opfunc{\Phi}^{h_{\vctfunc{c}}}_t$ are that it is energy conserving and symplectic.
\begin{proposition}
  \label{prop:hamiltonian-conservation}
  The Hamiltonian in \cref{eq:hamiltonian} is conserved under the flow map $\opfunc{\Phi}^{h_{\vctfunc{c}}}_t$.
\end{proposition}
\begin{proposition}
  \label{prop:symplecticness-of-constrained-dynamic}
  The flow map $\opfunc{\Phi}^{h_{\vctfunc{c}}}_t$ preserves the symplectic form on $\tangent^*\set{M}$.
\end{proposition}
See for example \citet[Chapter 7]{leimkuhler2004simulating}. Proofs are also given in \appcref{app:proof-hamiltonian-conservation,app:proof-symplecticness-of-constrained-dynamic}.
Together these properties mean the flow map $\opfunc{\Phi}^{h_{\vctfunc{c}}}_t$ has an invariant measure on $\tangent^*\set{M}$.
\begin{corollary}
  The conservation properties in \cref{prop:symplecticness-of-constrained-dynamic,prop:hamiltonian-conservation} imply that the measure $\zeta(\diff\vct{q},\diff\vct{p}) \propto \exp(-h(\vct{q},\vct{p})) \riem{}{\tangent^*\set{M}}(\diff\vct{q},\diff\vct{p})$ is invariant under the flow map $\opfunc{\Phi}^{h_{\vctfunc{c}}}_t$ corresponding to the constrained dynamics in \cref{eq:constrained-hamiltonian-dynamic}.
\end{corollary}
Using the definitions in \cref{eq:hamiltonian} and \cref{eq:liouville-measure}, it readily follows that the target posterior $\pi(\diff\vct{q})\propto\exp(-\ell(\vct{q}))\riem{\mtx{M}}{\set{M}}(\diff\vct{q})$ is the marginal distribution on the position under the invariant measure $\zeta$. Thus, the flow map $\opfunc{\Phi}^{h_{\vctfunc{c}}}_t$ can be used to construct a family of Markov kernels which marginally leave $\pi$ invariant.

\subsection{Momentum resampling}
\label{subsec:momentum-resampling}

As the dynamics remain confined to a level-set of the Hamiltonian in \cref{eq:hamiltonian}, a Markov chain constructed by iterating $\opfunc{\Phi}^{{h_{\vctfunc{c}}}}_t$ will not be ergodic. By resampling the momentum between $\opfunc{\Phi}^{h_{\vctfunc{c}}}_t$ applications we can however move between Hamiltonian level-sets.

To orthogonally (with respect to the co-metric) project a momentum onto $\tangent^*_{\vct{q}}\set{M}$, the co-tangent space at $\vct{q}$, we apply the projector $\opfunc{P}_{\mtx{M}}(\vct{q})$, defined as
\begin{equation}\label{eq:projector-operator}
\opfunc{P}_{\mtx{M}}(\vct{q}) := \idmtx_{\intg{Q}} - \jacob\vctfunc{c}(\vct{q})\tr\opfunc{G}_{\mtx{M}}(\vct{q})^{-1}\jacob\vctfunc{c}(\vct{q})\mtx{M}^{-1}.
\end{equation}
Using $\opfunc{P}_{\mtx{M}}$ we can independently sample a momentum from its conditional distribution given the position under the measure $\zeta$ by projecting a sample from $\nrm(\vct{0},\mtx{M})$.
\begin{proposition}
  \label{prop:momentum-resampling-correctness}
  If $\tilde{\rvct{p}} \sim \nrm(\vct{0},\mtx{M})$ then $\rvct{p} \gets \opfunc{P}_{\mtx{M}}(\vct{q}) \tilde{\rvct{p}}$ is distributed with density \ifclass{statsoc}{}{\newline}$\exp(-\vct{p}\tr\mtx{M}^{-1}\vct{p} / 2)$ with respect to $\riem{\mtx{M}^{-1}}{\tangent^*_{\vct{q}}\set{M}}$, the distribution of $\rvct{p} \gvn \rvct{q}=\vct{q}$ for $\rvct{q},\rvct{p} \sim \zeta$.
\end{proposition}
See  \appcref{app:proof-momentum-resampling-correctness} for a proof.

\subsection{Numerical discretisation}
\label{subsec:constrained-hamiltonian-dynamics-numerical-discretisation}

In general the system of \acp{dae} in \cref{eq:constrained-hamiltonian-dynamic} will not have an analytic solution, and we are required to use a time discretisation to approximate the exact flow map $\smash{\opfunc{\Phi}^{h_{\vctfunc{c}}}_t}$. We will first introduce a class of symplectic integrators for the unconstrained Hamiltonian system before showing how they can be used to construct a symplectic integrator for the constrained Hamiltonian system.

\subsubsection{Unconstrained integrator: St\"ormer--Verlet and Gaussian splittings}
A standard approach for defining symplectic integrators for Hamiltonian systems is to \emph{split} the Hamiltonian into a sum of components for which the exact corresponding flow map can be computed, with a splitting of the form
$\label{eq:hamiltonian-splitting}
  h(\vct{q},\vct{p}) =
  h_1(\vct{q}) + h_2(\vct{q},\vct{p})
$
particularly common. If $\opfunc{\Phi}^{h_1}_t$ and $\opfunc{\Phi}^{h_2}_t$ denote the flow maps associated with the dynamics for Hamiltonians $h_1$ and $h_2$ respectively, then the symmetric composition
\(
  \smash{
    \opfunc{\Psi}_{t}
    =
    \opfunc{\Phi}^{h_1}_{\sfrac{t}{2}} \circ \opfunc{\Phi}^{h_2}_{t} \circ \opfunc{\Phi}^{h_1}_{\sfrac{t}{2}}
  }
\)
is a symplectic and second-order accurate integrator for the Hamiltonian system \citep{leimkuhler2004simulating}. Furthermore, as both $\opfunc{\Phi}^{h_1}_t$ and $\opfunc{\Phi}^{h_2}_t$ are time-reversible, $\opfunc{\Psi}_{t}$ is also time-reversible.

Various choices can be made for splitting the Hamiltonian of interest in \cref{eq:hamiltonian} between $h_1$ and $h_2$, subject to the requirement that the flow map $\opfunc{\Phi}^{h_2}_t$ can be computed, with $\opfunc{\Phi}^{h_1}_t(\vct{q}, \vct{p}) = (\vct{q}, \vct{p} - t\grad h_1(\vct{q}))$ always trivial to compute. An obvious splitting is $h_1(\vct{q}) = \ell(\vct{q})$ and $h_2(\vct{q},\vct{p}) = \frac{1}{2}\vct{p}\tr\mtx{M}^{-1}\vct{p}$; in this case $\opfunc{\Phi}^{h_2}_{t}(\vct{q}, \vct{p}) = (\vct{q} + t \mtx{M}^{-1} \vct{p}, \vct{p})$. The composition then corresponds to the St\"ormer--Verlet integrator \citep{verlet1967computer}.

In our setting, the log prior density on the ambient space $\log \inlinetd{\rho}{\lebm{\intg{Q}}}$ is quadratic in the components of the position $\rvct{q}$ corresponding to $\rvct{v}_{\range{\intg{S}\intg{T}}}$ due to their standard normal prior distribution. It will typically also be possible to choose an appropriate parametrization such that the prior densities $\inlinetd{\tilde\mu}{\lambda_{\intg{Z}}}$, $\inlinetd{\tilde\nu}{\lambda_{\intg{X}}}$ and $\inlinetd{\eta}{\lambda_{\intg{W}}}$ are equal to or well approximated by standard normal densities. An alternative splitting, which can be useful in this setting, is then
\(
  h_1(\vct{q}) = \ell(\vct{q}) - \tfrac{1}{2}\vct{q}\tr\vct{q}
\),
\(
  h_2(\vct{q},\vct{p}) = \tfrac{1}{2}\vct{q}\tr\vct{q} + \tfrac{1}{2}\vct{p}\tr\mtx{M}^{-1}\vct{p}
\),
with the simplification $h_1(\vct{q}) = \frac{1}{2}\log\det{\mtx{G}_{\mtx{M}}(\vct{q})}$ when $\log\inlinetd{\rho}{\lebm{\intg{Q}}}(\vct{q}) = -\tfrac{1}{2}\vct{q}\tr\vct{q}$. The quadratic form of $h_2$ and corresponding linear derivatives mean the corresponding flow map is still exactly computable.
If we let $\mtx{R}$ be an orthonormal matrix with the normalised eigenvectors of $\mtx{M}^{-1}$ as columns and $\mtx{\Omega}$ a diagonal matrix of the square-roots of the eigenvalues such that $\mtx{M}^{-1} = \mtx{R} \mtx{\Omega}^2 \mtx{R}\tr$ then we have that
\begin{equation*}\label{eq:gaussian-h2-flow-map}
  \opfunc{\Phi}^{h_2}_{t}(\vct{q}, \vct{p}) =
  (
    \mtx{R}\cos(\mtx{\Omega}t)\mtx{R}\tr\vct{q} + \mtx{R}\mtx{\Omega}\sin(\mtx{\Omega}t)\mtx{R}\tr\vct{p},
    \mtx{R}\cos(\mtx{\Omega}t)\mtx{R}\tr\vct{p} - \mtx{R}\mtx{\Omega}^{-1}\sin(\mtx{\Omega}t)\mtx{R}\tr\vct{q}
  ).
\end{equation*}
This splitting and corresponding integrator has been used previously in various settings \citep{beskos2011hybrid,neal2011mcmc,beskos2013advanced,shahbaba2014split}. Importantly as the flow-map $\opfunc{\Phi}^{h_2}_{t}$ exactly preserves Gaussian prior measures, under certain assumptions the change in Hamiltonian over a trajectory generated using the integrator does not grow as the dimension $\intg{Q}$ of the space is increased, so the probability of accepting a proposed move from the trajectory remains independent of dimension for a fixed step size. This in contrast to the St\"ormer-Verlet integrator for which for a fixed step size the accept probability will tend to zero as the dimension becomes large \citep{beskos2011hybrid}.

In the context here of inference in partially-observed diffusion models, as the time step $\delta$ of the discretisation of the diffusion is decreased (or equivalently $\intg{S}$ increased), the dimension of set of latent noise vectors $\rvct{v}_{\range{\intg{S}\intg{T}}}$ and so $\rvct{q}$ will increase, with the prior distribution on $\rvct{q}$ tending to a distribution with a density with respect to an infinite-dimensional Gaussian measure in the limit $\delta \to 0$. As here the target \emph{posterior} distribution has support only on a submanifold of the ambient space, the results of \citet{beskos2011hybrid} do not directly carry over, however, empirically we have found that a constrained integrator based on this \emph{Gaussian splitting} gives an improved scaling in sampling efficiency with $\intg{S}$ compared to the \emph{St\"ormer--Verlet splitting} as we illustrate in our numerical experiments in \cref{sec:numerical-experiments}.

\subsubsection{Constrained integrator}

We now show how a constraint-preserving symplectic integrator can be formed from the unconstrained integrator $\opfunc{\Psi}_t$. In \citet[Section 3.1]{reich1996symplectic} it is observed that the map defined by
\(
  \opfunc{\Pi}_{\vctfunc{\lambda}}(\vct{q},\vct{p}) = (\vct{q},\vct{p} - \jacob{\vctfunc{c}}(\vct{q})\tr\vctfunc{\lambda}(\vct{q},\vct{p})),
\)
with $\vct{q} \in \set{M}$, is symplectic for any function $\vctfunc{\lambda}$ that is sufficiently regular (e.g.~continuously differentiable).
\citet{reich1996symplectic} then shows that if a second-order accurate symplectic integrator for an unconstrained system with Hamiltonian as in \cref{eq:hamiltonian} is defined by the map $\opfunc{\Psi}_t$, then the integrator with step defined by the composition
\(
  (\vct{q}',\vct{p}') =
  \opfunc{\Pi}_{\vctfunc{\lambda}'} \circ
  \opfunc{\Psi}_{t} \circ
  \opfunc{\Pi}_{\vctfunc{\lambda}}(\vct{q},\vct{p}),
\)
with $\vctfunc{\lambda}$ implicitly defined by solving for the primary constraints, $\vctfunc{c}(\vct{q}') = \vct{0}$, and $\vctfunc{\lambda}'$ by solving for the secondary constraints, $\jacob\vctfunc{c}(\vct{q}')\mtx{M}^{-1}\vct{p}' = \vct{0}$, is a second-order accurate symplectic integrator for the corresponding constrained system.

Rather than composing instances of $\opfunc{\Pi}_{\vctfunc{\lambda}}$ with the overall map $\opfunc{\Psi}_{t}$  as proposed by \citet{reich1996symplectic}, we can instead consider composing $\opfunc{\Pi}_{\vctfunc{\lambda}}$ with the component maps which make up $\opfunc{\Psi}_{t}$ to enforce the constraints within each `sub-step'. This was proposed for the specific case of $\opfunc{\Psi}_{t}$ corresponding to a St\"ormer--Verlet integrator in the \emph{geodesic integration} algorithm of \citet{leimkuhler2016efficient}.

For a general quadratic $h_2$ (covering both the St\"ormer--Verlet and Gaussian splittings introduced above) the associated component flow-maps $\opfunc{\Phi}^{h_1}_t$ and $\opfunc{\Phi}^{h_2}_t$ can be expressed for suitable choices of matrices $(\opfunc{\Gamma}^{q,q}_{t}, \opfunc{\Gamma}^{q,p}_{t}, \opfunc{\Gamma}^{p,q}_{t}, \opfunc{\Gamma}^{p,p}_{t})$ as
\begin{equation}
  \opfunc{\Phi}^{h_1}_{t}(\vct{q}, \vct{p}) :=
  (\vct{q}, \vct{p} - t\grad h_1(\vct{q})),
  ~~
  \opfunc{\Phi}^{h_2}_{t}(\vct{q}, \vct{p}) := (
    \opfunc{\Gamma}^{q,q}_{t}\vct{q} + \opfunc{\Gamma}^{q,p}_t\vct{p},
    \opfunc{\Gamma}^{p,q}_{t}\vct{q} + \opfunc{\Gamma}^{p,p}_t\vct{p}
  ).
\end{equation}
First considering the flow-map $\opfunc{\Phi}^{h_1}_t$, we define the constraint-preserving composition%
\begin{equation}
  \opfunc{\Xi}^{h_1}_t(\vct{q},\vct{p}) :=
  \opfunc{\Pi}_{\vctfunc{\lambda}} \circ \opfunc{\Phi}^{h_1}_t(\vct{q},\vct{p}) =
  (\vct{q}, \vct{p} - t \grad h_1(\vct{q}) - \jacob\vctfunc{c}(\vct{q})\tr\vctfunc{\lambda}),
\end{equation}
with $\vctfunc{\lambda}$ implicitly defined by the condition $\opfunc{\Xi}^{h_1}_t(\vct{q},\vct{p}) \in \tangent^*\set{M} ~~\forall (\vct{q},\vct{p}) \in \tangent^*\set{M}$.
Solving for $\vctfunc{\lambda}$ yields the explicit definition
\(
  \opfunc{\Xi}^{h_1}_t(\vct{q},\vct{p}) :=
  (\vct{q}, \opfunc{P}_{\mtx{M}}(\vct{q})(\vct{p} - t \grad h_1(\vct{q}))).
\)
As $\opfunc{\Xi}^{h_1}_{-t} \circ \opfunc{\Xi}^{h_1}_{t}(\vct{q},\vct{p}) = (\vct{q},\vct{p})$ for all $(\vct{q},\vct{p})\in\tangent^*\set{M}$, the mapping $\opfunc{\Xi}^{h_1}_{t}$ is time-reversible. Now considering the $\opfunc{\Phi}^{h_2}_t$ map, we first consider the composition
\begin{equation}\label{eq:h2-step-partial-composition}
  \opfunc{\Phi}^{h_2}_t \circ \opfunc{\Pi}_{\vctfunc{\lambda}} (\vct{q},\vct{p}) = (
    \opfunc{\Gamma}^{q,q}_{t}\vct{q} + \opfunc{\Gamma}^{q,p}_t(\vct{p} - \jacob\vctfunc{c}(\vct{q})\tr\vctfunc{\lambda}),
    \opfunc{\Gamma}^{p,q}_{t}\vct{q} + \opfunc{\Gamma}^{p,p}_t(\vct{p} - \jacob\vctfunc{c}(\vct{q})\tr\vctfunc{\lambda})
  ),
\end{equation}
with $\vctfunc{\lambda}$ implicitly defined by requiring the following to hold for any $(\vct{q},\vct{p}) \in \tangent^*\set{M}$,
\begin{equation}\label{eq:h2-step-position-constraint-condition}
  \vctfunc{c}\left(
    \opfunc{\Gamma}^{q,q}_{t}\vct{q} + \opfunc{\Gamma}^{q,p}_t\vct{p}
- \opfunc{\Gamma}^{q,p}_t\jacob\vctfunc{c}(\vct{q})\tr\vctfunc{\lambda}
  \right) = \vct{0}.
\end{equation}
For general constraint functions $\vctfunc{c}$ this is a non-linear system of
equations in $\vctfunc{\lambda}$ that needs to be solved using an iterative method. Newton's method gives the update
\begin{equation}\label{eq:newton-iteration}
  \begin{aligned}
  (\vct{q}_\intg{j}, \vct{p}_\intg{j}) &\gets \opfunc{\Phi}^{h_2}_t(\vct{q}, \vct{p} - \jacob\vctfunc{c}(\vct{q})\tr\vctfunc{\lambda}_\intg{j})),\\
  \vctfunc{\lambda}_{\intg{j}+1} &\gets
  \vctfunc{\lambda}_{\intg{j}} +
  (\jacob\vctfunc{c}(\vct{q}_{\intg{j}})(\opfunc{\Gamma}^{q,p}_t)^{-1}\jacob\vctfunc{c}(\vct{q})\tr)^{-1}
  \vctfunc{c}(\vct{q}_{\intg{j}})
  \quad\text{with}~
  \vct{\lambda}_0 = \mathbf{0}.
  \end{aligned}
\end{equation}

Assuming for now the iterative solver can find a value for $\vctfunc{\lambda}$ to satisfy \eqref{eq:h2-step-position-constraint-condition}, the composition in \eqref{eq:h2-step-partial-composition} preserves the primary constraints, but not the secondary constraints in general. The secondary constraints can be enforced by composing with a further instance of the map $\opfunc{\Pi}_{\vctfunc{\lambda}'}$ resulting in the overall composition
\(
  \opfunc{\Xi}^{h_2}_t(\vct{q},\vct{p}) :=
  \opfunc{\Pi}_{\vctfunc{\lambda}'} \circ \opfunc{\Phi}^{h_2}_t \circ \opfunc{\Pi}_{\vctfunc{\lambda}} (\vct{q},\vct{p})
\)
with $\vctfunc{\lambda}'$ implicitly defined by the condition $\opfunc{\Xi}^{h_2}_t(\vct{q},\vct{p}) \in \tangent^*\set{M} ~~\forall (\vct{q},\vct{p}) \in \tangent^*\set{M}$. This can be explicitly solved for $\vctfunc{\lambda}'$ to give
$
  \smash{\opfunc{\Xi}^{h_2}_t(\vct{q},\vct{p})} =
  (\bar{\vct{q}}, \opfunc{P}_{\mtx{M}}(\bar{\vct{q}}) \bar{\vct{p}})
$
with $(\bar{\vct{q}}, \bar{\vct{p}}) = \opfunc{\Phi}^{h_2}_t \circ \opfunc{\Pi}_{\vctfunc{\lambda}} (\vct{q},\vct{p})$ as defined in \cref{eq:h2-step-partial-composition,eq:h2-step-position-constraint-condition}.

For sufficiently small $t$ and sufficiently smooth constraint functions it can be shown that there exists a locally unique solution to \eqref{eq:h2-step-position-constraint-condition} \citep{brubaker2012family,lelievre2019hybrid}. In general, though, there may be multiple or no solutions, and even if there is a unique solution the iterative solver may fail to converge. This lack of a guarantee of converging to a unique solution presents a challenge in terms of maintaining the time-reversibility of the $\opfunc{\Xi}^{h_2}_t$ step and so the overall integrator.

To enforce reversibility on $\opfunc{\Xi}^{h_2}_t$ we apply a \emph{reversibility-check transform} $R$ defined such that $R(\opfunc{\Xi}^{h_2}_t)(\vct{q}, \vct{p}) = \opfunc{\Xi}^{h_2}_t(\vct{q}, \vct{p})$ for all $(\vct{q}, \vct{p})$ where $\opfunc{\Xi}^{h_2}_{-t} \circ \opfunc{\Xi}^{h_2}_t(\vct{q}, \vct{p}) = (\vct{q}, \vct{p})$, with $(\vct{q}, \vct{p})$ values for which the condition is not met causing evaluation of $R(\opfunc{\Xi}^{h_2}_t)(\vct{q}, \vct{p})$ to raise an error. Similarly if the iterative solves in the evaluation of either the forward $\opfunc{\Xi}^{h_2}_t$ or time-reversed $\opfunc{\Xi}^{h_2}_t$ maps fail to converge an error is also raised. The map $R(\opfunc{\Xi}^{h_2}_t)$ is then by construction reversible unless an error is raised that can be suitably handled downstream by the \ac{hmc} implementation. The approach of using an explicit reversibility check in \ac{mcmc} methods using an iterative solver was first proposed by \citet{zappa2018monte} with subsequent application within the context of constrained \ac{hmc} in \citet{graham2017asymptotically} and \citet{lelievre2019hybrid}.

With the reversibility check, the map $R(\opfunc{\Xi}^{h_2}_{t})$ is guaranteed to be time-reversible if it does not raise an error. As $\smash{\opfunc{\Xi}^{h_1}_{t}}$ is also time-reversible and both maps are symplectic, the integrator $\smash{\opfunc{\Xi}^{h_1}_{\sfrac{t}{2}} \circ R(\opfunc{\Xi}^{h_2}_{t}) \circ \opfunc{\Xi}^{h_1}_{\sfrac{t}{2}}}$ defines a time-reversible symplectic map on $\tangent^*\set{M}$ whenever an error is not raised. In practice the equality conditions indicating whether the iterative solver has converged and the reversibility check is satisfied, are both relaxed to an error norm being less than tolerances $\theta_c$ and $\theta_q$ respectively. Further details of the implementation of the integrator and its relation to previous work are given in \appcref{app:integrator-implementation-details}.

\subsection{Choice of metric matrix representation $\mtx{M}$}
\label{subsec:metric-choice}

We recommend choosing $\mtx{M} = \mathrm{cov}(\rvct{q})^{-1}$ for $\rvct{q} \sim \rho$, i.e. the precision matrix under the prior $\rho$; this requires that $\rho$ has finite second-order central moments. While there is no fundamental requirement for $\mtx{M}$ to match the prior precision matrix and so a different choice of $\mtx{M}$ could be used when $\mathrm{cov}(\rvct{q})^{-1}$ is not defined, heuristically we find that the performance of the proposed methodology is improved when $\rho$ is exactly or `close to' Gaussian in all components, and this can usually be arranged by transforms of $\rvct{u}$, $\rvct{v}_0$ and $\rvct{w}_{\range{\intg{T}}}$ and corresponding reparametrisations of $(\tilde{\mu}, \vctfunc{g}_{\rvct{z}})$, $(\tilde{\nu}, \vctfunc{g}_{\rvct{x}_0})$ and $(\eta, \vctfunc{h})$. The constrained Hamiltonian dynamics in \eqref{eq:constrained-hamiltonian-dynamic} with $\mtx{M} = \mathrm{cov}(\rvct{q})^{-1}$ are equivalent to the dynamics under a linear transform $\rvct{q}' = \mtx{L}\tr\rvct{q}$ with $\mtx{L}\mtx{L}\tr = \mtx{M}$ for which $\mathrm{cov}(\rvct{q}') = \idmtx_{\intg{Q}}$ and normalising for the prior scales and correlations in this manner appears to improve the robustness and efficiency of the algorithm.

\subsection{Overall algorithm}

\begin{algorithm}[!t]
  \caption{Hamiltonian Monte Carlo with a constrained symplectic integrator.}
  \label{alg:constrained-hmc}
  {
\begin{algorithmic}
  \Require (reasonable default values for parameters are given in parenthesis)\\
  \setlength\multicolsep{0pt}
  \begin{multicols}{2}
    $\vct{q}$ : current state with $\Vert\vctfunc{c}(\vct{q})\Vert < \theta_c$,\\
    $\theta_c$ : constraint tolerance ($10^{-9}$),\\
    $\theta_q$: position change tolerance ($10^{-8}$),\\
    $\timestep$ : integrator time step,\\
    $\intg{I}$ : number of integrator steps / sample,\\
    $\intg{J}$ : maximum Newton iterations ($50$).
  \end{multicols}
  \Ensure\\
  $\vct{q}'$ : next state with $\Vert\vctfunc{c}(\vct{q}')\Vert < \theta_c$ and if $\vct{q} \sim \pi \implies \vct{q}' \sim \pi$.
\end{algorithmic}
\vspace{1mm}
\hrule
\vspace{-3mm}
\setlength{\columnsep}{6pt}
\begin{multicols}{2}
\ifclass{statsoc}{\footnotesize}{}%
\begin{algorithmic}[1]
    \algrenewcommand\algorithmicindent{1em}
    \Function{$\opfunc{\Xi}^{h_1}$}{$\vct{q}$, $\vct{p}$, $\timestep$}
        \State \Return $(\vct{q}, \opfunc{P}_\mtx{M}(\vct{q})(\vct{p} - t \grad h_1(\vct{q})))$
    \EndFunction
    \vspace{4pt}
    \Function{$\opfunc{\Xi}^{h_2}$}{$\vct{q}$, $\vct{p}$, $\timestep$}
        \State $(\vct{\lambda},\bar{\vct{q}}) \gets (\vct{0},\vct{q})$
        \For{$\intg{j} \in \range{\intg{J}}$}
            \State $(\vct{q}',\vct{p}') \gets \opfunc{\Phi}^{h_2}_\timestep(\vct{q},\vct{p} - \jacob\vctfunc{c}(\vct{q})\tr\vct{\lambda})$
            \State $\vct{e} \gets \vctfunc{c}(\vct{q}')$\label{ln:evaluate-constraint-function}
            \If{$\Vert\vct{e}\Vert < \theta_c$ \textbf{and} $\Vert\vct{q}' - \bar{\vct{q}}\Vert < \theta_q$\label{ln:convergence-check}}
                \State \Return $(\vct{q}', \opfunc{P}_\mtx{M}(\vct{q'})\vct{p}')$
            \EndIf
            \vspace{-2pt}
            \State $\vct{\lambda} \gets \vct{\lambda} + (\jacob\vctfunc{c}(\vct{q}')(\opfunc{\Gamma}^{q,p}_t)^{-1}\jacob\vctfunc{c}(\vct{q})\tr)^{-1}\vct{e}$\label{ln:lambda-update}
            \State $\bar{\vct{q}} \gets \vct{q}'$
        \EndFor
        \State \Raise \textsc{IntegratorError} \label{ln:non-convergence}
    \EndFunction
    \vspace{4pt}
    \Function{Reversible$\opfunc{\Xi}^{h_2}$}{$\vct{q}$, $\vct{p}$, $\timestep$}
        \State $(\vct{q}', \vct{p}') \gets \Call{$\opfunc{\Xi}^{h_2}$}{\vct{q}, \vct{p}, \timestep}$
        \State $(\vct{q}_{r}, \vct{p}_r) \gets \Call{$\opfunc{\Xi}^{h_2}$}{\vct{q}', \vct{p}', -\timestep}$
        \If{$\left\Vert \vct{q} - \vct{q}_{r} \right\Vert_{\infty} > 2\theta_q$}
            \State \Raise \textsc{IntegratorError} \label{ln:reverse-check}
        \EndIf
        \State \Return $(\vct{q}', \vct{p}')$
    \EndFunction
    \columnbreak
    \Function{ConstrStep}{$\vct{q}$, $\vct{p}$, $\timestep$}
        \State $(\vct{q}, \vct{p}) \gets \Call{$\opfunc{\Xi}^{h_1}$}{\vct{q}, \vct{p}, \timestep / 2}$
        \State $(\vct{q}, \vct{p}) \gets \Call{Reversible$\opfunc{\Xi}^{h_2}$}{\vct{q}, \vct{p}, \timestep}$
        \State \Return $\Call{$\opfunc{\Xi}^{h_1}$}{\vct{q}, \vct{p}, \timestep / 2}$
    \EndFunction
    \vspace{4pt}
    \State $\tilde{\vct{p}} \sim \nrm(\vct{0},\mtx{M})$
    \State $\vct{p} \gets \opfunc{P}_{\mtx{M}}(\vct{q})\tilde{\vct{p}}$
    \State $(\vct{q}', \vct{p}') \gets (\vct{q}, \vct{p})$
    \Try
        \For{$\intg{i}\in\range{\intg{I}}$} \label{ln:main-chain-loop}
            \State $(\vct{q}', \vct{p}') \gets \Call{ConstrStep}{\vct{q}', \vct{p}', t}$
        \EndFor
        \State $u \sim \mathcal{U}(0,1)$
        \If{$u > \exp(h(\vct{q},\vct{p}) - h(\vct{q}',\vct{p}'))$}
          \State $(\vct{q}',\vct{p}') \gets (\vct{q}, -\vct{p})$
        \EndIf
    \Catch{\textsc{IntegratorError}}
        \State $(\vct{q}', \vct{p}') \gets (\vct{q}, -\vct{p})$
    \EndTry
\end{algorithmic}
\end{multicols}
\vspace{-4mm}   }
\end{algorithm}

Pseudo-code corresponding to applying the reversible, constraint-preserving and symplectic integrator with step $\opfunc{\Xi}^{h_1}_{\sfrac{\timestep}{2}} \circ R(\opfunc{\Xi}^{h_2}_{\timestep}) \circ \opfunc{\Xi}^{h_1}_{\sfrac{\timestep}{2}}$ within a \ac{hmc} algorithm is summarised in \cref{alg:constrained-hmc}. Any errors raised when integrating the trajectory by iteratively applying the \textsc{ConstrStep} function are handled by terminating the trajectory and the \ac{hmc} transition returning the initial state i.e.~a `rejection'. Although for simplicity we have described in \cref{alg:constrained-hmc} the use of a constraint-preserving integrator within a Metropolis-adjusted \ac{hmc} algorithm with a static integration time $\intg{I}\timestep$ per chain iteration, in practice we use a \ac{hmc} algorithm which dynamically adapts the integration time $\intg{I}t$, in particular the dynamic multinomial \ac{hmc} algorithm described in the appendix of \citet{betancourt2017conceptual}, an extension of the \emph{No-U-Turn-Sampler} algorithm \citep{hoffman2014no}. We also use the dual-averaging scheme of \citet{hoffman2014no} to adaptively tune the integrator step-size $t$ in a \emph{warm-up} sampling phase to target an acceptance statistic of 0.8. A general purpose implementation of the full algorithm is provided in Python package \emph{Mici} \citep{graham2019mici}, which we use in the numerical experiments in \cref{sec:numerical-experiments}.

We have found the suggested defaults values for the various algorithmic parameters work well in practice for a range of different models. This therefore results in an automated methodology with a practitioner only needing to specify functions to evaluate the log prior density $\log \inlinetd{\rho}{\lebm{\intg{Q}}}$ and constraint function $\vctfunc{c}$ for the diffusion model in question, with the required derivatives of these functions being able to be constructed algorithmically \citep{griewank2008evaluating}. %
\section{Computational cost}
\label{sec:computational-cost}

We can apply \cref{alg:constrained-hmc} to perform inference in partially-observed diffusion models by targeting the manifold-supported posterior distribution in the non-centred parametrisation of the time-discretised model described in \cref{subsec:posterior-distribution-non-centred}. While this approach allows significant generality in the choice of the elements of the diffusion and observation model, it can be computationally expensive to run. To analyse the cost of \cref{alg:constrained-hmc} in this setting we make the following simplifying assumption.
\begin{assumption}
  \label{ass:comp-cost-newton-iterations}
  The Newton iteration to solve \eqref{eq:h2-step-position-constraint-condition} converges within $\intg{J}$ iterations for some $\intg{J} > 0$ that does not depend on $\intg{S}$ and $\intg{T}$, for fixed $t$, $\theta_c$ and $\theta_q$.
\end{assumption}
This assumption appears to hold in practice, and we provide numerical evidence to this effect in the numerical experiments in \cref{sec:numerical-experiments}. We then have the following.
\begin{proposition}
  \label{the:comp-cost-without-conditioning}
  Under \cref{ass:comp-cost-newton-iterations} the computational cost of a single constrained integrator step in \cref{alg:constrained-hmc} when directly applied to the posterior density \eqref{eq:manifold-density} of the generative model in \cref{alg:non-centred-generative-model} is $\bigo(\intg{S}\intg{T}^3)$.
\end{proposition}
A proof is given in \appcref{app:comp-cost-without-conditioning}. The cost of \cref{alg:constrained-hmc} when applied directly to the posterior distribution with density in \eqref{eq:manifold-density} therefore scales linearly with the number of discrete time steps per observation $\intg{S}$ but cubically with the number of observation times $\intg{T}$.

\subsection{Exploiting Markovianity for scalability}
\label{subsec:exploiting-markovianity}

While we have so far considered only sampling from the posterior distribution on latent variables $(\rvct{u}, \rvct{v}_0, \rvct{v}_{\range{\intg{S}\intg{T}}}, \rvct{w}_{\range{\intg{T}}})$ given observations $\rvct{y}_{\range{\intg{T}}}$, the constrained \ac{hmc} approach we have described can be applied to sampling from any conditional distribution of the joint distribution on observations and latent variables under the generative model, of which the target posterior distribution is just one example.

One way to improve the scaling of the computational cost with respect to the number of observation times is therefore to restrict the information flow through the state sequence $\rvct{x}_{\range{\intg{S}\intg{T}}}$ by conditioning on a set of intermediate states in the sequence. Due to the Markovian nature of the state dynamics, the state sequences $\rvct{x}_{\range[0]{\intg{s}-1}}$ and $\rvct{x}_{\range[\intg{s}+1]{\intg{S}\intg{T}}}$ are conditionally independent given the state $\rvct{x}_{\intg{s}}$ and the parameters $\rvct{z}$ for any $\intg{s} \in \range{\intg{ST}}$. As a consequence under the non-centred parametrisation of the generative model in \cref{alg:non-centred-generative-model}, we have that if we condition on the intermediate state $\rvct{x}_{\intg{S}\intg{t}}$ at the $\intg{t}$\textsuperscript{th} observation time then we can independently generate the observation sequences $\rvct{y}_{\range{\intg{t}}}$ and $\rvct{y}_{\range[\intg{t}+1]{\intg{T}}}$ from respectively $(\rvct{u},\rvct{v}_{\range[0]{\intg{S}\intg{t}}}, \rvct{w}_{\range{\intg{t}}})$ and $(\rvct{u},\rvct{v}_{\range[\intg{S}\intg{t}+1]{\intg{S}\intg{T}}},\rvct{w}_{\range[\intg{t+1}]{\intg{T}}})$.

We can extend this idea to conditioning on multiple intermediate states in the sequence. If at $\intg{B}-1$ observation time indices $\intg{t}_{\range{\intg{B}-1}} \subseteq\range{\intg{T}}$ the full state is conditioned on, $\rvct{x}_{\intg{S}\intg{t}_\intg{b}} = \vct{x}_{\intg{S}\intg{t}_\intg{b}} ~\forall \intg{b} \in \range{\intg{B}-1}$, then we have that the observation subsequence $\smash{\rvct{y}_{\intg{t}_\intg{b-1}+1:\intg{t}_\intg{b}}}$ and conditioned state $\rvct{x}_{\intg{S}\intg{t}_\intg{b}}$ depend only on the latent variables $(\rvct{u},\rvct{v}_{\range[\intg{S}\intg{t}_{\intg{b}-1}+1]{\intg{S}\intg{t}_\intg{b}}},\rvct{w}_{\range[\intg{t}_{\intg{b}-1}+1]{\intg{t}_\intg{b}}})$ for each $\intg{b} \in \range[1]{\intg{B}-1}$ (with $\intg{t}_0 = 0$) and the final observation subsequence $\smash{\rvct{y}_{\intg{t}_\intg{B-1}+1:\intg{T}}}$ depends only on the latent variables  $\smash{(\rvct{u},\rvct{v}_{\range[\intg{S}\intg{t}_{\intg{B}-1}+1]{\intg{S}\intg{T}}},\rvct{w}_{\range[\intg{t}_{\intg{B}-1}+1]{\intg{T}}})}$.
Due to these conditional independencies introduced when conditioning on the values of $(\rvct{x}_{\intg{S}\intg{t}_\intg{b}})_{\intg{b}=1}^{\intg{B}-1}$, we can `split' the generation of the state sequence in to $\intg{B}$ independent calls to a function which given a conditioned state $\vct{x}_{\intg{S}\intg{t}_{\intg{b}-1}}$ generates the subsequence of states for step indices $\range[\intg{S}\intg{t}_{\intg{b}-1}+1]{\intg{S}\intg{t}_\intg{b}}$ and outputs the observations $\smash{\rvct{y}_{\range[\intg{t}_{\intg{b}-1}+1]{\intg{t}_\intg{b}}}}$ and final state $\rvct{x}_{\intg{S}\intg{t}_{\intg{b}}}$ of the subsequence (or just observations $\smash{\rvct{y}_{\range[\intg{t}_{\intg{B}-1}+1]{\intg{T}}}}$ for the final subsequence). For noiseless observations, $\rvct{y}_{\intg{t}_\intg{b}}$ is completely determined by $\rvct{x}_{\intg{S}\intg{t}_{\intg{b}}}$, and so only $\smash{\rvct{y}_{\range[\intg{t}_{\intg{b}-1}+1]{\intg{t}_\intg{b}}-1}}$ and $\rvct{x}_{\intg{S}\intg{t}_{\intg{b}}}$ should be returned for the non-final subsequences. The resulting conditioned generative model is summarised in \cref{alg:conditioned-generative-model}.

\begin{model}[t]
  \caption{Generative model conditioning on intermediate states $(\vct{x}_{\intg{S}\intg{t}_\intg{b}})_{\intg{b}=1}^{\intg{B}-1}$.}
  \label{alg:conditioned-generative-model}
\setlength\multicolsep{0pt}
\begin{multicols}{2}
  \begin{algorithmic}
    \algrenewcommand\algorithmicindent{1.3em}
    \Function{$\vctfunc{g}_{\bar{\rvct{y}}}$}
      {$\vct{u}$, $\vct{v}_0$, $\vct{v}_{\range{\intg{St}}}$, $\vct{w}_{\range{\intg{t}}}$, $\intg{b}$}
      \State $\vct{z} \gets \vctfunc{g}_{\rvct{z}}(\vct{u})$
      \State $\vct{x}_0 \gets$ \IfThenElse{$\intg{b}\equiv 1$}{$\vctfunc{g}_{\rvct{x}_0}\mkern-3mu\left(\vct{z},\vct{v}_0\right)$}{$\vct{v}_0$}
      \State $\vct{x}_{\range{\intg{St}}}$, $\vct{y}_{\range{\intg{t}}} \gets \vctfunc{g}_{\rvct{x}_{:},\rvct{y}_{:}}\mkern-3mu(\vct{z},\vct{x}_0,\vct{v}_{\range{\intg{St}}},\vct{w}_{\range{\intg{t}}})$
      \If{$\intg{b}\neq\intg{B}$}
       \State \Return $[\vct{y}_{\range{\intg{t}}}; \vct{x}_{\intg{St}}]$
      \Else
       \State \Return $[\vct{y}_{\range{\intg{t}}}]$
      \EndIf
    \EndFunction
    \columnbreak
    \State $\rvct{u} \sim \tilde{\mu}$
    \State $\rvct{v}_0 \sim \tilde{\nu}$
    \State $\rvct{v}_{\intg{s}} \sim \nrm(\vct{0},\idmtx_\intg{V}) ~\forall \intg{s} \in \range{\intg{S}\intg{T}}$
    \State $\rvct{w}_{\intg{t}} \sim \eta ~\forall \intg{t} \in \range{\intg{T}}$
    \State $\bar{\rvct{y}}_{1} \gets \vctfunc{g}_{\bar{\rvct{y}}}(\rvct{u}, \rvct{v}_0, \rvct{v}_{\range[1]{\intg{S}\intg{t}_\intg{1}}}, \rvct{w}_{\range{\intg{t}_1}}, 1)$
    \For{$\intg{b} \in \range[2]{\intg{B}}$} \Comment{$\intg{t}_\intg{B} \equiv \intg{T}$}
      \State $\rvct{v}_:, \rvct{w}_: \gets \rvct{v}_{\range[\intg{S}\intg{t}_{\intg{b}-1}+1]{\intg{S}\intg{t}_\intg{b}}},\rvct{w}_{\range[\intg{t}_{\intg{b}-1}+1]{\intg{t}_\intg{b}}}$
      \State $\bar{\rvct{y}}_{\intg{b}} \gets \vctfunc{g}_{\bar{\rvct{y}}}
      (\rvct{u}, \vct{x}_{\intg{S}\intg{t}_{\intg{b}-1}},\rvct{v}_:, \rvct{w}_:,\intg{b})$
    \EndFor
  \end{algorithmic}
\end{multicols}
 \end{model}

Using $\vctfunc{g}_{\bar{\rvct{y}}}$ from \cref{alg:conditioned-generative-model} we can then define \emph{partial} constraint functions%
\begin{align*}
  \vctfunc{c}_1(\vct{u}, [\vct{v}_0; \vct{v}_{\range{\intg{S}\intg{t}_1}}], [\vct{w}_{\range{\intg{t}_\intg{1}}}]) &:=
  \vctfunc{g}_{\bar{\rvct{y}}}(\vct{u}, \vct{v}_0, \vct{v}_{\range{\intg{S}\intg{t}_1}}, \vct{w}_{\range{\intg{t}_\intg{1}}}, 1) -
  \bar{\vct{y}}_1,~\text{and }\forall\,\intg{b}\in\range[2]{\intg{B}}\\
  \vctfunc{c}_\intg{b}(\vct{u}, [\vct{v}_{\range[\intg{S}\intg{t}_{\intg{b}-1}+1]{\intg{S}\intg{t}_\intg{b}}}], [\vct{w}_{\range[\intg{t}_\intg{b-1}+1]{\intg{t}_\intg{b}}}]) &:=
  \vctfunc{g}_{\bar{\rvct{y}}}(\vct{u}, \vct{x}_{\intg{S}\intg{t}_{\intg{b}-1}}, \vct{v}_{\range[\intg{S}\intg{t}_{\intg{b}-1}+1]{\intg{S}\intg{t}_\intg{b}}}, \vct{w}_{\range[\intg{t}_\intg{b-1}+1]{\intg{t}_\intg{b}}}, \intg{b}) -
  \bar{\vct{y}}_\intg{b}
\end{align*}
with $\bar{\vct{y}}_\intg{b} = [\vct{y}_{\intg{t}_\intg{b-1}+1:\intg{t}_\intg{b}}; \vct{x}_{\intg{S}\intg{t}_\intg{b}}] ~\forall\,\intg{b}\in\range{\intg{B}-1}$ and $\bar{\vct{y}}_{\intg{B}} = [\vct{y}_{\range[\intg{t}_{\intg{B}-1}+1]{\intg{T}}}]$. We then define respectively \emph{partitioned} and \emph{full} constraint functions $\bar{\vctfunc{c}} : \reals^{\intg{U}} \times \reals^{\intg{V}_0+\intg{STV}} \times \reals^{\intg{TW}} \to \reals^\intg{C}$ and $\vctfunc{c} : \reals^{\intg{Q}} \to \reals^{\intg{C}}$, with $\intg{C} = (\intg{B} - 1)\intg{X} + \intg{T}\intg{Y}$ the number of constraints, as%
\begin{equation}
  \bar{\vctfunc{c}}\left(
    \vct{u}, [\bar{\vct{v}}_{\range{\intg{B}}}], [\bar{\vct{w}}_{\range{\intg{B}}}]
  \right) :=
  \vctfunc{c}\left([
    \vct{u}; \bar{\vct{v}}_{\range{\intg{B}}}; \bar{\vct{w}}_{\range{\intg{B}}}]
  \right) =
  \left[
    \left(
      \vctfunc{c}_\intg{b}(\vct{u}, \bar{\vct{v}}_\intg{b}, \bar{\vct{w}}_\intg{b})
    \right)_{\intg{b}\in\range{\intg{B}}}
  \right].
\end{equation}
The Jacobian of the full constraint function will then have the block structure \ifclass{statsoc}{\(}{\[}%
  \jacob\vctfunc{c}([\vct{u}; \bar{\vct{v}}; \bar{\vct{w}}]) =
  \begin{bmatrix}
    \jacob_1\bar{\vctfunc{c}}(\vct{u},\bar{\vct{v}},\bar{\vct{w}}) &
    \jacob_2\bar{\vctfunc{c}}(\vct{u},\bar{\vct{v}},\bar{\vct{w}}) &
    \jacob_3\bar{\vctfunc{c}}(\vct{u},\bar{\vct{v}},\bar{\vct{w}})
  \end{bmatrix}
\ifclass{statsoc}{\)}{\]}%
with $\jacob_1\bar{\vctfunc{c}}(\vct{u},\bar{\vct{v}},\bar{\vct{w}})$ a dense $\intg{C} \times \intg{U}$ matrix, and $\jacob_\intg{i}\bar{\vctfunc{c}}(\vct{u}, \bar{\vct{v}}, \bar{\vct{w}})$ for $\intg{i}\in\lbrace 2, 3\rbrace$ block diagonal $\intg{C} \times (\intg{V}_0 + \intg{S}\intg{T}\intg{V})$ ($\intg{i} = 2$) and $\intg{C} \times \intg{TW}$ ($\intg{i} = 3$) matrices with $\jacob_\intg{i}\bar{\vctfunc{c}}(\vct{u}, [\bar{\vct{v}}_{\range{\intg{B}}}], [\bar{\vct{w}}_{\range{\intg{B}}}]) = \diag(\jacob_\intg{i}\vctfunc{c}_{\intg{b}}(\vct{u},\bar{\vct{v}}_{\intg{b}},\bar{\vct{w}}_{\intg{b}}))_{\intg{b}\in\range{\intg{B}}}$.

As $\rvct{u}$, $[\rvct{v}_{\range[0]{\intg{ST}}}]$ and $[\rvct{w}_{\range{\intg{T}}}]$ are independent under the prior $\rho$, under the recommendation in \cref{subsec:metric-choice} the metric matrix is $\mtx{M} = \diag(\mtx{M}_u,\mtx{M}_v,\mtx{M}_w)$ with $\mtx{M}_u$ a $\intg{U}\times\intg{U}$ matrix, $\mtx{M}_v$ a $(\intg{V}_0 + \intg{S}\intg{T}\intg{V})\times(\intg{V}_0 + \intg{S}\intg{T}\intg{V})$ block-diagonal matrix and $\mtx{M}_w$ a $\intg{T}\intg{W}\times\intg{T}\intg{W}$ block-diagonal matrix. The Gram matrix can then be decomposed as
\begin{align}
  \label{eq:conditioned-gram-matrix}
  &\qquad\quad
  \mtx{G}_{\mtx{M}}([\vct{u};\bar{\vct{v}};\bar{\vct{w}}]) =
  \jacob_1\bar{\vctfunc{c}}(\vct{u},\bar{\vct{v}},\bar{\vct{w}})
  \mtx{M}_u^{-1}
  \jacob_1\bar{\vctfunc{c}}(\vct{u},\bar{\vct{v}},\bar{\vct{w}})\tr +
  \opfunc{D}([\vct{u};\bar{\vct{v}};\bar{\vct{w}}]) \text{ with}
  \\
  \nonumber
  &
  \opfunc{D}([\vct{u};\bar{\vct{v}};\bar{\vct{w}}]) := \jacob_2\bar{\vctfunc{c}}(\vct{u},\bar{\vct{v}},\bar{\vct{w}})\mtx{M}_v^{-1}\jacob_2\bar{\vctfunc{c}}(\vct{u},\bar{\vct{v}},\bar{\vct{w}})\tr + \jacob_3\bar{\vctfunc{c}}(\vct{u},\bar{\vct{v}},\bar{\vct{w}})\mtx{M}_w^{-1}\jacob_3\bar{\vctfunc{c}}(\vct{u},\bar{\vct{v}},\bar{\vct{w}})\tr,
\end{align}
corresponding to a rank $\intg{U}$ correction of a block-diagonal matrix $\opfunc{D}([\vct{u};\bar{\vct{v}};\bar{\vct{w}}])$.

Using the matrix determinant lemma we then have that
\begin{equation}
  \label{eq:conditioned-gram-determinant}
  \log\det{
    \mtx{G}_{\mtx{M}}(\vct{q})
  } =
  \log\det{\opfunc{C}(\vct{q})} +
  \log\det{\opfunc{D}(\vct{q})} -
  \log\det{\mtx{M}_u},
\end{equation}
with  $\opfunc{C}([\vct{u};\bar{\vct{v}};\bar{\vct{w}}]) := \mtx{M}_u + \jacob_1\bar{\vctfunc{c}}(\vct{u},\bar{\vct{v}},\bar{\vct{w}})\tr \opfunc{D}([\vct{u};\bar{\vct{v}};\bar{\vct{w}}])^{-1} \jacob_1\bar{\vctfunc{c}}(\vct{u},\bar{\vct{v}},\bar{\vct{w}})$.
Similarly, the Woodbury matrix identity yields, for a vector $\vct{r} \in \reals^{\intg{C}}$ and $\vct{q} = [\vct{u};\bar{\vct{v}};\bar{\vct{w}}]$, that
\begin{equation}\label{eq:conditioned-gram-inverse}
  \opfunc{G}_{\mtx{M}}(\vct{q})^{-1}\vct{r} =
  \opfunc{D}(\vct{q})^{-1}
  (\vct{r} -
  \jacob_1\bar{\vctfunc{c}}(\vct{u},\bar{\vct{v}},\bar{\vct{w}})
  \opfunc{C}(\vct{q})^{-1}\jacob_1\bar{\vctfunc{c}}(\vct{u},\bar{\vct{v}},\bar{\vct{w}})\tr
  \opfunc{D}(\vct{q})^{-1}\vct{r}).
\end{equation}
By applying a sequence of constrained \ac{hmc} Markov kernels, each conditioning on intermediate states $(\rvct{x}_{\intg{S}\intg{t}_\intg{b}})_{\intg{b}=1}^{\intg{B}-1}$ at a different set of observation time indices $\intg{t}_{\range{\intg{B}-1}}$ as well as the observations $\rvct{y}_{\range{\intg{T}}}$ we can construct a \ac{mcmc} method which asymptotically samples from the target posterior distribution at a substantially reduced computational cost compared to the case of conditioning only on the observations $\rvct{y}_{\range{\intg{T}}}$. To analyse the computational cost of applying the constrained \ac{hmc} implementation in \cref{alg:constrained-hmc} to the conditioned generative model, we assume the following.
\begin{assumption}
  \label{ass:comp-cost-conditioning-partition}
  $\intg{T} = \intg{B}\intg{R}$ and $\intg{t}_{\intg{b}} = \intg{b}\intg{R}~\forall\,\intg{b}\in\range{\intg{B}-1}$, i.e. that the observations are split in to $\intg{B}$ equally sized subsequences of $\intg{R}$ observation times.
\end{assumption}
\begin{assumption}
  \label{ass:comp-cost-newton-iterations-conditioned}
  The Newton iteration to solve \eqref{eq:h2-step-position-constraint-condition} converges within $\intg{J}$ iterations for $\intg{J}> 0$ that does not depend on $\intg{R}$, $\intg{S}$ and $\intg{T}$, for fixed $t$, $\theta_c$ and $\theta_q$.
\end{assumption}
In practice we will need to alternate with conditioning on a different set of observation times to allow the Markov chain to be ergodic, e.g. $\intg{t}'_{\intg{b}} = \floor{\frac{(2\intg{b}-1)\intg{R}}{2}}~\forall\,\intg{b}\in\range{\intg{B}}$, with in this case the observation times split in to $\intg{B-1}$ subsequences of $\intg{R}$ observations times and two subsequences of $\floor{\frac{\intg{R}}{2}}$ and $\ceil{\frac{\intg{R}}{2}}$ observation times. This alternative splitting will result in only minor difference in operation cost compared to the equispaced partition hence we consider only the former case in our analysis. \cref{ass:comp-cost-newton-iterations-conditioned} is motivated by our observation that the average number of Newton iterations needed for convergence appears to be independent of $\intg{R}$, $\intg{S}$ and $\intg{T}$.%
\begin{proposition}
  \label{the:comp-cost-with-conditioning}
  Under \cref{ass:comp-cost-conditioning-partition,ass:comp-cost-newton-iterations-conditioned} the computational cost of a single constrained integrator step in \cref{alg:constrained-hmc} when applied to the posterior of the generative model conditioning additionally on the values of the states at observation time indices $\intg{t}_{\range{\intg{B}-1}}$ as in \cref{alg:conditioned-generative-model} is $\bigo(\intg{R}^2\intg{S}\intg{T})$ operations.
\end{proposition}
A proof is given in \appcref{app:comp-cost-with-conditioning}. If the number of observations per subsequence $\intg{R}$ is kept fixed, the computational cost of each constrained integrator step therefore scales linearly with in both the number of time steps per observation $\intg{S}$ and the number of observation times $\intg{T}$.
\section{Related work}
\label{sec:related-work}

Our approach follows the general framework described in \cite{graham2017asymptotically} for performing inference in generative models where the simulated observations can be computed as a differentiable function of a random vector with a known prior distribution. As in this work, \citet{graham2017asymptotically} suggest using a constrained \ac{hmc} algorithm to target the manifold-supported posterior distribution arising in such a setting, and consider a diffusion model with high-frequency noiseless observations of the full state as an example. In this setting with $\intg{S} = 1$ the constraint Jacobian was observed to have a structure allowing a $\mathcal{O}(\intg{T}^2)$ cost implementation of the operations required for each constrained integrator step.

Here we make several important extensions to the framework of \citet{graham2017asymptotically}, with the scheme proposed in \cref{subsec:exploiting-markovianity} allowing efficient $\mathcal{O}(\intg{ST})$ cost constrained integrator steps irrespective of the observation regime (high- or low-frequency, partial or full, with or without noise) and the use of a constrained integrator based on a Gaussian splitting as proposed in \cref{subsec:constrained-hamiltonian-dynamics-numerical-discretisation} giving improved mixing performance as the time-discretisation is refined ($\intg{S}$ increased). Further by integrating the constrained integrator into an adaptive \ac{hmc} algorithm \citep{hoffman2014no} we eliminate the need to tune the integrator step size and number of integrator steps per trajectory, giving a more automated inference procedure.

The non-centred parametrisation of the diffusion generative model described in \cref{subsec:non-centred-parametrisation} has similarities to the \emph{innovation scheme} of \citet{chib2004likelihood}, and its later extension in \citet{golightly2008bayesian}, which recognises for the specific case of an Euler--Maruyama discretisation, that the state sequence $\rvct{x}_{\range{\intg{ST}}}$ can be computed as a function of the model parameter $\rvct{z}$, initial state $\rvct{x}_0$ and increments of the driving Brownian motion process. This relationship can be inverted to compute the increments given $\rvct{x}_{\range[0]{\intg{ST}}}$ and $\rvct{z}$.  By performing a Metropolis-within-Gibbs update to $\rvct{z}$ conditioning on the increments and observations, the degeneracy in the conditional distribution of parameters of the drift coefficient when instead conditioning on $\rvct{x}_{\range{\intg{ST}}}$ as $\intg{S} \to \infty$ is avoided, thus producing an algorithm respecting the Roberts--Stramer critique. Our approach generalizes this idea beyond the Euler--Maruyama case by allowing for a generic forward operator $\vctfunc{f}_\delta$, and jointly updated all latent variables under this reparametrisation rather than using it to only update the parameters.

The conditioning scheme proposed in \cref{subsec:exploiting-markovianity} is similar in spirit to \emph{blocking} schemes proposed previously in \ac{mcmc} methods for inference in partially-observed time series models, see e.g. \citet{shephard1997likelihood,golightly2008bayesian,mider2020computational}, however, the implementation and motivation of the approach here both differ. In the blocking schemes, conditioning on intermediate states introduces conditional independencies allowing proposing updates to blocks of the latent path given fixed parameters in a Metropolis-within-Gibbs type update, with a separate update to the parameters. Here we jointly update the parameters and latent path, and use the conditioning to induce structure in the constraint Jacobian which can be used to reduce the cost of the constrained integrator.

Hypoelliptic diffusions have a rank-deficient diffusion matrix $\opfunc{B}(\rvct{x},\rvct{z})\opfunc{B}(\rvct{x},\rvct{z})\tr$, but still have transition kernels $\kappa_\tau$ with smooth densities with respect to the Lebesgue measure due to the propagation of noise to all state components via the drift function. Prior work on the calibration of such models has often adopted a maximum likelihood approach, in the setting of high-frequency observations, see e.g.~\citet{ditlevsen2019hypoelliptic} and the references therein. The singularity of the Wiener noise increment covariance matrix when discretising using an Euler--Maruyama scheme can be avoided via the use of a higher-order discretisation scheme: \citet{ditlevsen2019hypoelliptic} use a strong order 1.5 Taylor scheme to obtain consistency in the estimation of parameters in both the drift and diffusion coefficient functions.

In terms of our criteria, in Remark \ref{rem:criteria}, to our knowledge there is currently no alternative algorithm that satisfies them all for noiselessly observed hypoelliptic systems. The \emph{guided proposals} framework \citep{meulen2018bayesian,bierkens2018simulation,meulen2020continuous} comes close, as it allows for Bayesian inference in both elliptic and hypoelliptic systems, fully or partially observed with noise or with noiseless observations and a linear observation function $\vcthbar(\vct{x}) = \mtx{L}\vct{x}$, and respects the Roberts--Stramer critique. The approach however does not allow for non-linear noiseless observations, and the methodology requires choosing a tractable auxiliary process used to construct the proposed updates to the latent path given observations and parameters, with the original and auxiliary processes needing to satisfy \emph{matching conditions} on their drift and diffusion coefficients, which can be non-trivial --- for example, when the diffusion coefficient is state dependent --- hindering the automation of the methodology. In contrast, our method can be applied without change to both hypoelliptic and elliptic diffusions.

A long line of previous work has considered \ac{mcmc} methods for performing inference in distributions on non-Euclidean spaces, particularly prominent being the influential paper \cite{girolami2011riemann} where the latent space is equipped with a user-defined Riemannian metric which facilitates local rescaling of the posterior distribution across different directions. Related algorithms have also been proposed based for finite-dimensional projections of distributions with densities with respect to Gaussian measures on Hilbert spaces \citep{beskos2014stable,beskos2017geometric}.

In our case the manifold structure arises directly from the observational constraints placed on the latent space of a generative model and the manifold is extrinsically defined by its embedding in an ambient latent space. Rather than the non-trivial task of selecting a positive-definite matrix valued function to define a Riemannian metric on the latent space, our method only requires the user to choose a matrix representing the fixed metric on the ambient space. As discussed in \cref{subsec:metric-choice} we find the prior precision matrix to be a good default in practice.
\section{Numerical examples}
\label{sec:numerical-experiments}

To demonstrate the flexibility and efficiency of our proposed approach we now present the results of numerical experiments in a range of different settings: hypoelliptic and elliptic systems, simulated and real data, noiseless and noisy observations. In all cases we use the same methodology, as described in the preceding sections, for performing inference, and where possible we compare to alternative approaches.

\subsection{FitzHugh--Nagumo model with noiseless observations}
\label{subsec:fhn-noiseless-experiments}

As a first example we consider a stochastic-variant of the FitzHugh--Nagumo model \citep{fitzhugh1961impulses,nagumo1962active}, a simplified description of the dynamics of action potential generation within an neuronal axon. Following \citet{ditlevsen2019hypoelliptic} we formulate the model as a $\intg{X} = 2$ dimensional hypoelliptic diffusion process $\rvct{x}$ with drift function $\vctfunc{a}(\vct{x},\vct{z}) = [\frac{1}{\epsilon}(x_1 - x_1^3 - x_2); \, \gamma x_1 - x_2 + \beta]$ and diffusion coefficient operator $\opfunc{B}(\vct{x},\vct{z}) = [0; \sigma]$ where the components of the $\intg{Z}= 4$ dimensional parameter vector are $\rvct{z} = [\sigma; \epsilon; \gamma; \beta]$ and the Wiener process $\rvar{b}$ has dimension $\intg{B} = 1$. We initially assume the $\intg{Y} = 1$ dimensional observations $\rvar{y}_{\range{\intg{T}}}$ correspond to noiseless observation of the first state component i.e.~$\vcthbar(\vct{x}) = x_1 $, with an inter-observation time interval $\upDelta = \frac{1}{5}$. Further details of the discretisation and priors used are given in the \appcref{app:fitzhugh-nagumo-model-details}.

To measure sampling efficiency in the experiments we use two complementary metrics: the average computation time per constrained integrator step $\hat{\tau}_{\textrm{step}}$ and the estimated computation time per effective sample $\hat{\tau}_{\textrm{eff}}$, i.e.~the total chain computation time divided by the estimated \ac{ess} for the chain for each parameter. \cref{the:comp-cost-with-conditioning} closely relates to $\hat{\tau}_{\textrm{step}}$, and so by estimating how this quantity varies with $\intg{R}$, $\intg{S}$ and $\intg{T}$ we can empirically test whether the proposed scaling holds in practice. While our analysis only considers the computational cost of the constrained integrator, ultimately we are interested in the overall sampling efficiency, which also depends on the mixing performance of the chains; by measuring $\hat{\tau}_{\textrm{eff}}$ we therefore also gain insight into how our approach performs on this metric. In order to empirically verify that Assumption \ref{ass:comp-cost-newton-iterations-conditioned} holds in practice we additionally record the average number of Newton iterations per constrained integrator step (averaged over both forward and time-reversed $\opfunc{\Xi}^{h_2}$ calls) which we denote $\bar{\intg{n}}$.

The \ac{ess} estimates were computed using the Python package \emph{ArviZ} \citep{arviz2019} using the rank-normalisation approach proposed by \citet{vehtari2019rank}. The chain computation times were measured by counting the calls of the key expensive operations in \cref{alg:constrained-hmc} and separately timing the execution of these operations - details are given in \appcref{app:chain-computation-times}. The Python package \emph{JAX} \citep{jax2018github} was used to allow automatic computation of the derivatives of model functions  and all plots were produced using the Python package \emph{Matplotlib} \citep{hunter2007matplotlib}.

For all experiments we use chains which alternate between Markov transitions which condition on the states at observation time indices
\(
  \lbrace
      \intg{t}_{\intg{b}} : \intg{b}\intg{R}
       ~~\forall\,\intg{b} \in\range{\intg{B}}
  \rbrace
\) and %
\(
  \lbrace
      \intg{t}_{\intg{b}} : \floor{\textstyle\frac{(2\intg{b}-1)\intg{R}}{2}}
       ~~\forall\,\intg{b} \in\range{\intg{B}}
  \rbrace \cup \lbrace \intg{T} \rbrace
\) %
with $\intg{B} = \intg{T} / \intg{R}$. For the experiments in this subsection we ran all chains with constrained integrators using both the Gaussian and St\"ormer--Verlet splittings to allow comparison of their relative performance. We use the parameter values $\sigma = 0.3$, $\epsilon = 0.1$, $\gamma =  1.5$ and $\beta = 0.8$ and initial state $\vct{x}_0 = [-0.5; 0.2]$ to generate the simulated data $\vct{y}_{\range{\intg{T}}}$ for all experiments.

To allow measuring how performance of our approach varies with $\intg{R}$, $\intg{S}$ and $\intg{T}$, we ran experiments over a grid values for each of these parameters with the other two kept fixed, specifically:
  $\intg{R} \in \lbrace 2, 5, 10, 20, 50, 100 \rbrace$ with $\intg{S} = 25$ and $\intg{T} = 100$,
  $\intg{S} \in \lbrace 25, 50, 100, 200, 400 \rbrace$ with $\intg{R} = 5$ and $\intg{T} = 100$,
  $\intg{T} \in \lbrace 25, 50, 100, 200, 400 \rbrace$ with $\intg{R} = 5$ and $\intg{S} = 25$.
For all $(\intg{R},\intg{S},\intg{T})$ values and splittings tested we ran three sets of four chains of 1250 iterations each with independent initialisations (details of the initializations are given in \appcref{app:chain-initialisation}), with the first 250 iterations of each set of four chains an adaptive warm-up phase used to tune the integrator step-size $t$, with the samples from these warm-up iterations omitted from the \ac{ess} estimates but included in the computation time estimates. For all sets of chains the split-$\hat{R}$ convergence diagnostic values computed from the (non-warm-up iterations of the) four chains for all parameter values using rank-normalisation and folding were less than $1.01$ as recommended in \citet{vehtari2019rank}. A dynamic \ac{hmc} implementation \citep{betancourt2017conceptual} was used to set the number of integrator steps per trajectory in each transition. A summary of all the algorithmic parameter values used in the numerical experiments is given in \appcref{app:algorithmic-parameters}.

\begin{figure}[t]
  \includegraphics[width=\textwidth]{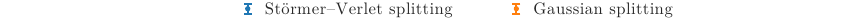}
  \includegraphics[width=\textwidth]{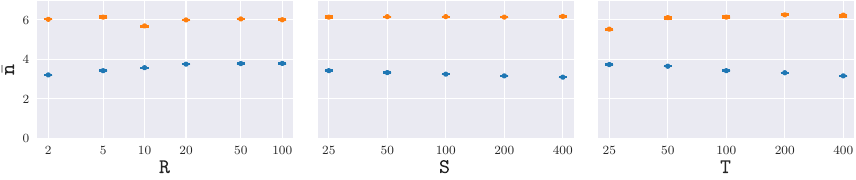}
  \\
  \includegraphics[width=\textwidth]{figures/fhn-noiseless-chmc-splitting-legend}
  \includegraphics[width=\textwidth]{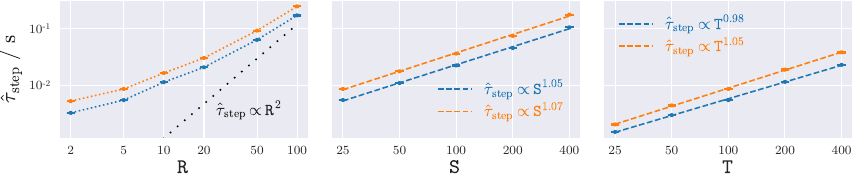}
  \caption{\emph{FitzHugh--Nagumo model (noiseless observations)}: Average number of Newton iterations per integrator step $\bar{\intg{n}}$ (top) and computation time per integrator step $\hat{\tau}_{\textrm{step}}$ in seconds (bottom) for varying $\intg{R}$, $\intg{S}$ and $\intg{T}$. The markers show the minimum, median and maximum across 3 independent runs. The dashed lines in the bottom right two plots show a log-domain least squares fit to the medians for each splitting.}
  \label{fig:fhn-noiseless-av-newton-iterations-and-computation-time-per-integrator-step}
\end{figure}

The top panels in \cref{fig:fhn-noiseless-av-newton-iterations-and-computation-time-per-integrator-step} show how the number of Newton iterations required to solve \eqref{eq:h2-step-position-constraint-condition} in each of the forward and reverse $\opfunc{\Xi}^{h_2}$ steps, averaged across the chains, varies with $\intg{R}$, $\intg{S}$ and $\intg{T}$ and for the two different Hamiltonian splittings. We see that for a given splitting, the number of Newton iterations is close to constant in all cases, providing empirical support for \cref{ass:comp-cost-newton-iterations,ass:comp-cost-newton-iterations-conditioned}.
The bottom panels in \cref{fig:fhn-noiseless-av-newton-iterations-and-computation-time-per-integrator-step} instead show the average time per integrator step $\hat{\tau}_{\textrm{step}}$ varies with $\intg{R}$, $\intg{S}$ and $\intg{T}$ for both splittings. We see that the log-domain least-square fits show a very close to linear scaling of $\hat{\tau}_{\textrm{step}}$ with both $\intg{S}$ and $\intg{T}$, verifying these aspects of the $\mathcal{O}(\intg{R}^2\intg{ST})$ scaling claimed in \cref{the:comp-cost-with-conditioning}. The growth of $\hat{\tau}_{\textrm{step}}$ with $\intg{R}$ over the range here is sub-quadratic (the dotted line shows a quadratic trend for reference), however there is visible acceleration in the growth. An inspection of a breakdown of the total computation time spent on different individual operations revealed that for smaller $\intg{R}$ the $\mathcal{O}(\intg{RST})$ computation of the constraint Jacobian is dominating, with the $\mathcal{O}(\intg{R}^2\intg{ST})$ linear algebra operations only becoming significant for larger $\intg{R}$.

\begin{figure}[t]
  \begin{subfigure}[b]{\linewidth}
    \centering
    \includegraphics[width=\textwidth]{figures/fhn-noiseless-chmc-splitting-legend}
    \includegraphics[width=\textwidth]{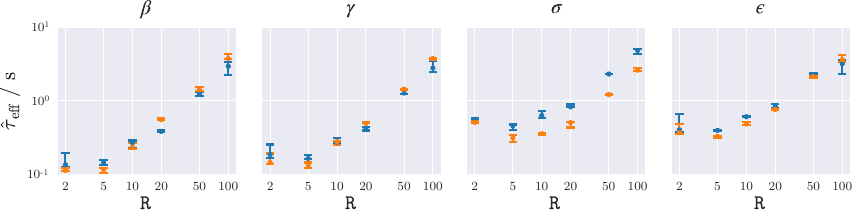}
    \label{sfig:computation-time-per-ess-vs-num-obs-per-block}
  \end{subfigure}%
  \\[-2ex]
  \begin{subfigure}[b]{\linewidth}
    \centering
    \includegraphics[width=\textwidth]{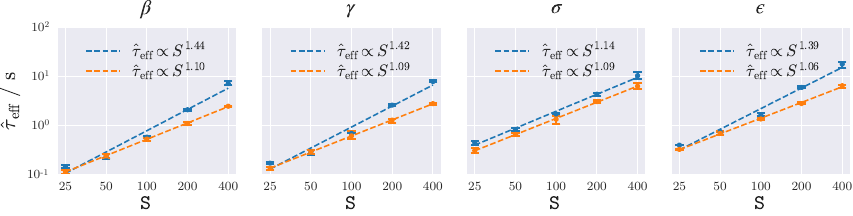}
    \label{sfig:computation-time-per-ess-vs-num-steps-per-obs}
  \end{subfigure}%
  \\[-2ex]
  \begin{subfigure}[b]{\linewidth}
    \centering
    \includegraphics[width=\textwidth]{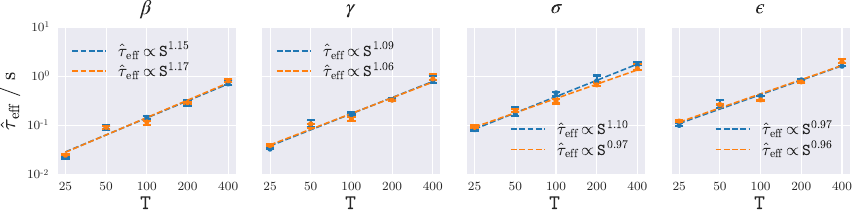}
  \end{subfigure}
  \caption{\emph{FitzHugh--Nagumo model (noiseless observations):} Computation time per effective sample $\hat{\tau}_{\mathrm{eff}}$ in seconds for varying $\intg{R}$, $\intg{S}$ and $\intg{T}$ for each model parameter, in all cases on a log-log scale. The markers shows the minimum, median and maximum across 3 independent runs. The dashed lines in the plots in bottom two rows show log-domain least-square fits to the medians for each splitting.
  }
  \label{fig:fhn-noiseless-computation-time-per-ess}
\end{figure}

\cref{fig:fhn-noiseless-computation-time-per-ess} shows how the estimated computation time per effective sample $\hat{\tau}_{\mathrm{eff}}$ varies with each of $\intg{R}$, $\intg{S}$ and $\intg{T}$, for each of the four model parameters and for each of the two Hamiltonian splittings. First considering the results for varying number of observations per subsequence $\intg{R}$ we see the efficiency is maximised ($\hat{\tau}_{\mathrm{eff}}$ minimised) for both splittings for an intermediate value of $\intg{R} \approx 5$, with a small drop-off in efficiency for $\intg{R} = 2$ and a larger decrease in efficiency as $\intg{R}$ is increased beyond 5. This reflects the competing effects of the reduced cost of each constrained integrator step as $\intg{R}$ is made smaller versus the reduced chain mixing performance in each transition for smaller $\intg{R}$ due to the extra states being (artificially) conditioned on. Importantly we see however that the latter effect is less significant (in this model at least), meaning that performance is still close to optimal for $\intg{R} = 2$, suggesting performance will not be too adversely effected if a too small $\intg{R}$ value is chosen.

Now turning our attention to the plots of $\hat{\tau}_{\mathrm{eff}}$ versus the number of discrete time steps per inter-observation interval $\intg{S}$, we see that there is a clear difference in the scaling of $\hat{\tau}_{\mathrm{eff}}$ with $\intg{S}$ for the two Hamiltonian splittings, with the Gaussian splitting giving a only slightly above linear scaling across all four parameters with exponents in the range 1.06--1.10 compared to 1.14--1.44 for the St\"ormer--Verlet splitting. Inspection of the integrator step sizes $t$ (not shown), which were adaptively tuned in a warm-up phase to control the average acceptance statistic for the chains to be fixed, reveals that for the Gaussian splitting the step size $t$ is in the range $t = 0.29 \pm 0.01$ for all $\intg{S}$ while for the St\"ormer--Verlet splitting shows a decrease from $t = 0.20$ for $\intg{S} = 25$ to $t = 0.10$ for $\intg{S} = 400$, consistent with results suggesting that the step size of the St\"ormer--Verlet integrator needs to scale with $\intg{Q}^{-\sfrac{1}{4}}$ to maintain a constant accept probability of a static integration time \ac{hmc} algorithm in the unconstrained setting \citep{neal2011mcmc,beskos2013optimal}, compared to a dimension-free dependence of the acceptance probability on $t$ for integrators using the Gaussian splitting in appropriate targets \citep{beskos2011hybrid}. While we have emphasised here the superior performance of the Gaussian splitting, we note that the growth of $\hat{\tau}_{\mathrm{eff}}$ with $\intg{S}$ for both methods is very favourable, and shows our approach is able to remain efficient for fine time-discretisations of the continuous time model.

Finally we consider the bottom row of plots in \cref{fig:fhn-noiseless-computation-time-per-ess}, showing how $\hat{\tau}_{\mathrm{eff}}$ varies with the number of observation times $\intg{T}$ for each model parameter. We see that in this case both splittings give very similar scalings, with a close to linear growth in $\hat{\tau}_{\mathrm{eff}}$ with $\intg{T}$ for all four parameters. The (infinite-dimensional) target posterior being approximated for each $\intg{T}$ value differs here unlike the case for varying $\intg{R}$ and $\intg{S}$), in particular becoming more concentrated as $\intg{T}$ increases. The increase in $\hat{\tau}_{\mathrm{eff}}$ with $\intg{T}$ seems to be largely attributable to the increase in the computational cost of each constrained integrator step with $\intg{T}$, and so the mixing performance of the chains seems to be largely independent of $\intg{T}$. This suggests that the constrained \ac{hmc} algorithm is able to efficiently explore posterior distributions with varying geometries. While here the concentration of the posterior is due to an increasing number of observations, in the following section we will see that our approach is also robust to varying informativeness of the individual observations.

\subsection{FitzHugh--Nagumo model with additive observation noise}
\label{subsec:fhn-noisy-experiments}

\begin{figure}[t]
  \includegraphics[width=\textwidth]{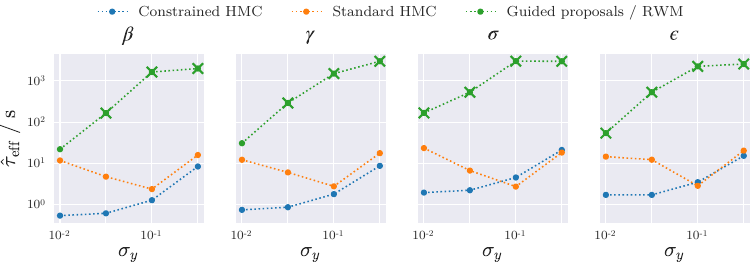}
  \caption{\emph{FitzHugh--Nagumo model (noisy observations)}: Computation time per effective sample $\hat{\tau}_{\mathrm{eff}}$ for varying observation noise standard deviation $\sigma_y$ for each model parameter, in all cases on a log-log scale. Points with cross markers indicate chains with an estimated split-$\hat{R}$ value of greater than $1.01$, indicative of non-convergence.}
  \label{fig:fhn-noisy-computation-time-per-ess}
\end{figure}

As a second example we consider the same hypoelliptic diffusion model as in the preceding section, but now with observations subject to additive Gaussian noise of standard deviation $\sigma_y$, i.e. $\vctfunc{h}(\vct{x},\vct{z},w) = x_1 + \sigma_y w$ and $\eta = \nrm(0, 1)$. The presence of additive observation noise means that the posterior on $(\rvct{u},\rvct{v}_{\range[0]{\intg{ST}}})$ given $\rvct{y}_{\range{\intg{T}}} = \vct{y}_{\range{\intg{T}}}$ has a tractable Lebesgue density. We therefore compare our constrained \ac{hmc} approach to running a standard (unconstrained) \ac{hmc} algorithm targeting the posterior on $(\rvct{u},\rvct{v}_{\range[0]{\intg{ST}}})$ with details of the posterior density and \ac{hmc} algorithm used given in \appcref{app:hmc-additive-noise-models}. As a further baseline we also compare to the approach of \citet{meulen2020continuous}, which uses a Metropolis-within-Gibbs scheme alternating \emph{guided proposal} Metropolis--Hastings updates to the latent path $\rvct{x}_{\range[0]{\intg{ST}}}$ given parameters $\rvct{z}$, with \ac{rwm} updates to the parameters $\rvct{z}$ given the path $\rvct{x}_{\range[0]{\intg{ST}}}$. An application of this approach to the FitzHugh--Nagumo model considered here is described in \citet{meulen2020bayesian}, and we use the Julia code accompanying that article to run the experiments.

We use the same priors and time-discretisation as in the previous section, and fix $\intg{S} = 40$ and $\intg{T} = 100$. Simulated observed sequences $\vct{y}_{\range{\intg{T}}}$ were generated for each of the observation noise variances $\sigma_y^2 \in \lbrace 10^{-4}, 10^{-3}, 10^{-2}, 10^{-1}\rbrace$. In all cases $\vct{y}_{\range{\intg{T}}}$ was generated using the same parameters and initial state as in the previous section and sharing the same values for $\rvct{v}_{\range{\intg{ST}}}$ and $\rvct{w}_{\range{\intg{T}}}$ (sampled from their standard normal priors). Chains targeting the resulting posteriors were run for each $\sigma_y^2$ value and for each of the three \ac{mcmc} methods being considered.

For our constrained \ac{hmc} algorithm we used $\intg{R} = 5$ and ran chains using a constrained integrator based on the St\"ormer--Verlet splitting, with results instead using the Gaussian splitting showing a similar pattern of performance and hence omitted here to avoid duplication. For the standard \ac{hmc} algorithm a diagonal metric matrix representation $\mtx{M}$ was adaptively tuned in the warm-up iterations with this found to uniformly outperform using a fixed identity matrix for all $\sigma_y$ values tested here. For both the standard and constrained \ac{hmc} algorithms we run four chains of 3000 iterations with the first 500 iterations an adaptive warm-up phase used to tune the integrator step-size $t$ (and $\mtx{M}$ for the standard \ac{hmc} case). For the guided proposals / \ac{rwm} case we ran four chains of $3\times 10^5$ iterations, with the first $5 \times 10^4$ iterations an adaptive warm-up phase where the persistence parameter of the guided proposals update to $\rvct{x}_{\range{\intg{ST}}} \gvn \rvct{z}$ and step sizes of the random-walk proposals for the update to $\rvct{z} \gvn \rvct{x}_{\range{\intg{ST}}}$ were adapted as described in \citet{meulen2020continuous}.

Estimated computation time per effective sample $\hat{\tau}_{\textrm{eff}}$ values were calculated for the chains of each of the three \ac{mcmc} methods and each of the observation noise variance values $\sigma_y^2$. The \ac{ess} estimates were calculated as described in the preceding section, however, the true total wall-clock run times were used for the chain computation times here due the difficulty in ensuring a consistent treatment of different \ac{mcmc} algorithms in the approach used in the previous section. To ensure as fair a comparison as possible all chains were run on the same computer and limited to a single processor core to avoid differences due to varying exploitation of parallel computation in the implementations.

\cref{fig:fhn-noisy-computation-time-per-ess} summarises the results. We first note that despite the large number of iterations used for the guided proposals / \ac{rwm} chains ($3\times 10^5$), the per-parameter split-$\hat{R}$ diagnostics \citep{vehtari2019rank} for the chains indicated non-convergence of the chains in nearly all cases, with only the chains for $\sigma_y = 10^{-2}$ appearing to be approaching convergence with $\smash{\hat{R}}$ values in the range [1.01, 1.05] compared to $\smash{\hat{R}}$ values in the range [1.39, 1.87] for $\sigma_y = 10^{-0.5}$. The poor convergence here seems to be at least in part due to the difficulty in finding globally appropriate values of the \ac{rwm} step sizes, with the step sizes still changing significantly in the final iterations of the warm-up phase and the final adapted values differing significantly across chains for the same $\sigma_y$.

Given the poor convergence the estimated \ac{ess} values must be treated with some caution, however even for the $\sigma_y = 10^{-2}$ case where the chains appeared to be closest to convergence the estimated $\hat{\tau}_{\textrm{eff}}$ values for the guided proposals / \ac{rwm} chains are between 30 and 80 times larger than the corresponding values for the constrained \ac{hmc} chains. As Julia implementations of numerical algorithms generally significantly outperform Python equivalents \citep{bezanson2017julia}, the superior sampling performance of the (Python) constrained \ac{hmc} implementation compared to the (Julia) guided proposals / \ac{rwm} implementation here seems unlikely to be just due to differences in the efficiency of the implementations, but rather reflects significantly improved mixing of the joint gradient-informed updates to the latent variables by the constrained \ac{hmc} algorithm, compared to the non-gradient-informed Metropolis-within-Gibbs updates of the guided proposals / \ac{rwm} algorithm.

Comparing now the results for the standard and constrained \ac{hmc} algorithms, we see that while both algorithms perform similarly for larger $\sigma_y$ values (i.e. less informative observations), the constrained \ac{hmc} algorithm provides significantly better sampling efficiency for smaller $\sigma_y$ values. Inspecting the integrator step size $t$ set at the end of the adaptive warm-up phase for each of $\sigma_y$ values reveals that, while for constrained \ac{hmc} all step sizes $t$ fall in the range 0.17--0.18 and so seem invariant to $\sigma_y$, for the standard \ac{hmc} chains, $t$ ranges from $5.1\times 10^{-4}$ for $\sigma_y = 10^{-2}$ to $9.1 \times 10^{-3}$ for $\sigma_y = 10^{-1}$, resulting in a need to take more integrator steps per transition to make moves of the same distance in the latent space and hence a decreasing sampling efficiency as $\sigma_y$ becomes smaller.

The results for $\sigma_y = 10^{-0.5}$ break the trend of increasing $\sigma_y$ leading to increased efficiency for the standard \ac{hmc} chains, with a significant increase in $\hat{\tau}_{\textrm{eff}}$ compared to $\sigma_y = 10^{-1}$. This seems to be due to a roughly halving of the integrator step size $t$ set in the adaptive warm-up phase to $5.2\times10^{-3}$ for $\sigma_y = 10^{-0.5}$, which, combined with the more diffuse posterior for the larger $\sigma_y$ value, led to a significant increase in the average number of integrator steps per transition and so  computational cost per effective sample. A potential explanation for the decrease in the adapted step size is that the more diffuse posterior extends to regions where the posterior density has higher curvature necessitating a smaller step size to control the Hamiltonian error. In contrast, for the constrained \ac{hmc} chains, the Hamiltonian error is controlled with a close to constant step size for all $\sigma_y$, however, there is a drop in efficiency as $\sigma_y$ becomes larger, which seems to be due to the more diffuse posterior requiring a greater number of integrator steps to explore and so higher computational cost per effective sample on average.

\subsection{Susceptible-infected-recovered model with additive observation noise}
\label{subsec:sir-experiments}

As a final example we perform inference in an epidemiological compartmental model given real observations of the time course of the number of infected patients in an influenza outbreak in a boarding school \citep{bmj1978influenza}. Specifically we consider a diffusion approximation of a \ac{sir} model (see, e.g., the derivation in \citet[\S 5.1.3]{fuchs2013inference}), with a time-varying contact rate parameter itself modelled as a diffusion process as proposed in \citet{ryder2018black}, resulting in the following three-dimensional elliptic \ac{sde} system
\begin{equation}\label{eq:sir-sde}
  \begin{bmatrix}
      \dr \rvar{s} \\
      \dr \rvar{i} \\
      \dr \rvar{c} \\
  \end{bmatrix} =
  \begin{bmatrix}
    -N^{-1} \rvar{c} \mathsf{s} \rvar{i} \\
    N^{-1} \rvar{c} \mathsf{s} \rvar{i} - \gamma \mathsf{i} \\
    (\alpha(\beta - \log \rvar{c}) + \frac{\sigma^2}{2})\rvar{c}
  \end{bmatrix}\dr\tau +
  \begin{bmatrix}
    \sqrt{N^{-1}\rvar{c}\rvar{s} \rvar{i}} & 0 & 0\\
    -\sqrt{N^{-1}\rvar{c} \rvar{s} \rvar{i}} & \sqrt{\gamma\rvar{i}} & 0 \\
    0 & 0 & \sigma\rvar{c}
  \end{bmatrix}
  \begin{bmatrix}
    \dr \rvar{w}_1 \\
    \dr \rvar{w}_2 \\
    \dr \rvar{w}_3 \\
\end{bmatrix}
\end{equation}
where $\tau$ is the time in days, $\rvar{s}$  and $\rvar{i}$ are the number of susceptible and infected individuals respectively, $\rvar{c}$ the contact rate, $N$ the population size and $\gamma$ the recovery rate parameter. The \ac{sde} for $\rvar{c}$ arises from $\log\rvar{c}$ following an Ornstein--Uhlenbeck process with reversion rate $\alpha$, long term mean $\beta$ and instantaneous volatility $\sigma$.

As each of $\rvar{s}$, $\rvar{i}$ and $\rvar{c}$ represent positive-valued quantities, the diffusion state is defined to be $\rvct{x} = [\log\rvar{s}; \log\rvar{i};\log\rvar{c}] \in \reals^3$ with drift $\vctfunc{a}$ and diffusion coefficient $\opfunc{B}$ functions derived from the above \acp{sde} via It\^o's lemma. By computing the time-discretisation in this log-transformed space, the positivity of $\rvar{s}$, $\rvar{i}$ and $\rvar{c}$ is enforced and the numerical issues arising when evaluating the square-root terms in the diffusion coefficient for negative $\rvar{s}$, $\rvar{i}$ or $\rvar{c}$ are avoided. The observed data $\vct{y}_{\range{\intg{T}}}$ corresponds to measurements of the number of infected individuals $\rvar{i} = \exp(\rvar{x}_2)$ at daily intervals, i.e. $\Delta = 1$, over a period of $\intg{T} = 14$ days, with the observations assumed to be subject to additive noise of unknown standard deviation $\sigma_y$, i.e. $\rvar{y}_\intg{t} = \exp(\rvar{x}_2(\rvar{t})) + \sigma_y \rvar{w}_\intg{t}$. The $\intg{Z}= 5$ dimensional parameter vector is then $\rvct{z} = [\gamma; \alpha; \beta; \sigma; \sigma_y]$. Details of the priors and discretisation used are given in \appcref{app:sir-model-details}.

\begin{figure}[!t]
  \includegraphics[width=0.31\textwidth]{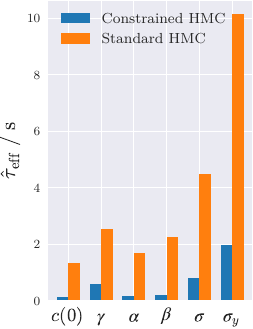}\qquad
  \includegraphics[width=0.62\textwidth]{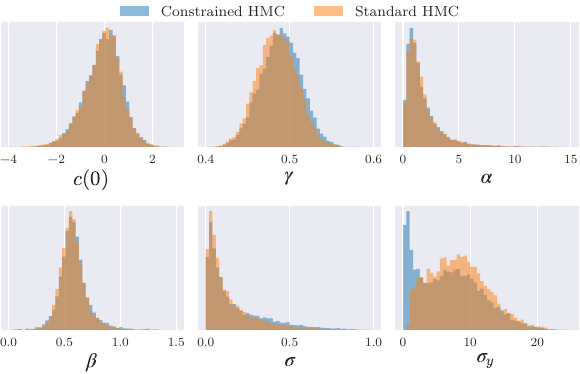}
  \caption{\emph{SIR model}: Computation time per effective sample $\hat{\tau}_{\mathrm{eff}}$ (leftmost panel) and estimated posterior marginals (right panels) for each model parameter computed using constrained and standard \ac{hmc} chain samples.}
  \label{fig:sir-noisy-computation-time-per-ess-and-marginals}
\end{figure}

We compare the performance of our proposed constrained \ac{hmc} approach to a standard \ac{hmc} algorithm, with the noise in the observations meaning that the posterior on $\rvct{u}$ and $\rvct{v}_{\range[0]{\intg{ST}}}$ admits a Lebesgue density. For each algorithm we run four chains of 3000 iterations with the first 500 iterations an adaptive warm-up phase. For our constrained \ac{hmc} algorithm due to the small number of, and high correlations between, the observations we do not introduce any artificial conditioning on intermediate states i.e. $\intg{R} = \intg{T} = 14$. Chains using constrained integrators based on both the St\"ormer--Verlet and Gaussian splitting show very similar performance here so we show only the St\"ormer--Verlet case to avoid duplication. For the standard \ac{hmc} algorithm a diagonal metric matrix representation $\mtx{M}$ was adaptively tuned in the warm-up iterations. The estimated computation time per effective sample $\hat{\tau}_{\textrm{eff}}$ values were calculated using the total wall-clock run times for the chains.

The results are summarised in \cref{fig:sir-noisy-computation-time-per-ess-and-marginals}. The left plot shows the estimated per-parameter $\hat{\tau}_{\textrm{eff}}$ values for each of the two \ac{mcmc} methods: we see that the constrained \ac{hmc} algorithm is able to give significantly improved sampling efficiency over standard \ac{hmc} here. More importantly however it appears that the standard \ac{hmc} algorithm is in fact failing to explore the full posterior. The grid of six plots in the right of \cref{fig:sir-noisy-computation-time-per-ess-and-marginals} shows the estimated posterior marginals for each parameter computed from either the constrained or standard \ac{hmc} chain samples. There is a clear discrepancy in the estimated posterior marginal of $\sigma_y$ between the two methods, with the standard \ac{hmc} chains having many fewer samples at smaller $\sigma_y$ values compared to the constrained \ac{hmc} chains. \appcref{fig:sir-noisy-pair-plots} shows the corresponding estimated pairwise marginals with $\sigma_y$ on a log-scale, with the poorer coverage of small $\sigma_y$ values by the standard \ac{hmc} chains more apparent.

While the per-parameter $\hat{R}$ diagnostics for the standard \ac{hmc} chains are all below 1.01, some hint of the underlying issue being encountered here is given by the high number of iterations in which the integration of the Hamiltonian dynamics diverged for the standard \ac{hmc} chains - roughly 4\% of the non-warm-up iterations for each chain. Such divergences are indicative of the presence of regions of high curvature in the posterior distribution that result in the numerical simulation of the Hamiltonian trajectories becoming unstable, and in some cases may be ameliorated by use of a smaller integrator step size $t$ \citep{betancourt2017conceptual}.

Here specifically the adaptive tuning of the step size $t$ in the warm-up phase has led to a step size which is too large for exploring the regions of the posterior in which $\sigma_y$ is small. Although by setting a higher target acceptance statistics for the step size adaptation algorithm or hand tuning $t$ to a smaller value we could potentially fix this issue, this would be at the cost of an associated decrease in sampling efficiency, leading to even poorer performance relative to the constrained \ac{hmc} chains. As seen in the results for the FitzHugh--Nagumo model in Section \ref{subsec:fhn-noisy-experiments}, if $\sigma_y$ is fixed the integrator step size $t$ for the standard \ac{hmc} algorithm needs to be decreased as $\sigma_y$ is decreased to control the acceptance rate resulting in a higher computation cost per effective sample - \appcref{fig:sir-noisy-computation-time-per-ess} illustrates this directly for the \ac{sir} model. When $\sigma_y$ instead is unknown as here, standard \ac{hmc} needs to use a step size appropriate for the smallest $\sigma_y$ value `typical' under the posterior, which if $\sigma_y$ is poorly informed by the data (as is the case for this model) can require using very small integrator step sizes $t$. In contrast as the constrained \ac{hmc} algorithm is able to use an integrator step size $t$ which is independent of $\sigma_y$, the sampling efficiency of the chains is not limited by the need to use a small step size $t$ to explore regions of the posterior in which $\sigma_y$ is small. %
\section{Conclusions \& further directions}
\label{sec:conclusion}

We have introduced a methodology for calibrating a wide class of diffusion-driven models. Our approach is based on recasting the inferential problem as one of exploring a posterior supported on a manifold, the structure of the latter determined by the observational constraints on the generative model. Once this viewpoint is adopted, available techniques from the literature on constrained  \ac{hmc} can be called upon to allow for effective traversing of the high-dimensional latent space. We have further shown that the Markovian structure of the model can be exploited to design a methodology with computational complexity that scales linearly with both the resolution of the time-discretisation and the number of observation times.

A critical argument put forward via the methodology developed in this work is that practitioners working with \ac{sde} models are now provided  with the option to refer to a \emph{single} and \emph{highly automated}, algorithmic framework for Bayesian calibration of their models. This algorithmic framework employs efficient Hamiltonian dynamics and adheres to all sought out criteria listed in Remark~\ref{rem:criteria}.

When exploring distributions with rapidly varying curvatures, standard \ac{hmc} methods with a fixed step size can yield trajectories that either require too small of a step size (as in the FitzHugh--Nagumo model with noise in Section \ref{subsec:fhn-noisy-experiments}), or become unstable and diverge if the step size is not small enough in areas of high curvature of the posterior on the latent space  (as with our \ac{sir} example in Section~\ref{subsec:sir-experiments} where variations in the scale parameter $\sigma_y$ have strong effect on the curvature). In both cases, particularly strong effects can render standard  \ac{hmc} non-operational (as in the \ac{sir} case). Although the methodology presented in \cite{girolami2011riemann} can in principle be helpful in such contexts, this class of algorithms is intrinsically constructed to induce good performances in the centre of the target distribution as it involves an expectation over the data, and not the \emph{given} data, for the specification of the employed Riemannian metric. Constrained \ac{hmc} dynamics can provide a more appropriate approach for dealing with rapidly varying curvature across the whole of the support of the target distribution. When combined with efficient discretisations of the dynamics --- as in the case of the class of diffusion models we have studied in this work --- they can provide statistically efficient methods.

The viewpoint adopted  in this paper is potentially relevant to a larger class of stochastic models for time series (e.g.~random ordinary differential equations), as well as other Markovian model classes (e.g.~Markov networks). Some of the authors are currently involved in applying such \ac{mcmc} methods to Bayesian inverse problems; manifold structures naturally appear in the low noise regime \citep{beskos2018asymptotic}. In general, we believe that the approach presented in this paper warrants further investigation, with a corresponding study of critical algorithmic aspects, e.g., computational complexity and mixing properties.
\section*{Acknowledgements}

MMG acknowledges support from the Lloyd's Register Foundation programme on data-centric engineering at the Alan Turing Institute. AHT and MMG acknowledge support from the Singapore Ministry of Education Tier 2 (MOE2016-T2-2-135) and a Young Investigator Award Grant (NUSYIA FY16 P16; R-155-000-180-133). AB acknowledges support from the Leverhulme Trust Prize. This research made use of the Rocket High Performance Computing service at Newcastle University, UK.

\appendix

\section{Proof of \cref{the:comp-cost-without-conditioning}}
\label{app:comp-cost-without-conditioning}

We will use the following standard results from \ac{ad}. For more details see for example \citet{griewank1993some} and \citet[Chapter 4]{griewank2008evaluating}.

\begin{lemma}
  \label{lem:ad-comp-cost-jvp-and-vjp}
  For a differentiable function $\vctfunc{f} : \reals^{\intg{M}} \to \reals^\intg{N}$ the \emph{Jacobian-vector product} operator $\textsc{jvp}(\vctfunc{f})(\vct{v})(\vct{x}) := \jacob\vctfunc{f}(\vct{x})\vct{v}$ and the \emph{vector-Jacobian product} operator $\textsc{vjp}(\vctfunc{f})(\vct{v})(\vct{x}) := \vct{v}\tr\jacob\vctfunc{f}(\vct{x})$ can each be evaluated at $\bigo(1)$ times the operation cost of evaluating $\vctfunc{f}(\vct{x})$ using respectively forward- and reverse-mode \ac{ad}.
\end{lemma}
\begin{corollary}
  \label{cor:ad-comp-cost-jacobian}
  For a differentiable function $\vctfunc{f} : \reals^{\intg{M}} \to \reals^\intg{N}$ the full Jacobian $\jacob\vctfunc{f}(\vct{x}) = [(\textsc{jvp}(\vctfunc{f})(\vct{e}_{\intg{m}}^{\intg{M}})(\vct{x})\tr)_{\intg{m}=1}^{\intg{M}}]\tr = [(\textsc{vjp}(\vctfunc{f})(\vct{e}_{\intg{n}}^{\intg{N}})(\vct{x}))_{\intg{n}=1}^{\intg{N}}]$ ($\vct{e}^\intg{K}_{\range{\intg{K}}}$ are the standard basis vectors of $\reals^{\intg{K}}$) can be evaluated at $\bigo(\min(\intg{M},\intg{N}))$ times the cost of evaluating $\vctfunc{f}(\vct{x})$.
\end{corollary}

For the purposes of analysing the cost of a single constrained integrator step in \cref{alg:constrained-hmc} we will ignore the operations associated with the reversibility check: as the additional operations are a repeat of a subset of the `forward' step, they will only alter the complexity by a constant factor. As our interest is in the scaling of the cost with number of observation times $\intg{T}$ and number of time steps per interobservation interval $\intg{S}$ we will express the complexity of the operations making up a step in terms of only these two dimensions using \emph{Big O} notation $\bigo(\cdot)$, with the other problem dimensions assumed fixed. \cref{alg:constrained-step-no-conditioning} shows a detailed breakdown of the operations involved in a single constrained integrator step, annotated with the computational complexity of each line in terms of $\intg{S}$ and $\intg{T}$ and the maximum number of times each line will be called (relevant for lines within loops). Below we justify the stated complexities.

\stepcounter{algorithm}
\begin{algorithm}[!t]
\caption{Annotated constrained integrator step (no blocking / conditioning)}
\label{alg:constrained-step-no-conditioning}
\begin{algorithmic}
      \Require (current state and cached quantities)\\
        $\vct{q}$ : current position, $\vct{p}$ : current momentum,
        $\vct{g} = \grad h_1(\vct{q})$,
        $\mtx{L} = \chol(\opfunc{G}_\mtx{M}(\vct{q}))$.
      \Ensure (next state and cached quantities)\\
      $\bar{\vct{q}}$ : next position, $\bar{\vct{p}}$ : next momentum,
      $\bar{\vct{g}} = \grad h_1(\bar{\vct{q}})$,
      $\bar{\mtx{L}} = \chol(\opfunc{G}_\mtx{M}(\bar{\vct{q}}))$.
    \end{algorithmic}
    \hrule
  \begin{tabularx}{\textwidth}{@{}cXcc@{}}
  &  & Complexity $\bigo(\cdot)$ & Max. calls \\
  \rownumber \label{ln:vector-addition-scalar-mult-1} & $\tilde{\vct{p}} \gets \vct{p} - \tfrac{t}{2} \vct{g}$ & $\intg{ST}$ & 1 \\
  \rownumber \label{ln:multiply-metric-inv-1} & $\tilde{\vct{\nu}} \gets \mtx{M}^{-1}\tilde{\vct{p}}$ & $\intg{ST}$ & 1 \\
  \rownumber \label{ln:constraint-jvp-1} & $\tilde{\vct{\omega}} \gets \jacob\vctfunc{c}(\vct{q})\tilde{\vct{\nu}}$ & $\intg{ST}$ & 1 \\
  \rownumber \label{ln:cholesky-solve-1} & $\vct{\lambda}_1 \gets \mtx{L}^{-\mathsf{T}}(\mtx{L}^{-1}\tilde{\vct{\omega}})$ & $\intg{T^2}$ & 1 \\
  \rownumber & \textbf{for} $\intg{j} \in \range{\intg{J}}$ & & \\
    \rownumber \label{ln:constraint-vjp-1} & \quad $\vct{p}' \gets \tilde{\vct{p}} - \jacob\vctfunc{c}(\vct{q})\tr\vct{\lambda}_{\intg{j}}$ & $\intg{ST}$ & $\intg{J}$ \\
    \rownumber \label{ln:h2-update-1} & \quad $\vct{q}_{\intg{j}} \gets \opfunc{\Gamma}_t^{q,q}\vct{q} + \opfunc{\Gamma}_t^{q,p}\vct{p}'$ & $\intg{ST}$ & $\intg{J}$ \\
    \rownumber \label{ln:h2-update-2} & \quad $\vct{p}_{\intg{j}} \gets \opfunc{\Gamma}_t^{p,q}\vct{q} + \opfunc{\Gamma}_t^{p,p}\vct{p}'$ & $\intg{ST}$ & $\intg{J}$ \\
    \rownumber \label{ln:constraint-function} & \quad $\vct{e}_{\intg{j}} \gets \vctfunc{c}(\vct{q}_{\intg{j}})$ & $\intg{ST}$ & $\intg{J}$ \\
    \rownumber \label{ln:constraint-jacobian} & \quad $\mtx{C}_{\intg{j}} \gets \jacob\vctfunc{c}(\vct{q}_{\intg{j}})$ & $\intg{ST}^2$ & $\intg{J}$ \\
    \rownumber \label{ln:vector-norm} & \quad \textbf{if} $\Vert\vct{e}_{\intg{j}}\Vert < \theta_c$ \textbf{and} ($\intg{j} \equiv 1$ \textbf{or} $\Vert\vct{q}_{\intg{j}} - \vct{q}_{\intg{j}-1}\Vert < \theta_q$) & $\intg{ST}$ & $\intg{J}$ \\
      \rownumber & \qquad $(\bar{\vct{q}}, \bar{\vct{p}}) \gets (\vct{q}_{\intg{j}}, \vct{p}_{\intg{j}})$ & & \\[1pt]
      \rownumber \label{ln:gram} & \qquad $\mtx{G} \gets \mtx{C}_{\intg{j}}\mtx{M}^{-1}\mtx{C}_{\intg{j}}\tr$ & $\intg{ST}^3$ & $1$ \\
      \rownumber \label{ln:gram-cholesky} & \qquad $\bar{\mtx{L}} \gets \chol(\mtx{G})$ & $\intg{T}^3$ & $1$ \\
      \rownumber & \qquad \textbf{break} & & \\
    \rownumber \label{ln:newton-matrix} & \quad $\mtx{H}_{\intg{j}} \gets \mtx{C}_{\intg{j}} (\opfunc{\Gamma}_t^{q,p})^{-1} \mtx{C}\tr$ & $\intg{ST}^3$ & $\intg{J}$ \\
    \rownumber \label{ln:newton-solve} & \quad $\vct{\lambda}_{\intg{j}+1} \gets \vct{\lambda}_{\intg{j}} + \mtx{H}_{\intg{j}}^{-1}\vct{e}_{\intg{j}}$ & $\intg{T}^3$ & $\intg{J}$ \\
  \rownumber \label{ln:h1-gradient} & $\bar{\vct{g}} \gets \grad h_1(\bar{\vct{q}})$ & $\intg{S}\intg{T}^3$ & 1\\
  \rownumber \label{ln:vector-addition-scalar-mult-2} & $\tilde{\vct{p}} \gets \bar{\vct{p}} - \tfrac{t}{2} \bar{\vct{g}}$ & $\intg{ST}$ & 1 \\
  \rownumber \label{ln:multiply-metric-inv-2} & $\tilde{\vct{\nu}} \gets \mtx{M}^{-1}\tilde{\vct{p}}$ & $\intg{ST}$ & 1 \\
  \rownumber \label{ln:constraint-jvp-2}  & $\tilde{\vct{\omega}} \gets \jacob\vctfunc{c}(\vct{q})\tilde{\vct{\nu}}$ & $\intg{ST}$ & 1 \\
  \rownumber \label{ln:cholesky-solve-2} & $\tilde{\vct{\lambda}} \gets \bar{\mtx{L}}^{-\mathsf{T}}(\bar{\mtx{L}}^{-1}\tilde{\vct{\omega}})$ & $\intg{T^2}$ & 1 \\
  \rownumber \label{ln:constraint-vjp-2} & $\bar{\vct{p}} \gets \tilde{\vct{p}} - \jacob\vctfunc{c}(\bar{\vct{q}})\tr\tilde{\vct{\lambda}}$ & $\intg{ST}$ & 1 \\
  \end{tabularx} \end{algorithm}

The overall dimensionality of the latent space is $\intg{Q} = \intg{U} + \intg{V}_0 + \intg{STV} + \intg{TW} = \bigo(\intg{ST})$ and the number of constraints is $\intg{C} = \intg{TY} = \bigo(\intg{T})$. Vector operations with a linear cost in dimension such as addition and/ scalar multiplication (\cref{ln:vector-addition-scalar-mult-1,ln:vector-addition-scalar-mult-2}) or evaluating an $\infty$-norm (\cref{ln:vector-norm}) will therefore be of complexity $\bigo(\intg{ST})$.

The metric matrix representation $\mtx{M}$ is assumed to be equal to the prior precision matrix as recommended in \cref{subsec:metric-choice}. As $(\rvct{u},\rvct{v}_0,\rvct{v}_{\range{\intg{ST}}},\rvct{w}_{\range{\intg{T}}})$ are independent under the prior, $\mtx{M}$ thus has a block-diagonal structure $\mtx{M} = \diag(\mtx{M}_1, \mtx{M}_2, \mtx{M}_3, \mtx{M}_4)$ with $\mtx{M}_1$ a $\intg{U} \times \intg{U}$ matrix, $\mtx{M}_2$ a $\intg{V}_0 \times \intg{V}_0$ matrix, $\mtx{M}_3$ a $\intg{STV} \times \intg{STV}$ diagonal matrix and $\mtx{M}_4$ a $\intg{TW} \times \intg{TW}$ block-diagonal matrix with block size $\intg{W}$. Multiplication of a $\intg{Q}$ dimensional vector by $\mtx{M}^{-1}$ (\cref{ln:multiply-metric-inv-1,ln:multiply-metric-inv-2}) will therefore have a $\bigo(\intg{ST})$ cost.

Each evaluation of the constraint function in \cref{eq:constraint-function} requires $\intg{ST}$ evaluations of $\vctfunc{f}_\delta$, $\intg{T}$ evaluations of $\vctfunc{h}$ and a single evaluation of each of $\vctfunc{g}_{\rvct{z}}$ and $\vctfunc{g}_{\rvct{x}_0}$; evaluating $\vctfunc{c}(\vct{q})$ therefore has a $\bigo(\intg{ST})$ cost (\cref{ln:constraint-function}). By \cref{lem:ad-comp-cost-jvp-and-vjp} evaluating a Jacobian vector product $\jacob\vctfunc{c}(\vct{q})\vct{v}$ for a $\intg{Q}$ dimensional vector $\vct{v}$ therefore also has a $\bigo(\intg{ST})$ cost (\cref{ln:constraint-jvp-1,ln:constraint-jvp-2}), as does evaluating a vector Jacobian product $\vct{v}\tr\jacob\vctfunc{c}(\vct{q})$ for a $\intg{C}$ dimensional vector $\vct{v}$ (\cref{ln:constraint-vjp-1,ln:constraint-vjp-2}). By \cref{cor:ad-comp-cost-jacobian} evaluating the full Jacobian of the constraint function has a $\bigo(\intg{ST}^2)$ cost (\cref{ln:constraint-jacobian}).

The constraint Jacobian $\jacob\vctfunc{c}(\vct{q})$ is a $\bigo(\intg{T}) \times \bigo(\intg{ST})$ matrix. Postmultiplying a $\bigo(\intg{T}) \times \bigo(\intg{ST})$ matrix by the metric inverse $\mtx{M}^{-1}$ will have a $\bigo(\intg{ST}^2)$ due to its block-diagonal structure, and the multiplication of a $\bigo(\intg{T}) \times \bigo(\intg{ST})$ matrix with a $\bigo(\intg{ST}) \times \bigo(\intg{T})$ matrix will have a $\bigo(\intg{ST}^3)$ cost, therefore given the full Jacobian of the constraint function, evaluating the Gram matrix has an overall $\bigo(\intg{ST}^3)$ complexity (\cref{ln:gram}). The Gram matrix $\opfunc{G}_{\mtx{M}}(\vct{q})$ a $\bigo(\intg{T}) \times \bigo(\intg{T})$ matrix and so its Cholesky factorisation can be computed with $\bigo(\intg{T}^3)$ complexity (\cref{ln:gram-cholesky}). Given a lower-triangular Cholesky factor, a linear system in the Gram matrix can be solved with two triangular solves at $\bigo(\intg{T}^2)$ cost (\cref{ln:cholesky-solve-1,ln:cholesky-solve-2}).

For the St\"ormer--Verlet splitting $\grad h_1(\vct{q}) = \grad\ell(\vct{q})$ while for the Gaussian splitting $\grad h_1(\vct{q}) = \grad \ell(\vct{q}) - \vct{q}$. By \cref{cor:ad-comp-cost-jacobian} evaluating $\grad \ell(\vct{q})$ has the same complexity as evaluating  $\ell(\vct{q}) = -\log\inlinetd{\rho}{\lebm{\intg{Q}}}(\vct{q}) - \sfrac{1}{2}\log\det{\opfunc{G}_{\mtx{M}}(\vct{q})}$. Evaluating the log prior density $\log\inlinetd{\rho}{\lebm{\intg{Q}}}(\vct{q})$ has a $\bigo(\intg{ST})$ complexity while evaluating the Gram matrix log-determinant $\log\det{\opfunc{G}_{\mtx{M}}(\vct{q})}$ has a $\bigo(\intg{ST}^3)$ complexity. For both splittings, evaluating $\grad h_1(\vct{q})$ therefore has a $\bigo(\intg{ST}^3)$ complexity (\cref{ln:h1-gradient}). For the St\"ormer--Verlet splitting $\opfunc{\Gamma}^{q,q}_t = \idmtx_{\intg{Q}}$, $\opfunc{\Gamma}^{q,p}_t = t\mtx{M}^{-1}$, $\opfunc{\Gamma}^{p,q}_t = \mtx{0}$ and $\opfunc{\Gamma}^{p,p}_t = \idmtx_{\intg{Q}}$. For the Gaussian splitting $\opfunc{\Gamma}^{q,q}_t = \mtx{R}\cos(\mtx{\Omega}t)\mtx{R}\tr$, $\opfunc{\Gamma}^{q,p}_t = \mtx{R}\mtx{\Omega}\sin(\mtx{\Omega}t)\mtx{R}\tr$, $\opfunc{\Gamma}^{p,q}_t = -\mtx{R}\mtx{\Omega}^{-1}\sin(\mtx{\Omega}t)\mtx{R}\tr$ and $\opfunc{\Gamma}^{p,p}_t = \mtx{R}\cos(\mtx{\Omega}t)\mtx{R}\tr$, with $\mtx{R}$ an orthogonal matrix with columns formed by the normalised eigenvectors of $\mtx{M}^{-1}$ and $\mtx{\Omega}$ a diagonal matrix of the square-roots of the corresponding eigenvalues such that $\mtx{M}^{-1} = \mtx{R} \mtx{\Omega}^2 \mtx{R}\tr$. The orthogonal matrix $\mtx{R}$ will inherit the block-diagonal structure of $\mtx{M}^{-1}$ and so multiplying a $\intg{Q}$-dimensional vector by $\mtx{R}$ or $\mtx{R}\tr$ will have the same $\bigo(\intg{ST})$ complexity as multiplying by $\mtx{M}^{-1}$. For both splittings, evaluating $\opfunc{\Gamma}^{x,y}_t\vct{v}$ for a $\intg{Q}$-dimensional vector $\vct{v}$ and for any $(x,y) \in \lbrace q, p \rbrace^2$ will therefore have at most a $\bigo(\intg{ST})$ complexity (\cref{ln:h2-update-1,ln:h2-update-2}).

From the complexities in \cref{alg:constrained-step-no-conditioning}, the overall complexity per step is $\bigo(\intg{J}\intg{S}\intg{T}^3)$, with \cref{ln:newton-matrix} being the dominant cost. Under Assumption \ref{ass:comp-cost-newton-iterations} $\intg{J}$ can be kept constant with respect to $\intg{S}$ and $\intg{T}$, so the per step complexity simplifies to $\bigo(\intg{S}\intg{T}^3)$. \qed %
\section{Proof of \cref{the:comp-cost-with-conditioning}}
\label{app:comp-cost-with-conditioning}

By exploiting the structure in the constraint Jacobian and Gram matrix afforded by conditioning on intermediate states $(\vct{x}_{\intg{S}\intg{t}_{\intg{b}}})_{\intg{b}=1}^{\intg{B}}$ and using the identities in \cref{eq:conditioned-gram-determinant} and \cref{eq:conditioned-gram-inverse}, all of the operations with quadratic or cubic complexity in $\intg{T}$ in the previous analysis in \cref{app:comp-cost-without-conditioning} can be evaluated at only linear cost in $\intg{T}$ (and $\intg{S}$). As in \cref{app:comp-cost-without-conditioning}, we ignore the operations associated with the reversibility check as these introduce only a constant factor overhead. Under \cref{ass:comp-cost-conditioning-partition}, we have that $\intg{T} = \intg{B}\intg{R}$ for some \emph{block size} $\intg{R}$ and we also include the scaling of operations with $\intg{R}$ in the analysis here. \cref{alg:constrained-step-with-conditioning} summarises the operations involved in each constrained integrator step when conditioning on intermediate states and the associated complexities. In the following we describe how the quadratic / cubic in $\intg{T}$ complexity operations in \cref{alg:constrained-step-no-conditioning} are altered to linear cost in $\intg{T}$.

\begin{algorithm}[!p]
  \caption{Annotated constrained integrator step (with blocking / conditioning)}
  \label{alg:constrained-step-with-conditioning}
  \vspace{-1.6pt}
\begin{algorithmic}
    \Require (current state and cached quantities)\\
        $\vct{q} = [\vct{u}; \vct{v}; \vct{w}]$ : current position, $\vct{p}$ : current momentum,
        $\vct{g} = \grad h_1(\vct{q})$,\\
        $(\mtx{U}, \mtx{V}, \mtx{W}) = \jacob\bar{\vctfunc{c}}(\vct{u},\vct{v},\vct{w})$,
        $\mtx{D} = \mtx{V} \mtx{M}_v^{-1} \mtx{V}\tr + \mtx{W} \mtx{M}_w^{-1} \mtx{W}\tr$,
        $\mtx{C} = \mtx{M}_u + \mtx{U}\tr \mtx{D}^{-1}\mtx{U}$.
    \Ensure (next state and cached quantities)\\
        $\bar{\vct{q}} = [\bar{\vct{u}}; \bar{\vct{v}}; \bar{\vct{w}}]$ : next position, $\vct{p}$ : next momentum,
        $\bar{\vct{g}} = \grad h_1(\bar{\vct{q}})$,\\
        $(\bar{\mtx{U}}, \bar{\mtx{V}}, \bar{\mtx{W}}) = \jacob\bar{\vctfunc{c}}(\bar{\vct{u}},\bar{\vct{v}},\bar{\vct{w}})$,
        $\bar{\mtx{D}} = \bar{\mtx{V}} \mtx{M}_v^{-1} \bar{\mtx{V}}\tr + \bar{\mtx{W}} \mtx{M}_w^{-1} \bar{\mtx{W}}\tr$,
        $\bar{\mtx{C}} = \mtx{M}_u + \bar{\mtx{U}}\tr \bar{\mtx{D}}^{-1}\bar{\mtx{U}}$.
\end{algorithmic}
\hrule
\begin{tabularx}{\textwidth}{@{}cXcc@{}}
  &  & Complexity $\bigo(\cdot)$ & Max. calls \\
  \rownumber & $\tilde{\vct{p}} \gets \vct{p} - \tfrac{t}{2} \vct{g}$ & $\intg{ST}$ & 1 \\
  \rownumber & $\tilde{\vct{\nu}} \gets \mtx{M}^{-1}\tilde{\vct{p}}$ & $\intg{ST}$ & 1 \\
  \rownumber & $\tilde{\vct{\omega}} \gets \jacob\vctfunc{c}(\vct{q})\tilde{\vct{\nu}}$ & $\intg{ST}$ & 1 \\
  \rownumber \label{ln:premultiply-D-inv-1} & $\tilde{\vct{\chi}} \gets \mtx{D}^{-1} \tilde{\vct{\omega}}$ & $\intg{R}^2\intg{T}$ & 1 \\
  \rownumber \label{ln:premultiply-U-tr-1} & $\tilde{\vct{\psi}} \gets \mtx{U}\tr \tilde{\vct{\chi}}$ & $\intg{T}$ & 1 \\
  \rownumber \label{ln:premultiply-C-inv-1} & $\tilde{\vct{\phi}} \gets \mtx{C}^{-1} \tilde{\vct{\psi}}$ & $1$ & 1 \\
  \rownumber \label{ln:premultiply-U-1} & $\vct{\zeta} \gets \mtx{U}\tilde{\vct{\phi}}$ & $\intg{T}$ & 1 \\
  \rownumber \label{ln:premultiply-D-inv-2} & $\vct{\lambda}_1 \gets \mtx{D}^{-1}(\tilde{\vct{\omega}} - \tilde{\vct{\zeta}})$ & $\intg{R}^2\intg{T}$ & 1 \\
  \rownumber & \textbf{for} $\intg{j} \in \range{\intg{J}}$ & & \\
    \rownumber & \quad $\vct{p}' \gets \tilde{\vct{p}} - \jacob\vctfunc{c}(\vct{q})\tr\vct{\lambda}_{\intg{j}}$ & $\intg{ST}$ & $\intg{J}$ \\
    \rownumber & \quad $(\vct{q}_{\intg{j}}, \vct{p}_{\intg{j}}) \gets (\opfunc{\Gamma}_t^{q,q}\vct{q} + \opfunc{\Gamma}_t^{q,p}\vct{p}', \opfunc{\Gamma}_t^{p,q}\vct{q} + \opfunc{\Gamma}_t^{p,p}\vct{p}')$ & $\intg{ST}$ & $\intg{J}$ \\
    \rownumber & \quad $\vct{e}_{\intg{j}} \gets \vctfunc{c}(\vct{q}_{\intg{j}})$ & $\intg{ST}$ & $\intg{J}$ \\
    \rownumber \label{ln:split-position-vector} & \quad $(\vct{u}_{\intg{j}}, \vct{v}_{\intg{j}}, \vct{w}_{\intg{j}}) \gets \textsc{split}(\vct{q}_{\intg{j}}, (\intg{U}, \intg{V}_0 + \intg{STV}, \intg{TW}))$ & & \\
    \rownumber \label{ln:constraint-jacobian-blocks} & \quad $(\mtx{U}_{\intg{j}}, \mtx{V}_{\intg{j}}, \mtx{W}_{\intg{j}}) \gets \jacob\bar{\vctfunc{c}}(\vct{u}_{\intg{j}},\vct{v}_{\intg{j}},\vct{w}_{\intg{j}})$ & $\intg{RST}$ & $\intg{J}$ \\
    \rownumber & \quad \textbf{if} $\Vert\vct{e}_{\intg{j}}\Vert < \theta_c$ \textbf{and} ($\intg{j} \equiv 1$ \textbf{or} $\Vert\vct{q}_{\intg{j}} - \vct{q}_{\intg{j}-1}\Vert < \theta_q$) & $\intg{ST}$ & $\intg{J}$ \\
      \rownumber & \qquad $(\bar{\vct{q}}, \bar{\vct{p}}, \bar{\mtx{U}}) \gets (\vct{q}_{\intg{j}}, \vct{p}_{\intg{j}}, \mtx{U}_{\intg{j}})$ & & \\[1pt]
      \rownumber \label{ln:block-diagonal-matrix-D} & \qquad $\bar{\mtx{D}} \gets \mtx{V}_{\intg{j}} \mtx{M}_v^{-1} \mtx{V}_{\intg{j}}\tr + \mtx{W}_{\intg{j}} \mtx{M}_w^{-1} \mtx{W}_{\intg{j}}\tr$ & $\intg{R}^2\intg{ST}$ & $1$ \\
      \rownumber \label{ln:capacitance-matrix-C} & \qquad $\bar{\mtx{C}} \gets \mtx{M}_u + \bar{\mtx{U}}\tr \bar{\mtx{D}}^{-1}\bar{\mtx{U}}$ & $\intg{R}^2\intg{T}$ & $1$ \\
      \rownumber & \qquad \textbf{break} & & \\
    \rownumber \label{ln:block-diagonal-matrix-E} & \quad $\mtx{E}_{\intg{j}} \gets \mtx{V}_{\intg{j}} \mtx{M}_v^{-1} \mtx{V}\tr + \mtx{W}_{\intg{j}} \mtx{M}_w^{-1} \mtx{W}\tr$ & $\intg{R}^2\intg{ST}$ & $\intg{J}$ \\
    \rownumber \label{ln:premultiply-E-inv-1} & \quad $\vct{\chi}_{\intg{j}} \gets \mtx{E}_{\intg{j}}^{-1} \vct{e}_{\intg{j}}$ & $\intg{R}^2\intg{T}$ & $\intg{J}$ \\
    \rownumber \label{ln:premultiply-U-tr-2} & \quad $\vct{\psi}_{\intg{j}} \gets \mtx{U}\tr \vct{\chi}_{\intg{j}}$ & $\intg{T}$ & $\intg{J}$ \\
    \rownumber \label{ln:premultiply-C-inv-2} & \quad $\vct{\phi}_{\intg{j}} \gets (\mtx{M}_u + \mtx{U}\tr \mtx{E_{\intg{j}}}^{-1}\mtx{U}_{\intg{j}})^{-1} \vct{\psi}_{\intg{j}}$ & $\intg{R}^2\intg{T}$ & $\intg{J}$ \\
    \rownumber \label{ln:premultiply-U-2} & \quad $\vct{\zeta}_{\intg{j}} \gets \mtx{U}_{\intg{j}}\vct{\psi}_{\intg{j}}$ & $\intg{T}$ & $\intg{J}$ \\
    \rownumber \label{ln:premultiply-E-inv-2} & \quad $\vct{\lambda}_{\intg{j}+1} \gets \vct{\lambda}_{\intg{j}} + \mtx{E}^{-1}_{\intg{j}}(\vct{e}_{\intg{j}} - \vct{\zeta}_{\intg{j}})$ & $\intg{R}^2\intg{T}$ & $\intg{J}$ \\
  \rownumber \label{ln:h1-gradient-conditioned} & $\bar{\vct{g}} \gets \grad h_1(\bar{\vct{q}})$ & $\intg{R}^2\intg{S}\intg{T}$ & 1\\
  \rownumber & $\tilde{\vct{p}} \gets \bar{\vct{p}} - \tfrac{t}{2} \bar{\vct{g}}$ & $\intg{ST}$ & 1 \\
  \rownumber & $\tilde{\vct{\nu}} \gets \mtx{M}^{-1}\tilde{\vct{p}}$ & $\intg{ST}$ & 1 \\
  \rownumber  & $\tilde{\vct{\omega}} \gets \jacob\vctfunc{c}(\vct{q})\tilde{\vct{\nu}}$ & $\intg{ST}$ & 1 \\
  \rownumber \label{ln:premultiply-D-inv-3} & $\tilde{\vct{\chi}} \gets \bar{\mtx{D}}^{-1} \tilde{\vct{\omega}}$ & $\intg{R}^2\intg{T}$ & 1 \\
  \rownumber \label{ln:premultiply-U-tr-3} & $\tilde{\vct{\psi}} \gets \bar{\mtx{U}}\tr \tilde{\vct{\chi}}$ & $\intg{T}$ & 1 \\
  \rownumber \label{ln:premultiply-C-inv-3} & $\tilde{\vct{\phi}} \gets \bar{\mtx{C}}^{-1} \tilde{\vct{\psi}}$ & $1$ & 1 \\
  \rownumber \label{ln:premultiply-U-3} & $\tilde{\vct{\zeta}} \gets \bar{\mtx{U}}\tilde{\vct{\phi}}$ & $\intg{T}$ & 1 \\
  \rownumber \label{ln:premultiply-D-inv-4} & $\tilde{\vct{\lambda}} \gets \bar{\mtx{D}}^{-1}(\tilde{\vct{\omega}} - \tilde{\vct{\zeta}})$ & $\intg{R}^2\intg{T}$ & 1 \\
  \rownumber & $\bar{\vct{p}} \gets \tilde{\vct{p}} - \jacob\vctfunc{c}(\bar{\vct{q}})\tr\tilde{\vct{\lambda}}$ & $\intg{ST}$ & 1 \\
\end{tabularx} \end{algorithm}

Under \cref{ass:comp-cost-conditioning-partition}, each of the $\intg{B}$ partial constraint functions $\vctfunc{c}_{\range{\intg{B}}}$ outputs a vector of dimension $\intg{R}\intg{Y} + \intg{X} = \bigo(\intg{R})$ and requires $\intg{R}\intg{S}$ evaluations of $\vctfunc{f}_\delta$ and $\intg{R}$ evaluations of $\vctfunc{h}$ and so has a complexity of $\bigo(\intg{R}\intg{S})$. By \cref{cor:ad-comp-cost-jacobian} the partial Jacobians $(\partial_1\vctfunc{c}_{\intg{b}}(\vct{u},\bar{\vct{v}}_{\intg{b}},\bar{\vct{w}}_\intg{b}), \partial_2\vctfunc{c}_{\intg{b}}(\vct{u},\bar{\vct{v}}_{\intg{b}},\bar{\vct{w}}_\intg{b}),\partial_3\vctfunc{c}_{\intg{b}}(\vct{u},\bar{\vct{v}}_{\intg{b}},\bar{\vct{w}}_\intg{b}))$ can therefore be evaluated for each $\intg{b} \in \range{\intg{B}}$ at a $\bigo(\intg{R}^2\intg{S})$ complexity, resulting in an overall $\bigo(\intg{R}\intg{S}\intg{T})$ complexity to evaluate all $\intg{B}$ tuples of partial Jacobians. These partial Jacobians correspond to the non-zero blocks of the full constraint Jacobian $\jacob\vctfunc{c}(\vct{u};\bar{\vct{v}}_{\range{\intg{B}}};[\bar{\vct{w}}_{\range{\intg{B}}}]) = [\mtx{U}~\mtx{V}~\mtx{W}]$,
\begin{align*}
  \mtx{U}
  =
  \partial_1\bar{\vctfunc{c}}(\vct{u},[\bar{\vct{v}}_{\range{\intg{B}}}],[\bar{\vct{w}}_{\range{\intg{B}}}])
  &=
  [(\partial_1\vctfunc{c}_{\intg{b}}(\vct{u},\bar{\vct{v}}_{\intg{b}},\bar{\vct{w}}_\intg{b}))_{\intg{b}=1}^{\intg{B}}],
  \\
  \mtx{V} =
  \partial_2\bar{\vctfunc{c}}(\vct{u},[\bar{\vct{v}}_{\range{\intg{B}}}],[\bar{\vct{w}}_{\range{\intg{B}}}])
  &=
  \diag(\partial_2\vctfunc{c}_{\intg{b}}(\vct{u},\bar{\vct{v}}_{\intg{b}},\bar{\vct{w}}_\intg{b}))_{\intg{b}=1}^{\intg{B}}],
  \\
  \mtx{W}
  =
  \partial_3\bar{\vctfunc{c}}(\vct{u},[\bar{\vct{v}}_{\range{\intg{B}}}],[\bar{\vct{w}}_{\range{\intg{B}}}])
  &=
  \diag(\partial_3\vctfunc{c}_{\intg{b}}(\vct{u},\bar{\vct{v}}_{\intg{b}},\bar{\vct{w}}_\intg{b}))_{\intg{b}=1}^{\intg{B}}],
\end{align*}
with $\mtx{U}$ a $\intg{C} \times \intg{U}$ matrix, $\mtx{V}$ a $\intg{C} \times (\intg{V}_0 + \intg{STV})$ block-diagonal matrix and $\mtx{W}$ a $\intg{C} \times \intg{TW}$ block-diagonal matrix. Evaluating \cref{ln:constraint-jacobian-blocks} therefore has complexity $\bigo(\intg{R}\intg{S}\intg{T})$.

As $(\rvct{u},\rvct{v}_{\range[0]{\intg{ST}}}, \rvct{w}_{\range{\intg{T}}})$ are independent under the prior $\rho$, under the recommendation in \cref{subsec:metric-choice} the metric matrix representation is $\mtx{M} = \diag(\mtx{M}_u,\mtx{M}_v,\mtx{M}_w)$ with $\mtx{M}_u$ a $\intg{U}\times\intg{U}$ matrix, $\mtx{M}_v$ a $(\intg{V}_0 + \intg{S}\intg{T}\intg{V})\times(\intg{V}_0 + \intg{S}\intg{T}\intg{V})$ block-diagonal matrix and $\mtx{M}_w$ a $\intg{T}\intg{W}\times\intg{T}\intg{W}$ block-diagonal matrix. The matrix $\mtx{M}_v$ further decomposes as $\mtx{M}_v = \diag(\mtx{M}_{v_0}, \idmtx_{\intg{STV}})$, where $\mtx{M}_{v_0}$ is the prior precision matrix of $\rvct{v}_0 \sim \tilde{\nu}$, due to the standard normal prior distributions of $\rvct{v}_{\range{\intg{ST}}}$. The matrix product $\mtx{V}\mtx{M}_v$ can therefore be computed at a $\bigo(\intg{ST})$ complexity. The product $(\mtx{V} \mtx{M}_v)\mtx{V}\tr$ can be decomposed into $\intg{B}$ products of a dense blocks of size $\bigo(\intg{R}) \times \bigo(\intg{RS})$ and $\bigo(\intg{RS}) \times \bigo(\intg{R})$, with overall complexity $\bigo(\intg{R^2}\intg{ST})$. The block structure of the rectangular block-diagonal matrix $\mtx{W}$ is compatible with that of $\mtx{M}_w$ such that the matrix product $\mtx{W}\mtx{M}_w\mtx{W}\tr$ can be decomposed into $\intg{B}$ products of blocks of size $\bigo(\intg{R}) \times \bigo(\intg{R})$ resulting in an overall complexity $\bigo(\intg{R}^2\intg{T})$. Combining the previous results together, the complexity of evaluating a symmetric block-diagonal matrix $\mtx{D} = \mtx{V}\mtx{M}_v\mtx{V}\tr + \mtx{W}\mtx{M}_w\mtx{W}\tr$ (\cref{ln:block-diagonal-matrix-D}) is therefore $\bigo(\intg{R}^2\intg{ST})$ and by an equivalent argument the complexity of evaluating a non-symmetric block-diagonal matrix $\mtx{E} = \mtx{V}'\mtx{M}_v\mtx{V}\tr + \mtx{W}'\mtx{M}_w\mtx{W}\tr$ (\cref{ln:block-diagonal-matrix-E}) is also $\bigo(\intg{R}^2\intg{ST})$. As $\mtx{D}$ and $\mtx{E}$ are both block-diagonal with $\intg{B}$ blocks of size $\bigo(\intg{R}) \times \bigo(\intg{R})$, multiplying a vector by their inverses (\cref{ln:premultiply-D-inv-1,ln:premultiply-D-inv-2,ln:premultiply-D-inv-3,ln:premultiply-D-inv-4,ln:premultiply-E-inv-1,ln:premultiply-E-inv-2}) has a $\bigo(\intg{R}^2\intg{T})$ complexity.

As $\mtx{U}$ is of size $\bigo(\intg{T}) \times \bigo(1)$, by the same argument computing $\mtx{D}^{-1}\mtx{U}$ has complexity $\bigo(\intg{R}^2\intg{T})$. Computing $\mtx{C} = \mtx{M}_u + \mtx{U}\tr \mtx{D}^{-1} \mtx{U}$ therefore has complexity $\bigo(\intg{R}^2\intg{T})$ (\cref{ln:capacitance-matrix-C}) and the same complexity applies to the non-symmetric equivalent $\mtx{M}_u + \mtx{U}\tr \mtx{D}^{-1} \mtx{U}'$ (\cref{ln:premultiply-C-inv-2}). As $\mtx{C}$ is of size $\bigo(1) \times \bigo(1)$, premultiplying by its inverse (\cref{ln:premultiply-C-inv-1,ln:premultiply-C-inv-3}) has $\bigo(1)$ complexity. Premultiplying a vector with $\mtx{U}$ (\cref{ln:premultiply-U-1,ln:premultiply-U-2,ln:premultiply-U-3}) or $\mtx{U}\tr$ (\cref{ln:premultiply-U-tr-1,ln:premultiply-U-tr-2,ln:premultiply-U-tr-3}) has $\bigo(\intg{T})$ complexity.

Assuming the gradient of $h_1$ is computed using \ac{ad}, by \cref{cor:ad-comp-cost-jacobian} the cost is again the same as evaluating $h_1(\vct{q})$. For both the Gaussian and St\"ormer--Verlet splittings, the dominant cost in evaluating $h_1(\vct{q})$ is in computing the Gram matrix log determinant $\log\det{\opfunc{G}_{\mtx{M}}(\vct{q})}$; using the identity in \eqref{eq:conditioned-gram-determinant} this can be reduced to computing the log determinants of the matrices $\mtx{C}$ and $\mtx{D}$. Evaluating $\mtx{C}$ and $\mtx{D}$ has complexities $\bigo(\intg{R}^2\intg{T})$ and $\bigo(\intg{R}^2\intg{ST})$ respectively, as described previously. As $\mtx{C}$ is of size $\bigo(1) \times \bigo(1)$ computing its log determinant has $\bigo(1)$ complexity, while as $\mtx{D}$ is block-diagonal with $\intg{B}$ blocks of size $\bigo(\intg{R}) \times \bigo(\intg{R})$, its log determinant can be evaluated as the sum of the log determinants of the blocks at $\bigo(\intg{R}^2\intg{T})$ complexity. The overall complexity for evaluating $h_1$ and so $\grad h_1$ (\cref{ln:h1-gradient-conditioned}) is therefore $\bigo(\intg{R}^2\intg{ST})$.

The complexity of each constrained step is therefore $\bigo(\intg{J}\intg{R}^2\intg{S}\intg{T})$, with \cref{ln:block-diagonal-matrix-E} the dominant cost. Under \cref{ass:comp-cost-newton-iterations-conditioned}, $\intg{J} = \bigo(1)$ and the complexity is $\bigo(\intg{R}^2\intg{S}\intg{T})$. \qed
\section{Proof of \cref*{prop:manifold-density}}
\label{app:proof-manifold-density}

To prove the posterior distribution has a density of the form given in \cref*{prop:manifold-density} we use the co-area formula \cite{federer1969geometric}, an extension of Fubini's theorem.
\begin{lemma}[Co-area formula]
  \label{lem:co-area-formula}
  For any measurable function $f : \reals^{\intg{Q}} \to \reals$ and continuously differentiable function $\vctfunc{c} : \reals^{\intg{Q}} \to \reals^\intg{C}$ with $\intg{D} = \intg{Q} - \intg{C} > 0$  then
    $$
      \int_{\reals^\intg{Q}} f(\vct{q}) \det{\jacob\vctfunc{c}(\vct{q})\mtx{M}^{-1}\jacob\vctfunc{c}(\vct{q})}^{\frac{1}{2}} \det{\mtx{M}}^{\frac{1}{2}} \lebm{\intg{Q}}(\diff\vct{q}) =
      \int_{\reals^\intg{C}} \int_{\vctfunc{c}^{-1}[\vct{y}]} f(\vct{q}) \haum{\mtx{M}}{\intg{D}}(\diff\vct{q}) \,\lebm{\intg{C}}(\diff\vct{y})
    $$
  where $\haum{\mtx{M}}{\intg{D}}$ is the $\intg{D}$-dimensional Hausdorff measure on $\reals^{\intg{Q}}$ equipped with a metric with positive definite matrix representation $\mtx{M} \in \reals^{\intg{Q}\times\intg{Q}}$ and $\vctfunc{c}^{-1}[\vct{y}] = \lbrace \vct{q}\in\reals^\intg{Q} : \vctfunc{c}(\vct{q}) = \vct{y} \rbrace$ is the preimage of $\vct{y}$ under $\vctfunc{c}$.
\end{lemma}
\begin{proof}
  See for example Theorem 1 in \S 3.4.2 in \citet{evans1992measure}. Compared to the standard statement, the result here includes a change of variables $\vct{q}' = \mtx{M}^{-\frac{1}{2}}\vct{q}$ with corresponding transform in the Euclidean metric $\vct{q}\tr\vct{q} = (\vct{q}')\tr\mtx{M}\vct{q}'$.
\end{proof}
Although \cref{lem:co-area-formula} states the co-area formula in terms of a Hausdorff measure on the ambient space, an equivalent result can be stated in terms of a Riemannian measure on the submanifold using the correspondence in \cref{lem:equivalence-of-riemannian-and-hausdorff-measures} below.
\begin{definition}
  \label{def:riemannian-measure}
  Let $\set{M}$ be a $\intg{D}$-dimensional $\set{C}^1$ sub-manifold of an ambient space $\reals^\intg{Q}$ equipped with a metric with positive-definite matrix representation $\mtx{M}$. Then for any local parametrisation $\vctfunc{\phi}: \reals^\intg{D} \supseteq \set{U} \to \reals^\intg{Q}$ of $\set{M}$ we define $\riem{\mtx{M}}{\set{M}}$, the Riemannian measure on $\set{M}$ equipped with metric induced from the ambient metric, of a measurable $\set{A} \subseteq \vctfunc{\phi}(\set{U})$ as
  $$
    \riem{\mtx{M}}{\set{M}}(\set{A}) :=
    \int_{\vctfunc{\phi}^{-1}(\set{A})} \left| \jacob\vctfunc{\phi}(\vct{m})\tr\mtx{M}\jacob\vctfunc{\phi}(\vct{m})\right|^{\frac{1}{2}} \lebm{\intg{D}}(\diff\vct{m}).
  $$
\end{definition}
This definition can be extended to any $\set{A} \subseteq \set{M}$ using a partition of unity argument.
\begin{lemma}
  \label{lem:equivalence-of-riemannian-and-hausdorff-measures}
  Let $\set{M}$ be a $\intg{D}$-dimensional sub-manifold of an ambient space $\reals^\intg{Q}$. For any positive definite matrix $\mtx{M}$ and measurable $\set{A} \subseteq \set{M}$ we have
  $$
      \haum{\mtx{M}}{\intg{D}}(\set{A}) = \riem{\mtx{M}}{\set{M}}(\set{A}).
  $$
\end{lemma}
\begin{proof}
  See for example Theorem IV.1.8 in \citet{chavel2001isoperimetric} (note that the definition in \citet{chavel2001isoperimetric} of the Hausdorff measure differs by a normalising constant).
\end{proof}

We are now in a position to prove \cref*{prop:manifold-density} in the main text which we restate in a self-contained form below for ease of reference.
\begin{proposition}
  If $\rvct{q} \sim \rho$ for $\rho \in \borelprob(\reals^\intg{Q})$ absolutely continuous to $\lebm{\intg{Q}}$ and $\vctfunc{c} : \reals^\intg{Q} \to \reals^\intg{C}$ is of class $\mathcal{C}^1$ with $\jacob\vctfunc{c}$ full row-rank $\rho$-almost surely. Then
  $$
    \rvct{q} \gvn \vctfunc{c}(\rvct{q}) = \vct{0} \sim \pi,
    \quad
    \td{\pi}{\riem{\mtx{M}}{\set{M}}}(\vct{q}) \propto
    \td{\rho}{\lebm{\intg{Q}}}(\vct{q}) \det{\jacob\vctfunc{c}(\vct{q})\mtx{M}^{-1}\jacob\vctfunc{c}(\vct{q})\tr}^{-\frac{1}{2}},
    \quad
    \set{M} = \vctfunc{c}^{-1}[\vct{0}].
  $$
\end{proposition}
\begin{proof}
Let $f(\vct{q}) = h(\vct{q})\,\td{\rho}{\lebm{\intg{Q}}}(\vct{q}) \det{\jacob\vctfunc{c}(\vct{q})\mtx{M}^{-1}\jacob\vctfunc{c}(\vct{q})}^{-\frac{1}{2}} \det{\mtx{M}}^{-\frac{1}{2}}$ for a measurable function $h$. Then from \cref{lem:co-area-formula} and the equivalence in \cref{lem:equivalence-of-riemannian-and-hausdorff-measures} we have
\begin{multline*}
    \int_{\reals^\intg{Q}} h(\vct{q}) \rho(\diff\vct{q}) =
    \int_{\reals^\intg{C}} \int_{\vctfunc{c}^{-1}[\vct{y}]}
      \mkern-3mu h(\vct{q})
      \frac{
        \td{\rho}{\lebm{\intg{Q}}}(\vct{q})
        \det{
          \jacob\vctfunc{c}(\vct{q})
          \mtx{M}^{-1}
          \jacob\vctfunc{c}(\vct{q})\tr
        }^{-\frac{1}{2}}
      }{\det{\mtx{M}}^{\frac{1}{2}}}
      \riem{\mtx{M}}{\vctfunc{c}^{-1}[\vct{y}]}(\diff\vct{q}) \,\lebm{\intg{C}}(\diff\vct{y}).
\end{multline*}
Define $w(\vct{y}) =  \int_{\vctfunc{c}^{-1}[\vct{y}]} \td{\rho}{\lebm{\intg{Q}}}(\vct{q}) \det{\jacob\vctfunc{c}(\vct{q})\mtx{M}^{-1}\jacob\vctfunc{c}(\vct{q})}^{-\frac{1}{2}} \det{\mtx{M}}^{-\frac{1}{2}} \riem{\mtx{M}}{\vctfunc{c}^{-1}[\vct{y}]}(\diff\vct{q})$ then from \\ \cref{lem:co-area-formula} we have $\int_{\reals^{\intg{C}}} w(\vct{y}) \lebm{\intg{C}}(\diff\vct{y}) = \int_{\reals^{\intg{Q}}} \rho(\diff\vct{q}) = 1$ and%
\begin{multline*}
    \int_{\reals^\intg{Q}} h(\vct{q}) \rho(\diff\vct{q}) =\\
    \int_{\reals^\intg{C}} \int_{\vctfunc{c}^{-1}[\vct{y}]} %
      h(\vct{q})
      \frac{
        \td{\rho}{\lebm{\intg{Q}}}(\vct{q})
        \det{
          \jacob\vctfunc{c}(\vct{q})
          \mtx{M}^{-1}
          \jacob\vctfunc{c}(\vct{q})\tr
        }^{-\frac{1}{2}}
      }{\det{\mtx{M}}^{\frac{1}{2}}w(\vct{y})}
    \riem{\mtx{M}}{\vctfunc{c}^{-1}[\vct{y}]}(\diff\vct{q})
    \, w(\vct{y})\,\lebm{\intg{C}}(\diff\vct{y}).
\end{multline*}
As this applies for arbitrary $h$, defining $\omega: \borelprob(\reals^\intg{C})$ as the marginal law of $\vctfunc{c}(\rvct{q})$ and $\varpi : \reals^\intg{C} \to \borelprob(\real^\intg{Q})$ the law of $\rvct{q}$ given $\vctfunc{c}(\rvct{q})$, comparing with the law of total expectation
$$
  \int_{\reals^\intg{Q}} h(\vct{q}) \rho(\diff\vct{q}) =
  \int_{\reals^\intg{C}} \int_{\vctfunc{c}^{-1}[\vct{y}]} h(\vct{q}) \varpi(\vct{y})(\diff\vct{q})\, \omega(\diff \vct{y})
$$
we recognise $\omega(\diff\vct{y}) = w(\vct{y})\,\lebm{\intg{C}}(\diff\vct{y})$ and
$$
  \varpi(\vct{y})(\diff\vct{q}) =
  (w(\vct{y})\det{\mtx{M}}^{\frac{1}{2}})^{-1}
  \td{\rho}{\lebm{\intg{Q}}}(\vct{q})
  \det{
    \jacob\vctfunc{c}(\vct{q})
    \mtx{M}^{-1}
    \jacob\vctfunc{c}(\vct{q})\tr
  }^{-\frac{1}{2}}
  \riem{\mtx{M}}{\vctfunc{c}^{-1}[\vct{y}]}(\diff\vct{q}).
$$
The density given for $\pi = \varpi(\vct{0})$ in the proposition then follows directly.
\end{proof}
\section{Non-degeneracy of the symplectic form}
\label{app:non-degeneracy-of-symplectic-form}

For a general set of constraints $\phi_{\intg{i}}: \reals^{\intg{Q}} \times \reals^{\intg{Q}} \to \reals,~\intg{i} \in 1:\intg{M}$ on both the position and momenta, there is no guarantee that the restriction of the symplectic form is non-degenerate. In general, the restricted symplectic form will be non-degenerate if and only if all the constraints are \emph{second class} \citep{dirac2001lectures,bojowald2003poisson} which requires that each constraint has a non-zero \emph{Poisson bracket} with at least one other constraint, that is, for each $\intg{i} \in 1{:}\intg{M}$ there exists $\intg{j} \neq \intg{i}$ such that
\[
    \lbrace \phi_{\intg{i}}, \phi_{\intg{j}}\rbrace(\vct{q}, \vct{p}) =
    \partial_1\phi_{\intg{i}}(\vct{q},\vct{p})\tr\partial_2\phi_{\intg{j}}(\vct{q},\vct{p}) -
    \partial_2\phi_{\intg{i}}(\vct{q},\vct{p})\tr\partial_1\phi_{\intg{j}}(\vct{q},\vct{p})
    \neq 0,~
    \forall (\vct{q}, \vct{p}) \in \reals^{\intg{Q}\times\intg{Q}}.
\]

For arbitrary constraints, particularly those involving both position and momentum, this may not be the case however we consider only the special case of \emph{holonomic} primary constraints \(\phi_{\intg{i}}(\vct{q}, \vct{p}) = c_{\intg{i}}(\vct{q})\) with \(c_{\intg{i}} : \reals^{\intg{Q}} \to \reals, ~ \intg{i} \in 1{:}\intg{C}\) with a set of secondary constraints \(\phi_{\intg{i}+\intg{C}}(\vct{q}, \vct{p}) = \partial c_{\intg{i}}(\vct{q})\tr\mtx{M}^{-1}\vct{p}, ~ \intg{i} \in 1{:}\intg{C}\), due to the requirement that the time derivative of the constraint is zero. We then have that
\[
  \lbrace \phi_{\intg{i}}, \phi_{\intg{i}+\intg{C}} \rbrace(\vct{q}, \vct{p}) = \partial c_{\intg{i}}(\vct{q})\tr\mtx{M}^{-1}\partial c_{\intg{i}}(\vct{q}) > 0
\]
that is the Poisson bracket between each primary constraint and the corresponding secondary constraint is non-zero almost everywhere under \cref{ass:differentiability-and-surjectivity-constraint-function} (that $\jacob\vctfunc{c}$ is full row-rank $\rho$-almost everywhere) and \cref{ass:latent-space-metric} (that $\mtx{M}$ is positive definite). Hence all the constraints are second-class and the restriction of the symplectic form is non-degenerate. %
\section{Proof of \cref*{prop:hamiltonian-conservation}}
\label{app:proof-hamiltonian-conservation}

It is sufficient to show that the time-derivative of the Hamiltonian is zero under the dynamics described by the flow map $\opfunc{\Phi}^{h_{\vctfunc{c}}}_t$. We have that
\begin{align*}
  \td{h}{t}(\vct{q},\vct{p}) 
  &= 
  \grad_1 h(\vct{q},\vct{p})\tr \td{\vct{q}}{t} + \grad_2 h(\vct{q},\vct{p})\tr \td{\vct{p}}{t} \\
  &=
  \grad\ell(\vct{q})\tr \mtx{M}^{-1}\vct{p} - \vct{p}\tr\mtx{M}^{-1}\grad\ell(\vct{q}) - \vct{p}\tr\mtx{M}^{-1}\jacob\vctfunc{c}(\vct{q})\tr\vctfunc{\lambda}(\vct{p},\vct{q}) \\
  &=
  -\vctfunc{\lambda}(\vct{p},\vct{q})\tr\jacob\vctfunc{c}(\vct{q})\mtx{M}^{-1}\vct{p}
  = 0
\end{align*}
for all $(\vct{p},\vct{q}) \in \set{T}^*\set{M}$ as $\jacob\vctfunc{c}(\vct{q})\mtx{M}^{-1}\vct{p} = \vct{0}$. \qed
\section{Proof of \cref*{prop:symplecticness-of-constrained-dynamic}}
\label{app:proof-symplecticness-of-constrained-dynamic}

We will make use of the following definitions and lemma.

\begin{definition}
  \label{def:wedge-product}
  The bilinear and skew-symmetric \emph{wedge product} between the \emph{differentials} $\diff q_{\intg{i}} : \reals^{\intg{Q}} \to \reals$ and $\diff p_{\intg{j}}: \reals^{\intg{Q}} \to \reals$ is
  \[(
    \diff q_{\intg{i}} \wedge \diff p_{\intg{j}})(\vct{u}, \vct{v}) := \diff p_{\intg{j}}(\vct{u})\diff q_{\intg{i}}(\vct{v}) - \diff q_{\intg{i}}(\vct{u})\diff p_{\intg{j}}(\vct{v}),
  \]
  and between the vectors of differentials $\diff \vct{q} = [(\diff q_{\intg{i}})_{\intg{i}=1}^\intg{Q}]$ and $\diff \vct{p} = [(\diff p_{\intg{i}})_{\intg{i}=1}^\intg{Q}]$ is
  \(
    \diff\vct{q} \wedge \diff\vct{p} :=
    {\textstyle \sum_{\intg{i}=1}^{\intg{Q}}} \diff q_\intg{i} \wedge \diff p_\intg{i},
  \)
  with $\diff\vct{q} \wedge \diff\vct{p}$ equal to the \emph{symplectic form} on $\reals^{\intg{Q}}\times\reals^{\intg{Q}}$.
\end{definition}

\begin{definition}
A map $\opfunc{\Phi} : \tangent^*\set{M} \to \tangent^*\set{M}$ is \emph{symplectic} if the symplectic form on $\tangent^*\set{M}$, that is the restriction of $\diff\vct{q} \wedge \diff\vct{p}$ to $\tangent^*\set{M}$, is preserved by the map.

\end{definition}

\begin{lemma}
 \label{lem:symplectic-intermediate}
 The vector of differential 1-forms $\diff\vct{q}$ on the manifold $\set{M}$ satisfies $\diff\vct{q}\wedge\diff\left(\jacob\vctfunc{c}(\vct{q})\tr\vctfunc{\lambda}(\vct{q},\vct{p})\right) = 0$ for arbitrary $\vctfunc{\lambda} : \set{M}\times\reals^{\intg{Q}} \to \reals^{\intg{C}}$.
\end{lemma}

\begin{proof}
  Omitting arguments to $\vctfunc{c}$ and $\vctfunc{\lambda}$ for compactness we have that
\begin{align*}
    \diff\vct{q}\wedge\diff(\jacob\vctfunc{c}\tr\vctfunc{\lambda}) &=
    \diff\vct{q}\wedge\jacob\vctfunc{c}\tr\diff\vctfunc{\lambda} +
    {\textstyle \sum_{\intg{i}=1}^{\intg{C}}} \lambda_{\intg{i}}\diff\vct{q} \wedge\hess c_{\intg{i}}\diff\vct{q}\\
    &=
    \jacob\vctfunc{c} \diff\vct{q}\wedge\diff\vctfunc{\lambda} +
    {\textstyle \sum_{\intg{i}=1}^{\intg{C}}} \lambda_{\intg{i}}\diff\vct{q} \wedge\hess c_{\intg{i}}\diff\vct{q}.
\end{align*}
For $\vct{q}$ restricted to $\set{M}$ and so satisfying the constraint equation $\vctfunc{c}(\vct{q}) = \vct{0}$ the vector of differential 1-forms $\diff\vct{q}$ satisfies $\jacob\vctfunc{c} \diff\vct{q} = \vct{0}$. Further the Hessians $\hess c_{\range{\intg{C}}}$ are all symmetric therefore $\diff\vct{q} \wedge\hess c_{\intg{i}}\diff\vct{q} = 0$ for all $\intg{i}\in\range{\intg{C}}$ due to the skew-symmetry of the wedge product. Therefore $\diff\vct{q}\wedge\diff(\jacob\vctfunc{c}\tr\vctfunc{\lambda}) = 0$ as required. \qed
\end{proof}

To prove \cref*{prop:symplecticness-of-constrained-dynamic} it is sufficient to show that the time-derivative of $\diff\vct{q}\wedge\diff\vct{p}$ is identical to zero under the dynamics described by the flow map $\opfunc{\Phi}^{h_{\vctfunc{c}}}_t$ as from this preservation of the symplectic form directly follows. We have that
\begin{align*}
  \td{}{t}(\diff\vct{q}\wedge\diff\vct{p}) &=
  \diff\left(\td{\vct{q}}{t}\right)\wedge\diff\vct{p} + \diff\vct{q}\wedge\diff\left(\td{\vct{p}}{t}\right)\\
  &=
  \diff(\mtx{M}^{-1}\vct{p})\wedge\diff\vct{p} - \diff\vct{q}\wedge\diff(\grad\ell(\vct{q}) + \jacob\vctfunc{c}(\vct{q})\tr\vctfunc{\lambda})\\
  &=
  \mtx{M}^{-1}\diff\vct{p}\wedge\diff\vct{p} - \diff\vct{q} \wedge \hess\ell(\vct{q})\diff\vct{q} - \diff\vct{q}\wedge\diff(\jacob\vctfunc{c}(\vct{q})\tr\vctfunc{\lambda}).
\end{align*}
As both $\mtx{M}^{-1}$ and $\hess\ell$ are symmetric matrices then from the skew-symmetry of the wedge product the first two terms in the last line are zero as is the third term from \cref*{lem:symplectic-intermediate}. Therefore $\td{}{t}(\diff\vct{q}\wedge\diff\vct{p}) = 0$. \qed
\section{Proof of \cref*{prop:momentum-resampling-correctness}}
\label{app:proof-momentum-resampling-correctness}

For $\tilde{\rvct{p}} \sim \nrm(\vct{0},\mtx{M})$ and $\rvct{p} = \opfunc{P}_{\mtx{M}}(\rvct{q})\tilde{\rvct{p}}$ we have for any measurable $\set{A} \subseteq \reals^\intg{Q}$
$$
  \prob(\rvct{p} \in \set{A} \gvn \rvct{q} = \vct{q}) \propto
  \int_{\reals^{\intg{Q}}}
    \ind{\set{A}}(\opfunc{P}_{\mtx{M}}(\vct{q})\tilde{\vct{p}})
    \exp\left(-\frac{1}{2}\tilde{\vct{p}}\tr\mtx{M}^{-1}\tilde{\vct{p}}\right)
    \lebm{\intg{Q}}(\diff\tilde{\vct{p}}).
$$
As $\mtx{M}$ is positive definite it has a non-singular symmetric square-root $\mtx{M}^{\frac{1}{2}}$. Further, as $\jacob\vctfunc{c}(\vct{q})$ is full row-rank $\rho$-almost surely, then we can find a decomposition
$$
  \mtx{M}^{-\frac{1}{2}}\jacob\vctfunc{c}(\vct{q}) =
  \begin{bmatrix}\mtx{Q}_{\bot} &\mtx{Q}_{\|}\end{bmatrix}
  \begin{bmatrix}\mtx{R} \\ \mtx{0}\end{bmatrix} =
  \mtx{Q}_{\bot}\mtx{R}
$$
where $\mtx{Q}_{\bot}$ and $\mtx{Q}_{\|}$ are respectively $\intg{Q}\times\intg{C}$ and $\intg{Q}\times(\intg{Q}-\intg{C})$ matrices with orthonormal columns (i.e. $\mtx{Q}_{\perp}\tr\mtx{Q}_{\perp}=\idmtx_{\intg{C}}$, $\mtx{Q}_{\|}\tr\mtx{Q}_{\|}=\idmtx_{\intg{Q}-\intg{C}}$ and $\mtx{Q}_{\perp}\tr\mtx{Q}_{\|}=\mtx{0}$) and $\mtx{R}$ is a non-singular $\intg{C}\times\intg{C}$ upper-triangular matrix. From the definition of $\opfunc{P}_{\mtx{M}}$ in \cref*{eq:projector-operator} we then have that
$$
  \opfunc{P}_{\mtx{M}}(\vct{q}) =
  \idmtx_{\intg{Q}} -
  \mtx{M}^{\frac{1}{2}}\mtx{Q}_{\bot}\mtx{R}(\mtx{R}\tr\mtx{Q}_{\bot}\tr\mtx{Q}_{\bot}\mtx{R})^{-1}\mtx{R}\tr\mtx{Q}_{\bot}\tr\mtx{M}^{-\frac{1}{2}} =
  \idmtx_{\intg{Q}} - \mtx{M}^{\frac{1}{2}}\mtx{Q}_{\bot}\mtx{Q}_{\bot}\tr\mtx{M}^{-\frac{1}{2}}.
$$
Defining the linear change of variables $\tilde{\vct{p}} = \mtx{M}^{\frac{1}{2}}\mtx{Q}_{\bot}\vct{n} + \mtx{M}^{\frac{1}{2}}\mtx{Q}_{\|}\vct{m}$ we then have that $\opfunc{P}_{\mtx{M}}(\vct{q})\tilde{\vct{p}} = (\idmtx_{\intg{Q}} - \mtx{M}^{\frac{1}{2}}\mtx{Q}_{\bot}\mtx{Q}_{\bot}\tr\mtx{M}^{-\frac{1}{2}})(\mtx{M}^{\frac{1}{2}}\mtx{Q}_{\bot}\vct{n} + \mtx{M}^{\frac{1}{2}}\mtx{Q}_{\|}\vct{m}) = \mtx{M}^{\frac{1}{2}}\mtx{Q}_{\|}\vct{m}$ and
$$
  \prob(\rvct{p} \in \set{A} \gvn \rvct{q} = \vct{q}) \propto
  \int_{\reals^{\intg{C}}} \int_{\reals^{\intg{Q}-\intg{C}}}
    \mkern-16mu
    \ind{\set{A}}(\mtx{M}^{\frac{1}{2}}\mtx{Q}_{\|}\vct{m})
    \exp\left(-\frac{1}{2}(\vct{n}\tr\vct{n} + \vct{m}\tr\vct{m})\right)
    \lebm{\intg{Q}-\intg{C}}(\diff\vct{m})
    \lebm{\intg{C}}(\diff\vct{n}).
$$
Integrating out the density on $\vct{n}$ to a constant, defining $\phi(\vct{m}) := \mtx{M}^{\frac{1}{2}}\mtx{Q}_{\|}\vct{m}$ and using $\vctfunc{\phi}(\vct{m})\tr\mtx{M}^{-1}\vctfunc{\phi}(\vct{m}) = \vct{m}\tr\vct{m}$ and $\det{\jacob\vctfunc{\phi}(\vct{m})\tr\mtx{M}^{-1}\jacob\vctfunc{\phi}(\vct{m})} = 1$ we have
$$
  \prob(\rvct{p} \in \set{A} \gvn \rvct{q} = \vct{q}) \propto
  \int_{\vctfunc{\phi}^{-1}(\set{A})}
    \det{\jacob\vctfunc{\phi}\tr\mtx{M}^{-1}\jacob\vctfunc{\phi}}^{\frac{1}{2}}
    \exp\left(-\frac{1}{2}\vctfunc{\phi}(\vct{m})\tr\mtx{M}^{-1}\vctfunc{\phi}(\vct{m})\right)
    \lebm{\intg{Q}-\intg{C}}(\diff\vct{m}).
$$
Recognising that $\vctfunc{\phi}$ defines a (global) parametrisation of $\set{T}^*_{\vct{q}}\set{M}$, comparing to the definition of the Riemannian measure $\riem{\mtx{M}^{-1}}{\set{T}^*_{\vct{q}}\set{M}}$ (see \cref{def:riemannian-measure}) we then have that
$$
  \prob(\rvct{p} \in \set{A} \gvn \rvct{q} = \vct{q}) \propto
  \int_{\set{A}}
    \exp\left(-\frac{1}{2}\vct{p}\tr\mtx{M}^{-1}\vct{p}\right)
  \riem{\mtx{M}^{-1}}{\set{T}^*_{\vct{q}}\set{M}}(\diff\vct{p}).
$$
As this holds for any measurable $\set{A}$ we have that $\rvct{p} \gvn \rvct{q} =\vct{q}$ is conditionally distributed according to $\zeta(\diff\vct{q},\diff\vct{p})\propto\exp(-\ell(\vct{q})-\vct{p}\tr\mtx{M}^{-1}\vct{p} / 2)\riem{\mtx{M}}{\set{M}}(\diff\vct{q})\riem{\mtx{M}^{-1}}{\set{T}^*_{\vct{q}}\set{M}}(\diff\vct{p})$. \qed
\section{Details of integrator implementation and relation to previous work}
\label{app:integrator-implementation-details}

In this section we give some additional details on the implementation of the constrained integrator described in the main text in \cref{subsec:constrained-hamiltonian-dynamics-numerical-discretisation}.
The iterative solver in \cref{eq:newton-iteration} is terminated once the norm of the left-hand-side of the constraint equation in \eqref{eq:h2-step-position-constraint-condition} is below a tolerance $\theta_c > 0$ and the norm of the change in position between iterations is less than a tolerance $\theta_q > 0$. A maximum of $\intg{J}$ iterations are performed with an error being raised if the solver has not converged by the end of these. The reversibility check is similarly relaxed to requiring the norm of the difference between the initial position and position computed by applying forward and then reversed steps is less than $2\theta_q$, based on the intuition that we control the error norm in the positions in forward and reversed steps to around $\theta_q$ so we expect the composition to have error norm of approximately $2 \theta_q$ or less.

In our double-precision floating-point arithmetic implementation, the $\infty$-norm was used in all cases, and $\intg{J} = 50$, $\theta_q = 10^{-8}$ and $\theta_c = 10^{-9}$. These values were found to give a good trade-off between accuracy and robustness. While smaller tolerances closer to machine precision ($\approx 10^{-16}$) can often be used successfully, in some settings we found that accumulated floating-point error, particularly in cases where the matrices involved in the linear algebra operations in each constrained integrator step were ill-conditioned, could result in the Newton solver being unable to converge within the prescribed tolerance at some points in the latent space.

Without the reversibility check, the integrator $\opfunc{\Xi}^{h_1}_{\sfrac{t}{2}} \circ R(\opfunc{\Xi}^{h_2}_{t}) \circ \opfunc{\Xi}^{h_1}_{\sfrac{t}{2}}$ is equivalent for the St\"ormer--Verlet splitting to an instance of the geodesic integrator proposed by \citet{leimkuhler2016efficient} (in their nomenclature a g-BAB composition with $K_r = 1$). Compared to the composition $R(\opfunc{\Pi}_{\vct{\lambda}'} \circ \opfunc{\Phi}^{h_1}_{\sfrac{t}{2}} \circ \opfunc{\Phi}^{h_2}_{t} \circ \opfunc{\Phi}^{h_1}_{\sfrac{t}{2}} \circ \opfunc{\Pi}_{\vct{\lambda}})$, i.e. the integrator proposed in \citet{reich1996symplectic} with a reversibility check on the whole step as analysed in \citet{lelievre2019hybrid}, the step here, $\opfunc{\Pi}_{\vct{\lambda}'} \circ \opfunc{\Phi}^{h_1}_{\sfrac{t}{2}} \circ R(\opfunc{\Pi}_{\vct{\lambda}'} \circ \opfunc{\Phi}^{h_2}_{t} \circ \opfunc{\Pi}_{\vct{\lambda}}) \circ \opfunc{\Pi}_{\vct{\lambda}'} \circ \opfunc{\Phi}^{h_1}_{\sfrac{t}{2}}$, includes additional intermediate projections of the momentum onto the co-tangent space, after the initial $\smash{\opfunc{\Phi}^{h_1}_{\sfrac{t}{2}}}$ step and after the $\smash{\opfunc{\Phi}^{h_2}_{t}}$ step.

If the iterative solver for the system of equations in \eqref{eq:h2-step-position-constraint-condition} solved in the forward and reverse $\opfunc{\Xi}^{h_2}_{t}$ steps always converged from any initialisation these additional projection steps would be redundant, as they are equivalent to omitting the extra projections but using a non-zero initialization $\vct{\lambda}_0 = \opfunc{G}_{\mtx{M}}({\vct{q}})^{-1}\jacob\vctfunc{c}({\vct{q}})\mtx{M}^{-1}{\vct{p}}$ in the iterative solver. In practice, however, convergence of the iterative solver is not guaranteed.
Empirically we observe that starting both the forward and reverse $\opfunc{\Xi}^{h_2}_{t}$ steps from position-momentum pairs in the co-tangent bundle rather than a pair in which the momentum is not necessarily in the co-tangent space (which is the case for the composition proposed by \citet{reich1996symplectic}), leads to fewer cases of rejections due to the iterative solves failing to converge or non-reversible steps being flagged for a given integrator step size, thus improving sampling efficiency. By performing the reversibility check on the `inner' $\opfunc{\Xi}^{h_2}_{t}$ sub-step, rather than the whole step as considered by \citet{lelievre2019hybrid} there is also a computational saving of avoiding additional $h_1$ gradient computations in the time-reversed step. %
\section{HMC in models with additive observation noise}
\label{app:hmc-additive-noise-models}

In models where the observations are subject to additive noise (as defined in \cref{rem:observation-models}) and $\intg{W} = \intg{Y}$, $\opfunc{L} : \set{Z} \to \reals^{\intg{Y}\times\intg{Y}}$ is non-singular $\mu$-almost surely and $\eta$ is absolutely continuous with respect to $\lebm{\intg{Y}}$, then the posterior distribution (under the non-centred parametrisation described in \cref{subsec:non-centred-parametrisation}) on $\rvct{u}$, $\rvct{v}_0$ and $\rvct{v}_{\range{\intg{ST}}}$ given observations $\rvct{y}_{\range{\intg{T}}} = \vct{y}_{\range{\intg{T}}}$ has a tractable density with respect to the Lebesgue measure.

Specifically if we define the concatenated latent state $\bar{\rvct{q}} = [\rvct{u};\rvct{v}_0;\rvct{v}_{\range{\intg{ST}}}] \in \reals^{\bar{\intg{Q}}}$ with $\bar{\intg{Q}} = \intg{U}+\intg{V}_0+\intg{STV}$ then a-priori we have that $\bar{\rvct{q}} \sim \bar{\rho}$ with Lebesgue density
\begin{equation}\label{eq:standard-hmc-prior-lebesgue-density}
    \td{\bar{\rho}}{\lebm{\bar{\intg{Q}}}}([\vct{u};\vct{v}_0;\vct{v}_{\range{\intg{ST}}}]) \propto
    \td{\tilde\mu}{\lebm{\intg{U}}}(\vct{u})
    \td{\tilde\nu}{\lebm{\intg{V}_0}}(\vct{v}_0)
    \prod_{\intg{s}=1}^{\intg{S}{\intg{T}}}
    \exp\big(
      -\tfrac{1}{2}
        \vct{v}_{\intg{s}}\tr\vct{v}_{\intg{s}}
    \big).
\end{equation}
Defining $\bar{\vct{y}}_{\range{\intg{T}}} = \vctfunc{g}_{\bar{\rvct{y}}_:}\mkern-2mu(\vct{u},\vct{v}_0, \vct{v}_{\range{\intg{ST}}})$ with $\vctfunc{g}_{\bar{\rvct{y}}_:}$ defined by
{
\vspace{3pt}
\begin{algorithmic}
  \Function{$\vctfunc{g}_{\bar{\rvct{y}}_{:}}\mkern-1mu$}
    {$\vct{u},\vct{v}_0,\vct{v}_{\range[1]{\intg{St}}}$}
  \State $\vct{z} \gets \vctfunc{g}_{\rvct{z}}(\vct{u})$
  \State $\vct{x}_0 \gets \vctfunc{g}_{\rvct{x}_0}\mkern-2mu(\vct{z},\vct{v}_0)$
  \For{$\intg{s} \in \range[1]{\intg{St}}$}
    \State $\vct{x}_{\intg{s}} \gets \vctfunc{f}_\delta(\vct{z},\vct{x}_{\intg{s}-1},\vct{v}_s)$
    \If{$\intg{s} \bmod \intg{S} \equiv 0$}
      \State $\bar{\vct{y}}_{\intg{s}/\intg{S}} \gets \vcthbar(\vct{x}_{\intg{s}})$
    \EndIf
  \EndFor
  \vspace{-3pt}
  \State \Return $\bar{\vct{y}}_{\range[1]{\intg{t}}}$
  \EndFunction
\end{algorithmic} }
\noindent the posterior $\bar{\pi}$ on $\bar{\rvct{q}}$ given $\rvct{y}_{\range{\intg{T}}} = \vct{y}_{\range{\intg{T}}}$ then has a density with respect to $\bar{\rho}$
\begin{equation}\label{eq:standard-hmc-posterior-lebesgue-density}
    \td{\bar{\pi}}{\bar{\rho}}([\vct{u};\vct{v}_0;\vct{v}_{\range{\intg{ST}}}]) \propto
    \prod_{\intg{t}=1}^{\intg{T}}
    |\opfunc{L} \circ \vctfunc{g}_{\rvct{z}}(\vct{u}))|^{-1}
    \td{\eta}{\lebm{\intg{Y}}}\big(
      (\opfunc{L} \circ \vctfunc{g}_{\rvct{z}}(\vct{u}))^{-1}
      (\vct{y}_t - \bar{\vct{y}}_t)
    \big)
\end{equation}
In the common special case of additive isotropic Gaussian observation noise, i.e. $\eta = \nrm(\mathbf{0},\idmtx_{\intg{Y}})$ and $\opfunc{L}(\vct{z}) = \sigma_y(\vct{z})\,\idmtx_{\intg{Y}}$ with $\sigma_y : \set{Z} \to \reals_{>0}$, then we have
\begin{equation}
    \td{\bar{\pi}}{\bar{\rho}}([\vct{u};\vct{v}_0;\vct{v}_{\range{\intg{ST}}}]) \propto
    (\sigma_y\circ\,\vctfunc{g}_{\rvct{z}}(\vct{u}))^{-\intg{TY}}
    \exp\left(
      \frac{-\| [\vct{y}_{\range{\intg{T}}}] - [\vctfunc{g}_{\bar{\rvct{y}}_:}\mkern-2mu(\vct{u},\vct{v}_0, \vct{v}_{\range{\intg{ST}}})] \|^2}{2(\sigma_y\circ\,\vctfunc{g}_{\rvct{z}}(\vct{u}))^2}
    \right).
\end{equation}
As this posterior $\bar{\pi}$ has a density with respect to the Lebesgue measure which we can pointwise evaluate up to an unknown normalising constant, we can generate approximate samples from $\bar{\pi}$ using any of the large range of \ac{mcmc} methods which apply to target distributions of this form. Typically the density function will be differentiable and so \ac{hmc} methods offer a particularly convenient and efficient approach to inference.

For the numerical experiments in \cref{sec:numerical-experiments} in which \ac{hmc} is used as a baseline, we use a dynamic integration time \ac{hmc} implementation \citep{betancourt2017conceptual} with a dual-averaging algorithm to adapt the integrator step size \citep{hoffman2014no}. We consider two different options for the metric matrix representation $\mtx{M}$ (mass matrix): the identity matrix and a diagonal matrix with diagonal elements adaptively set in the warm-up chain iterations to online estimates of the precisions of $\bar{\rvct{q}}$ under the posterior using the windowed adaptation algorithm used for adapting $\mtx{M}$ in \emph{Stan} \citep{carpenter2017stan}. As the matrix-vector multiplies required if using a dense $\mtx{M}$ would incur a $\mathcal{O}(\intg{S}^2\intg{T}^2)$ cost in each integrator step we do not consider adapting $\mtx{M}$ based on an estimate of the full posterior precision matrix. All parameters for the algorithms used are summarized in \cref{app:algorithmic-parameters}.

\section{FitzHugh--Nagumo model details}
\label{app:fitzhugh-nagumo-model-details}

The stochastic FitzHugh--Nagumo model we use is defined by the \acp{sde}
\begin{equation}\label{eq:fhn-sde}
    \dr{\rvar{x}}_1 = \tfrac{1}{\epsilon}(\rvar{x}_1 - \rvar{x}_1^3 - \rvar{x}_2) \dr \tau,
    \quad
    \dr\rvar{x}_2 = (\gamma\rvar{x}_1 - \rvar{x}_2 + \beta)\dr \tau + \sigma \dr \rvar{w}.
\end{equation}
As in \citet{ditlevsen2019hypoelliptic} we use a strong-order 1.5 Taylor discretisation scheme corresponding for this model to a forward operator
\begin{equation}
  \begin{aligned}
  \vctfunc{f}_\delta(\vct{z},\vct{x},\vct{v}) =
  \vct{x} + \delta \vctfunc{a}(\vct{x},\vct{z})
  &+
  \tfrac{\delta^2}{2} \jacob_1\vctfunc{a}(\vct{x},\vct{z})\vctfunc{a}(\vct{x},\vct{z})
  \\
  &+
  \delta^{\frac{1}{2}} \opfunc{B}(\vct{x},\vct{z}) v_1
  +
  \tfrac{\delta^{\frac{3}{2}}}{2}\jacob_1\vctfunc{a}(\vct{x},\vct{z}) \opfunc{B}(\vct{x},\vct{z})(v_1 + \tfrac{v_2}{\sqrt{3}})
  \end{aligned}
\end{equation}
with a $\intg{V} = 2$ dimensional standard normal input vector $\vct{v}$. Note unlike the approach of \citet{ditlevsen2019hypoelliptic} our approach remains well-defined even when using a Euler--Maruyama discretisation of a hypoelliptic diffusion. We use the more accurate order 1.5 discretisation scheme here however as there is negligible additional computational cost or complexity compared to a Euler--Maruyama discretisation.

We use priors $\log\sigma \sim \nrm(0, 1)$, $\log\epsilon \sim \nrm(0, 1)$, $\log\gamma \sim \nrm(0, 1)$, $\beta \sim \nrm(0, 1)$ and $\rvct{x}_{0} \gvn \rvct{z} = [\sigma; \epsilon; \gamma; \beta] \sim \nrm([0; \beta], \idmtx_2)$. Under the transformation $\bar{\rvar{x}}_1 = \rvar{x}_1, \bar{\rvar{x}}_2 = \rvar{x}_2 - \beta$, $\bar{\epsilon}=\epsilon$, $\bar{\gamma} = \gamma$, $\bar{\sigma} = \sigma$, $\bar{s} = -\beta$ we have that an equivalent \ac{sde} to \eqref{eq:fhn-sde} in the transformed space is
\begin{equation}\label{eq:fhn-sde-alternative-form}
  \dr\bar{\rvar{x}}_1 = \tfrac{1}{\bar\epsilon}(\bar{\rvar{x}}_1 - \bar{\rvar{x}}_1^3 - \bar{\rvar{x}}_2 + \bar{s}) \dr \tau,
  \quad
  \dr\bar{\rvar{x}}_2 = (\bar\gamma\bar{\rvar{x}}_1 - \bar{\rvar{x}}_2)\dr \tau + \bar\sigma \dr \rvar{w}.
\end{equation}
This formulation corresponds to the \ac{sde} used in \citet{meulen2020bayesian}. For the comparison with the approach of \citet{meulen2020bayesian} in the noisy observation case in \cref{subsec:fhn-noisy-experiments}, we use a time discretisation of the \ac{sde} in \eqref{eq:fhn-sde-alternative-form} to run the experiments using the Julia code accompanying the article, with priors on the parameters $\log\bar\sigma \sim \nrm(0, 1)$, $\log\bar\epsilon \sim \nrm(0, 1)$, $\log\bar\gamma \sim \nrm(0, 1)$, $\bar{s} \sim \nrm(0, 1)$ and on the initial state $\bar{\rvct{x}}_{0} \sim \nrm(\vct{0}, \idmtx_2)$, giving an equivalent generative model in the transformed space to the model in the original space.
\section{Susceptible-infected-recovered model details}
\label{app:sir-model-details}

The system of \acp{sde} for the \ac{sir} in terms of the log-transformed state $\rvct{x} = [\log\rvar{s}; \log\rvar{i};\log\rvar{c}]$ can be derived from the \ac{sde} system in the original space in \eqref{eq:sir-sde} using It\^o's lemma as
\begin{multline*}\label{eq:sir-sde-transformed}
  \begin{bmatrix}
      \dr \rvar{x}_1 \\
      \dr \rvar{x}_2 \\
      \dr \rvar{x}_3 \\
  \end{bmatrix} =
  \begin{bmatrix}
    -\frac{1}{N}\exp(\rvar{x}_2 + \rvar{x}_3) - \frac{1}{2N}\exp(\rvar{x}_2 + \rvar{x}_3 - \rvar{x}_1) \\
    \frac{1}{N}\exp(\rvar{x}_1 + \rvar{x}_3) - \frac{1}{2N}\exp(\rvar{x}_1 + \rvar{x}_3 - \rvar{x}_2) - \frac{\gamma}{2}\exp(-\rvar{x}_2) - \gamma \\
    \alpha(\beta - \rvar{x}_3)
  \end{bmatrix}\dr\tau \\+
  \begin{bmatrix}
    N^{-\frac{1}{2}}\exp(\frac{\rvar{x}_2 + \rvar{x}_3 - \rvar{x}_1}{2}) & 0 & 0\\
    -N^{-\frac{1}{2}}\exp(\frac{\rvar{x}_1 + \rvar{x}_3 - \rvar{x}_2}{2}) & \gamma^{\frac{1}{2}}\exp(-\frac{\rvar{x}_2}{2}) & 0 \\
    0 & 0 & \sigma
  \end{bmatrix}
  \begin{bmatrix}
    \dr \rvar{w}_1 \\
    \dr \rvar{w}_2 \\
    \dr \rvar{w}_3 \\
\end{bmatrix}.
\end{multline*}
An Euler--Maruyama scheme was used for the time-discretisation with $\intg{S} = 20$ steps per inter-observation interval. The population size is fixed to the known value $N = 763$ and the initial values $\rvar{s}(0)$ and $\rvar{i}(0)$ to 762 and 1 respectively. The parameters are given priors $\log\gamma \sim \nrm(0,1)$, $\log\alpha \sim \nrm(0, 1)$, $\beta \sim \nrm(0, 1)$, $\log \sigma \sim \nrm(-3, 1)$ and $\log \sigma_y \sim \nrm(0, 1)$ and the initial value $\rvar{x}_3(0) = \log \rvar{c}(0) \sim \nrm(0, 1)$.

The observed data $\vct{y}$ were taken from \citet{bmj1978influenza}, with the number of cases \emph{confined to bed} manually digitized using a tool \emph{WebPlotDigitizer} \citep{rohatgi2020webplotdigitizer} from the points on the graph in the article. The specific values used were\\
\begin{tabular}{rcccccccccccccc}
  Day $\intg{t}$ & 0 & 1 & 2 & 3 & 4 & 5 & 6 & 7 & 8 & 9 & 10 & 11 & 12 & 13 \\
  Obs. $y_{\intg{t}}$ & 3 & 8 & 28 & 75 & 221 & 281 & 255 & 235 & 190 & 125 & 70 & 28 & 12 & 5
\end{tabular}
\section{Chain initialisation}
\label{app:chain-initialisation}

To initialise a constrained \acs{hmc} chain targeting the posterior distribution with density in \cref*{eq:manifold-density} supported on the manifold $\set{M}$, we need to find an initial $\vct{q} \in \set{M}$, which in practice we relax to the condition $\Vert\vctfunc{c}(\vct{q})\Vert_\infty < \theta_c$. In our experiments to find one or more initial points satisfying this condition, we use one of the following heuristics, depending on whether observation noise is present or not in the model.

\vspace{-10pt}
\subsection{Additive observation noise}
\label{app:chain-initialisation-noisy-obs}

In the case of additive observation noise with $\vctfunc{h}(\vct{x},\vct{z},\vct{w}) = \vcthbar(\vct{x}) + \opfunc{L}(\vct{z})\vct{w}$, if we further assume that $\intg{W}  = \intg{Y}$ and $\opfunc{L}$ is non-singular $\mu$-almost surely then we can almost surely find $\vct{w}_{\range{\intg{T}}}$ as a function of $\vct{u}$, $\vct{v}_0$ and $\vct{v}_{\range{\intg{ST}}}$ such that
$$
  \vctfunc{c}([\vct{u}; \vct{v}_0; \vct{v}_{\range{\intg{ST}}}; \vct{w}_{\range{\intg{T}}}]) = \vct{0}
  \iff
  [\vctfunc{g}_{\rvct{y}_{:}}\mkern-2mu(\vct{u}, \vct{v}_0, \vct{v}_{\range[1]{\intg{S}\intg{T}}}, \vct{w}_{\range{\intg{T}}})]
  = [\vct{y}_{\range{\intg{T}}}],
$$
by computing $\vct{x}_{\range{\intg{ST}}} = \vctfunc{g}_{\rvct{x}_{:}}\mkern-1mu(\vct{u},\vct{v}_0,\vct{v}_{\range[1]{\intg{St}}})$
where
\vspace{3pt}
\begin{algorithmic}
  \Function{$\vctfunc{g}_{\rvct{x}_{:}}\mkern-1mu$}
    {$\vct{u},\vct{v}_0,\vct{v}_{\range[1]{\intg{St}}}$}
  \State $\vct{z} \gets \vctfunc{g}_{\rvct{z}}(\vct{u})$
  \State $\vct{x}_0 \gets \vctfunc{g}_{\rvct{x}_0}\mkern-2mu(\vct{z},\vct{v}_0)$
  \For{$\intg{s} \in \range[1]{\intg{St}}$}
    \State $\vct{x}_{\intg{s}} \gets \vctfunc{f}_\delta(\vct{z},\vct{x}_{\intg{s}-1},\vct{v}_s)$
  \EndFor
  \vspace{-3pt}
  \State \Return $\vct{x}_{\range{\intg{ST}}}$
  \EndFunction
\end{algorithmic}
and setting $\vct{w}_{\intg{t}} = \opfunc{L}(\vct{z})^{-1}(\vct{y}_{\intg{t}} - \vcthbar(\vct{x}_{\intg{St}}))$ for each $\intg{t} \in \range{\intg{T}}$. We can therefore, for example, sample $\vct{u} \sim \tilde{\mu}$, $\vct{v}_0 \sim \tilde{\nu}$ and $\vct{v}_{\intg{s}} \sim \nrm(\vct{0},\idmtx_{\intg{V}})~\forall \intg{s}\in\range{\intg{ST}}$, and compute $\vct{w}_{\range{\intg{T}}}$ as just described to find a random initial point $\vct{q} = [\vct{u}; \vct{v}_0; \vct{v}_{\range{\intg{ST}}}; \vct{w}_{\range{\intg{T}}}] \in \set{M}$.

While this procedure is guaranteed to find a point on the manifold $\set{M}$, sampling $\vct{u}$, $\vct{v}_0$ and $\vct{v}_{\range{\intg{ST}}}$ from their priors will tend to produce points that are atypical under the posterior distribution. Often the geometry of the manifold at such points will be more challenging for constrained \ac{hmc} chains to navigate, with for example regions of high curvature which require a small integrator step size $t$ to control the proportion of trajectories which are terminated due to the Newton iteration failing to converge. This can lead to excessively long computation times for the early chain iterations in the adaptive warm-up phase where the step size is tuned to control an acceptance rate statistic, as the small step size combined with the dynamic \ac{hmc} implementation which expands the simulated trajectory until a termination condition is satisfied results in very long trajectories of many steps. %

Rather than directly initialising the chains at a point computed from a prior sample, we found that first running a number of steps of a gradient-descent optimisation algorithm on the negative logarithm of the Lebesgue density of the posterior distribution defined in \cref{eq:standard-hmc-posterior-lebesgue-density} from the prior sample and using the optimised $\bar{\vct{q}} = [\vct{u}; \vct{v}_0; \vct{v}_{\range{\intg{ST}}}; \vct{w}_{\range{\intg{T}}}]$ to compute $\vct{w}_{\range{\intg{T}}}$, and so $\vct{q} \in \set{M}$, as described above, led to more robust performance with lower variance in the times taken to complete the adaptive warm-up phase. As there is no need to perform iterative Newton solves during such an optimisation, the computational run time overhead of this optimisation was negligible compared to the run times of the warm-up and main chain phases. In our experiments we used the adaptive moments `Adam' algorithm \citep{kingma2015adam} to perform the optimisation, continuing the optimisation until the condition $\intg{T}^{-1}\sum_{\intg{t}=1}^\intg{T}\opfunc{L}(\vct{z})^{-1}(\vct{y}_{\intg{t}} - \vcthbar(\vct{x}_{\intg{St}})) < 1$ is met. The same procedure was also used to generate the initial states for the standard \ac{hmc} chains.

\vspace{-10pt}
\subsection{Noiseless observations}
\label{app:chain-initialisation-noiseless-obs}

In the case of noiseless observations we need to use an alternative approach to find an initial point $\vct{q}$ on the manifold $\set{M}$.
We first find a sequence of $\intg{T}$ states $\tilde{\vct{x}}_{\range{\intg{T}}}$ which is consistent with the observations, that is, for $\intg{t}\in\range{\intg{T}}$ we have that $\Vert \vcthbar(\tilde{\vct{x}}_{\intg{t}}) - \vct{y}_{\intg{t}}\Vert_{\infty} < \theta_c$. For linear observation functions $\vcthbar(\vct{x}) = \mtx{H}\vct{x}$, given an initial arbitrary sequence $\tilde{\vct{\chi}}_{\range{\intg{T}}}\in\set{X}^{\intg{T}}$, then
\begin{equation}\label{eq:linear-observation-solve-for-states}
  \tilde{\vct{x}}_{\intg{t}} =
  \vct{\chi}_{\intg{t}} - \mtx{H}\tr(\mtx{H}\mtx{H}\tr)^{-1}(\mtx{H}\vct{\chi}_{\intg{t}} - \vct{y}_{\intg{t}})
\end{equation}
will satisfy $\vcthbar(\tilde{\vct{x}}_{\intg{t}}) = \vct{y}_{\intg{t}}$ (modulo floating point error) for each $\intg{t} \in \range{\intg{T}}$. A simple mechanism for generating the initial sequence $\tilde{\vct{\chi}}_{\range{\intg{T}}}$ is to use a prior sample, i.e. generate $\vct{u} \sim \tilde{\mu}$, $\vct{v}_0 \sim \tilde{\nu}$ and $\vct{v}_{\intg{s}} \sim \nrm(\vct{0},\idmtx_{\intg{V}})~\forall \intg{s}\in\range{\intg{ST}}$, compute $\vct{x}_{\range{\intg{ST}}} = \vctfunc{g}_{\rvct{x}_{:}}\mkern-1mu(\vct{u},\vct{v}_0,\vct{v}_{\range[1]{\intg{St}}})$ and then set $\tilde{\vct{x}}_{\intg{t}} = \vct{x}_{\intg{St}}~\forall \intg{t} \in \range{\intg{T}}$. For non-linear observation functions $\vcthbar$, we can follow an analogous procedure but replacing the $\intg{T}$ linear solves in \eqref{eq:linear-observation-solve-for-states} with an iterative method to solve $\intg{T}$ independent systems of non-linear equations; for example, one choice would be the Gauss--Newton like iteration
\begin{equation}\label{eq:nonlinear-observation-solve-for-states}
  \tilde{\vct{x}}_{\intg{t}}^{\intg{j}+1} =
  \tilde{\vct{x}}_{\intg{t}}^{\intg{j}} - \jacob\vcthbar(\tilde{\vct{x}}_{\intg{t}}^{\intg{j}})\tr( \jacob\vcthbar(\tilde{\vct{x}}_{\intg{t}}^{\intg{j}}) \jacob\vcthbar(\tilde{\vct{x}}_{\intg{t}}^{\intg{j}})\tr)^{-1}( \jacob\vcthbar(\tilde{\vct{x}}_{\intg{t}}^{\intg{j}})\tilde{\vct{x}}_{\intg{t}}^{\intg{j}} - \vct{y}_{\intg{t}}),
  \quad
  \tilde{\vct{x}}_{\intg{t}}^{0} = \tilde{\chi}_{\intg{t}}.
\end{equation}
As the number of equations $\intg{Y}$ and variables $\intg{X}$ in each system will be relatively small for most diffusion models, solving the systems will not usually be overly burdensome.

Once we have a state sequence $\tilde{\vct{x}}_{\range{\intg{T}}}$ consistent with the observations, we can then use either a direct approach to solve for a corresponding constraint satisfying state $\vct{q}$ if the forward operator $f_\delta$ meets a linearity condition, or in the more general case we can use an iterative approach. Both methods are described below.

In many settings the forward operator $\vctfunc{f}_\delta$ is linear as a function of its third argument (state noise vector $\vct{v}$) with the two other arguments fixed, and has a Jacobian with respect to this argument, $\jacob_3 \vctfunc{f}_\delta$, which is full row-rank. This applies for example to the common case of using an Euler-Maruyama scheme to discretise an elliptic diffusion. It also applies to the strong-order 1.5 Taylor discretisation of the hypoelliptic FitzHugh--Nagumo model described in \cref{app:fitzhugh-nagumo-model-details}.

In these cases, the forward operator $f_\delta$ can be decomposed as
\begin{equation}
  \vctfunc{f}_\delta(\vct{z}, \vct{x}, \vct{v}) =
  \vctfunc{m}_\delta(\vct{z}, \vct{x}) + \mtx{S}_\delta(\vct{z}, \vct{x}) \vct{v}
\end{equation}
with $\vctfunc{m}_\delta : \set{Z} \times \set{X} \to \set{X}$ and $\mtx{S}_\delta : \set{Z} \times \set{X} \to \lbrace \mtx{A} \in \mathbb{R}^{\intg{X}\times\intg{V}} : \mathrm{rank}(\mtx{A}) = \intg{X} \rbrace$.

If this property is satisfied, then given values for the unconstrained vectors $\vct{u}$ and $\vct{v}_0$ (for example sampled from their priors) and a state sequence $\tilde{\vct{x}}_{\range{\intg{T}}}$, we can solve for a sequence of noise vectors $\vct{v}_{\range{\intg{ST}}}$ such that the full state sequence $\vct{x}_{\range{\intg{ST}}} = \vctfunc{g}_{\rvct{x}_{:}}\mkern-1mu(\vct{u},\vct{v}_0,\vct{v}_{\range[1]{\intg{ST}}})$ linearly interpolates $\tilde{\vct{x}}_{\range{\intg{T}}}$, that is $\vct{x}_{\intg{S}(\intg{t}-1) + \intg{s}} = \tilde{\vct{x}}_{\intg{t}-1} + \frac{\intg{s}}{\intg{S}}(\tilde{\vct{x}}_{\intg{t}} - \tilde{\vct{x}}_{\intg{t}-1})$ for $\intg{t} \in \range{\intg{T}}$, $\intg{s} \in \range{\intg{S}}$, as follows
\vspace{3pt}
\begin{algorithmic}
    \State $\vct{z} \gets \vctfunc{g}_{\rvct{z}}(\vct{u})$
    \State $\tilde{\vct{x}}_0 \gets \vct{x}_0 \gets \vctfunc{g}_{\rvct{x}_0}\mkern-2mu(\vct{z},\vct{v}_0)$
    \For{$\intg{t} \in \range[1]{\intg{T}}$}
      \For{$\intg{s} \in \range[1]{\intg{S}}$}
        \State $\vct{x}_{\intg{S}(\intg{t}-1)+\intg{s}} \gets \tilde{\vct{x}}_{\intg{t}-1} + \frac{\intg{s}}{\intg{S}}(\tilde{\vct{x}}_{\intg{t}} - \tilde{\vct{x}}_{\intg{t}-1})$
        \State $\vct{v}_{\intg{S}(\intg{t}-1) + \intg{s}} \gets \mtx{S}_\delta(\vct{z}, \vct{x}_{\intg{S}(\intg{t}-1)+\intg{s}-1}) \backslash (\vct{x}_{\intg{S}(\intg{t}-1)+\intg{s}} - \vctfunc{m}_\delta(\vct{z}, \vct{x}_{\intg{S}(\intg{t}-1)+\intg{s}-1}))$
      \EndFor
    \EndFor
\end{algorithmic}%
where $\mtx{A}\backslash\vct{b}$ returns $\vct{x}$ such that $\mtx{A}\vct{x} = \vct{b}$. If the state sequence $\tilde{\vct{x}}_{\range{\intg{T}}}$ is chosen such that $\vct{y}_{\intg{t}} = \vcthbar(\tilde{\vct{x}}_{\intg{t}})$ for all $\intg{t} \in \range{\intg{T}}$ as detailed above, then by construction the interpolated state sequence $\vct{x}_{\range{\intg{ST}}}$ will also be consistent with the observations, and $\vct{q} = [\vct{u};\vct{v}_0;\vct{v}_{\range{\intg{ST}}}]$ will be on the manifold $\set{M}$.

In more general settings, we can independently sample a point $\vct{q}$ from the prior $\rho$ on the ambient space (or reuse the existing prior sample $[\vct{u}; \vct{v}_0; \vct{v}_{\range{\intg{ST}}}]$ used to generate $\tilde{\vct{\chi}}_{\range{\intg{T}}}$) and run an adaptive gradient-descent algorithm \emph{Adam} \citep{kingma2015adam}, to minimise the following objective function $\gamma : \reals^{\intg{Q}} \times \set{X}^{\intg{T}} \to \reals$
\vspace{3pt}
\begin{algorithmic}
  \Function{$\gamma$}
      {$[\vct{u};\vct{v}_0;\vct{v}_{\range[1]{\intg{ST}}}]$, $\tilde{\vct{x}}_{\range{\intg{T}}}$}
    \State $\vct{z} \gets \vctfunc{g}_{\rvct{z}}(\vct{u})$
    \State $\vct{x}_0 \gets \vctfunc{g}_{\rvct{x}_0}\mkern-2mu(\vct{z},\vct{v}_0)$
    \For{$\intg{s} \in \range[1]{\intg{ST}}$}
      \State $\vct{x}_{\intg{s}} \gets \vctfunc{f}_\delta(\vct{z},\vct{x}_{\intg{s}-1},\vct{v}_s)$
    \EndFor
    \State \Return $\frac{1}{\intg{T}\intg{X}}\sum_{\intg{t=1}}^{\intg{T}} \Vert \vct{x}_{\intg{S}\intg{t}} - \tilde{\vct{x}}_{\intg{t}} \Vert^2_2$
  \EndFunction
\end{algorithmic}%
with respect to its first argument, initialised at $\vct{q}$ and with the second argument fixed at the computed $\tilde{\vct{x}}_{\range{\intg{T}}}$ value. The optimisation is continued until $\gamma(\vct{q},\tilde{\vct{x}}_{\range{\intg{T}}}) < \theta_\gamma$ with the optimisation restarted from a new $\vct{q} \sim \rho$ if this is not satisfied within $\intg{M}_g$ iterations (we used $\theta_\gamma = 10^{-6}$ and $\intg{M}_g = 1000$). We then run the Newton iteration in \eqref{eq:newton-iteration} (with $\vct{p} = \vct{0}$) to project $\vct{q}$ on to the manifold, within the tolerance $\Vert\vctfunc{c}(\vct{q})\Vert_\infty < \theta_c$.

We found this combination of gradient descent to find a point close to the manifold then Newton iteration to project to within the specified tolerance $\theta_c$ was more effective than either solely using gradient descent until within the constraint tolerance (with the gradient-descent iteration tending to converge slowly once close to the manifold) or using the iteration in \eqref{eq:newton-iteration} directly on $\vct{q}$ sampled from the prior, as for points far from the manifold the Newton iteration generally fails to converge.

It is also possible to use gradient-descent optimisation directly on the norm of the (non-conditioned) constraint function, i.e. $\gamma(\vct{q}) := \frac{1}{\intg{C}} \Vert \vctfunc{c}(\vct{q}) \Vert^2_2$, which sidesteps the requirement to find a constraint-satisfying state sequence $\tilde{\vct{x}}_{\range{\intg{T}}}$. However we found that as the number of observations times $\intg{T}$ becomes large this approach begins to suffer from the optimisation getting stuck in local minima, with the conditional independencies introduced by instead fixing the values of the states at the observation times seeming to make the optimisation problem simpler to solve.
\section{Measuring chain computation times}
\label{app:chain-computation-times}

To compute the chain computation times in the numerical experiments in \cref*{sec:numerical-experiments}, we recorded the number of evaluations in each chain of the key expensive operations in \cref*{alg:constrained-hmc} and multiplied these by estimated compute times for each operation calculated by separately timing the execution of a compiled loop iterating the operation a large number of times in a single threaded-process. Compared to directly using the wall-clock run times for the chain this approach eliminates the effect of the Python interpreter overhead in the implementation in the computation time estimates, removes the variability in run time estimates due to the effect of other background processes and allowed experiments to be run on multiple machines with differing hardware while remaining comparable. The operations monitored were: evaluations of the constraint function; evaluations of the constraint Jacobian; matrix decompositions to solve linear systems in the Gram matrix and evaluation of the log-determinant of the Gram matrix; evaluation of the gradient of the log-determinant of the Gram matrix.
\section{Summary of algorithmic parameters}
\label{app:algorithmic-parameters}

In this section we summarize the values of all algorithmic parameters required for reproducing the numerical experiments. Most of these have been provided elsewhere in the text (or in cited works), but we collate them together here for ease of reference.

\subsection*{Chain initialisation}

For the FitzHugh--Nagumo model with noiseless observations, the linear interpolation scheme described in \cref{app:chain-initialisation-noiseless-obs} was used to initialise the chain states for each experiment. The initial state sequences $\tilde{\vct{x}}_{\range{\intg{T}}}$ consistent with observations  were generated as $\tilde{\vct{x}}_{\intg{t}} = [y_{\intg{t}}, 0.5 r_{\intg{t}}], ~ r_{\intg{t}} \sim \mathcal{N}(0, 1)$ for all $\intg{t} \in \range{\intg{T}}$. For the FitzHugh--Nagumo model with noisy observations, the same scheme was used to generate the initial states with the observation noise vector components $\vct{w}_{\range{\intg{T}}}$ all set to zero. The same scheme was use for initialising both the constrained and standard \ac{hmc} chains, with the standard \ac{hmc} chain states simply not including the (zeroed) observation noise vector components $\vct{w}_{\range{\intg{T}}}$.

For the \ac{sir} model the scheme described in \cref{app:chain-initialisation-noisy-obs} was used to initialise the chain states for each experiment. The adaptive moments (\emph{Adam} \citep{kingma2015adam}) optimizer used to minimize the negative log posterior density on $\vct{u}$, $\vct{v}_0$ and $\vct{v}_{\range{\intg{ST}}}$, had parameters (following notation from Algorithm 1 in the original paper)
\begin{center}
\begin{tabular}{rl}
  Step size $\alpha$ & 0.1 \\
  Exponential decay rate for first moment estimate $\beta_1$ & 0.8 \\
  Exponential decay rate for first moment estimate $\beta_2$  & 0.999 \\
  Additive term for numerical stability $\epsilon$ & $1 \times 10^{-8}$
\end{tabular}
\end{center}
The optimization was terminated when $\frac{1}{T\sigma_y^2}\sum_{\intg{t}=1}^{\intg{T}} \left|\vct{y}_{\intg{t}} - \vcthbar(\vct{x}_{\intg{St}}) \right|^2 < 1$ where $\sigma_y$ and $\vct{x}_{\range{\intg{ST}}}$ correspond to the values generated from the current $\vct{u}$, $\vct{v}_0$ and $\vct{v}_{\range{\intg{ST}}}$ values. If this criterion was not met in 5000 iterations the optimization was restarted from a new point sampled from the prior. For the constrained \ac{hmc} chains the observation noise vector components of the initial state were set to $\vct{w}_{\intg{t}} = \frac{1}{\sigma^y} (\vct{y}_{\intg{t}} - \vcthbar(\vct{x}_{\intg{St}}))$ for all $\intg{t} \in \range{\intg{T}}$ from the $\sigma_y$ and $\vct{x}_{\range{\intg{ST}}}$ values corresponding to the optimized values of $\vct{u}$, $\vct{v}_0$ and $\vct{v}_{\range{\intg{ST}}}$, with this giving a constraint satisfying initial state $\vct{q} = [\vct{u}; \vct{v}_0; \vct{v}_{\range{\intg{ST}}}; \vct{w}_{\range{\intg{T}}}]$. For the standard \ac{hmc} chains the same scheme was used, other than the $\vct{w}_{\range{\intg{T}}}$ components of the state not being included, with initial states $\vct{q} = [\vct{u}; \vct{v}_0; \vct{v}_{\range{\intg{ST}}}]$

\subsection*{Reversible constrained leapfrog integrator}

For the constrained \ac{hmc} chains, the constrained leapfrog integrator described in \cref{alg:constrained-hmc} in the main paper was used to simulate the constrained Hamiltonian dynamics trajectories. The following parameters were used in all cases
\begin{center}
\begin{tabular}{rl}
Projection solver constraint tolerance $\theta_c$ & $10^{-9}$ \\
Projection solver position tolerance $\theta_q$ & $10^{-8}$ \\
Projection solver maximum iterations $\intg{J}$ & 50 \\
Reversibility check tolerance & $2 \times 10^{-8}$
\end{tabular}
\end{center}

\subsection*{Integrator step size adaptation}

Both constrained and standard \ac{hmc} chains used the dual-averaging algorithm described in Section 3.2 of \citet{hoffman2014no} to tune the integrator step size during the adaptive warm-up phase of the chains. The initial step size $\epsilon_0$ was set using the heuristic described in Algorithm 4 in \citet{hoffman2014no}. The dual-averaging algorithm was used within a dynamic \ac{hmc} algorithm which automatically selected the number of integrator steps per iteration (see next subsection for more details), with the algorithmic parameters used (corresponding to the notation in Algorithm 6 in \citet{hoffman2014no})
\begin{center}
\begin{tabular}{rl}
Target acceptance rate $\delta$ & 0.8 \\
Regularisation scale $\gamma$ & 0.1 \\
Relaxation exponent $\kappa$ & 0.75 \\
Iteration offset $t_0$ & 10 \\
Regularization target $\mu$ & $\log 10 \epsilon_0$
\end{tabular}
\end{center}
The values for the relaxation exponent $\kappa$, iteration offset $t_0$ and regularization target $\mu$ follow the suggested defaults recommended in \cite{hoffman2014no}. The target acceptance rate value $\delta = 0.8$ was based on the default value used in the current \emph{Stan} implementation of the dual-averaging step size tuning algorithm \citep{carpenter2017stan}. A larger regularisation scale value $\gamma = 0.1$ was used compared to the default of $\gamma = 0.05$ recommended in \citet{hoffman2014no} and used as the default in \emph{Stan}. Larger values of this parameter give stronger regularisation of the integrator step size towards the (usually relatively large) value $\exp(\mu)$. We found that a common problem with both constrained and standard \ac{hmc} chains in this setting, was that there was a tendency for there to be a large proportions of rejections in the initial few iterations of the warm-up phase. This tended to cause the step size to be adapted to very small values initially, which in turn led to many integrator steps being taken per iteration, and so very long computation times for these initial iterations. Increasing $\gamma$ so that the regularization has a stronger effect in these initial iterations reduced this behaviour while still leading to final adapted integrator step sizes that gave acceptance rates close to the target value.

In all cases the acceptance rate statistic used to tune the step size was the mean of the Metropolis accept probability of a move from the initial state the transition started at, to each of the states on the generated trajectory, with the acceptance rate statistic set to zero if the trajectory is terminated due to a convergence error in the integrator or if a non-reversible step was detected.

\subsection*{Metric (mass matrix) adaptation}

For the standard \ac{hmc} chains, as well as the integrator step size a diagonal metric matrix representation (mass matrix) was also adaptively tuned during the warm-up phase of the chains; the constrained \ac{hmc} chains used a fixed identity metric in all cases. Directly following the implementation used in \emph{Stan} \citep{carpenter2017stan}, the 500 iterations in the warm-up phases of the chains were split in to a number of sub-intervals. In an initial interval of 75 iterations only the step-size adaptation was active, with the metric initialised to the identity. A sequence of four growing intervals of 25, 50, 100 and 200 iterations were then used to tune the diagonal metric (with step size adaptation also active), with the marginal empirical posterior variances computed using the chain samples from each interval and the reciprocal of these values used to set the diagonal metric at the end of each interval (with the estimates starting afresh in each interval). In a final interval of 50 iterations only the step size was adapted with the metric fixed.

\subsection*{Dynamic Hamiltonian Monte Carlo algorithm}

The constrained and standard \ac{hmc} chains both used a dynamic \ac{hmc} implementation, that iteratively expands a binary tree corresponding to the current simulated trajectory by integrating forward and backward in time until a termination criterion is satisfied on the outermost states of the current sub-trees, or the overall tree depth reaches some predetermined maximum. The algorithm used is the same as that underlying the default \ac{hmc} implementation used in \emph{Stan} as of version 2.21. This is based on the \ac{nuts} algorithm proposed by \citep{hoffman2014no}, with updates to use a more efficient multinomial scheme to sample states from the generated trajectory \citep{betancourt2017conceptual}, a generalized termination criterion \citep{betancourt2013generalizing} and checks of the termination criterion on additional sub-tree for improved robustness.

A maximum binary tree depth of 10 (giving a maximum of $2^{10} - 1 = 1023$ integrator steps per trajectory) was used for the constrained \ac{hmc} chains, and a maximum binary tree depth of 20 (giving a maximum of $2^{20} - 1 = 1\,048\,575$ steps per trajectory) was used for the standard \ac{hmc} chains. The larger value used for the tree depth for the standard \ac{hmc} chains was motivated by the greater range of integrator step sizes required for standard \ac{hmc}, with smaller step sizes requiring a greater maximum tree depth to avoid repeatedly terminating the trajectory early because of the maximum tree depth being reached. For the constrained \ac{hmc} chains a maximum tree depth of 10 was already sufficiently large for this saturation to not occur in any of the experiments.

As in the original \ac{nuts} algorithm \citep{hoffman2014no}, if at any point during the trajectory expansion the Hamiltonian at the current state minus the Hamiltonian at the initial state exceeded a threshold $\Delta_{\mathrm{max}} = 1000$, with this indicative of the numerical integrator becoming unstable and diverging, the trajectory was terminated and the chain state set to the initial state at the beginning of the transition (that is a rejection occurs). Such early terminations and rejections were also triggered in the case of the constrained \ac{hmc} chains if a projection step failed to converge or a reversibility check indicated a step was non-reversible.
\newpage
\section{Additional SIR model results}
\vspace{-5pt}
\begin{figure}[!h]
  \centering
  \includegraphics[width=0.9\textwidth]{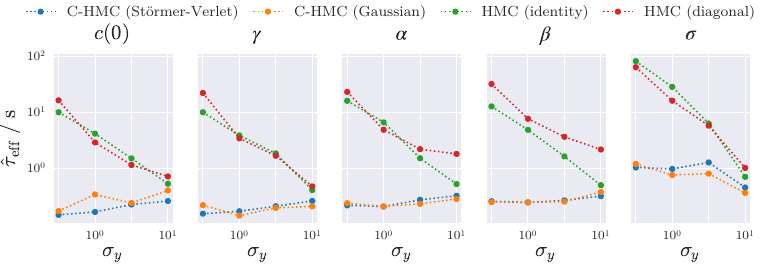}
  \caption{\emph{SIR model}: Computation time per effective sample $\hat{\tau}_{\mathrm{eff}}$ for different fixed $\sigma_y$ values for each model parameter, in all cases on a log-log scale. `\textsc{c-hmc} ($\circ$)' indicates constrained \ac{hmc} chains using an integrator based on the $\circ$ splitting, while `\textsc{hmc} ($\bullet$)' indicates standard \ac{hmc} chains with a $\bullet$ metric.
  }
  \label{fig:sir-noisy-computation-time-per-ess}
\end{figure}
\vspace{-10pt}
\begin{figure}[!h]
  \centering
  \includegraphics[width=0.82\textwidth]{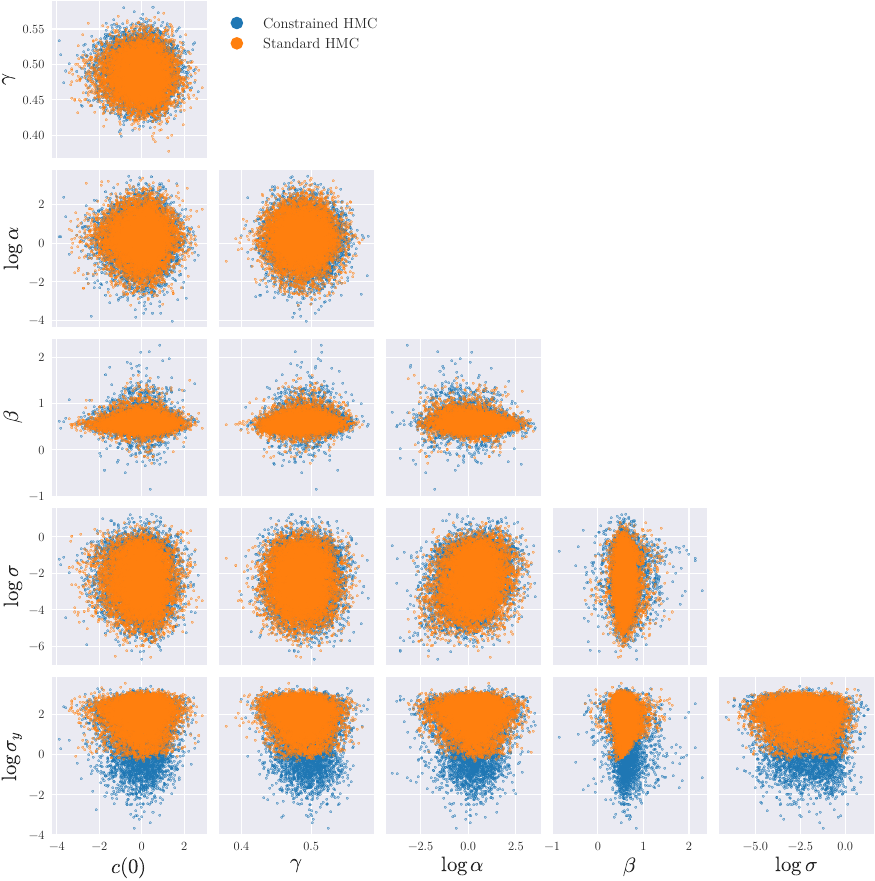}
  \caption{\emph{SIR model}: Posterior pairwise marginals (sample scatter plots) for each model parameter, for both constrained and standard \ac{hmc} chains.}
  \label{fig:sir-noisy-pair-plots}
\end{figure}

\bibliographystyle{plainnat}

\end{document}